\documentclass[10pt, letterpaper]{article}
\usepackage{graphics}
\usepackage{natbib}
\usepackage{color}
\usepackage{verbatim}
\usepackage{epstopdf}
\usepackage{epsfig}
\usepackage{algorithmicx}
\usepackage{algpseudocode}
\usepackage{algorithm}
\usepackage{latexsym}
\usepackage{longtable}
\usepackage{rotating}
\usepackage{multirow}
\usepackage{wrapfig}
\usepackage{blindtext}
\usepackage{textcomp}
\usepackage{float}
\usepackage{amsmath}
\usepackage{amssymb}
\usepackage{amsthm}
\usepackage{caption}

\usepackage{amsfonts}
\usepackage{dsfont}

\usepackage{times}
\usepackage{footnote}
\usepackage{lineno}
\usepackage{authblk}
\usepackage{url}
\usepackage{dcolumn}
\usepackage{array}
\usepackage{url}
\usepackage{threeparttable,url}
\usepackage{longtable}
\usepackage{txfonts}
\usepackage{epstopdf}
\usepackage{subfigure}
\usepackage{hyperref}
\usepackage{enumitem}
\usepackage{booktabs}
\usepackage{bm}
\usepackage{diagbox}
\usepackage{textcomp}

\usepackage[title]{appendix}
\usepackage[sc]{mathpazo}
\usepackage{amsmath,amsfonts,bm}

\def\1{\bm{1}}

\def\vzero{{\bm{0}}}

\def\vmu{{\bm{\mu}}}

\def\vc{{\bm{c}}}
\def\vd{{\bm{d}}}
\def\ve{{\bm{e}}}
\def\vf{{\bm{f}}}
\def\vg{{\bm{g}}}

\def\vq{{\bm{q}}}

\def\vu{{\bm{u}}}
\def\vv{{\bm{v}}}

\def\vx{{\bm{x}}}
\def\vy{{\bm{y}}}

\def\mG{{\bm{G}}}

\def\mQ{{\bm{Q}}}

\def\mV{{\bm{V}}}

\DeclareMathAlphabet{\mathsfit}{\encodingdefault}{\sfdefault}{m}{sl}
\SetMathAlphabet{\mathsfit}{bold}{\encodingdefault}{\sfdefault}{bx}{n}

\def\gA{{\mathcal{A}}}
\def\gB{{\mathcal{B}}}

\def\gD{{\mathcal{D}}}
\def\gE{{\mathcal{E}}}
\def\gF{{\mathcal{F}}}
\def\gG{{\mathcal{G}}}
\def\gH{{\mathcal{H}}}

\def\gL{{\mathcal{L}}}

\def\gO{{\mathcal{O}}}
\def\gP{{\mathcal{P}}}

\def\gR{{\mathcal{R}}}
\def\gS{{\mathcal{S}}}
\def\gT{{\mathcal{T}}}

\def\gW{{\mathcal{W}}}

\def\sR{{\mathbb{R}}}

\newtheorem{proposition}{Proposition}
\newtheorem{theorem}{Theorem}

\newtheorem{definition}{Definition}

\newtheorem{assumption}{Assumption}

\newtheorem{remark}{Remark}

\makeatletter
\newenvironment{breakablealgorithm}
  {%
   \begin{center}
     \refstepcounter{algorithm}%
     \hrule height.8pt depth0pt \kern2pt%
     \renewcommand{\caption}[2][\relax]{%
       {\raggedright\textbf{\ALG@name~\thealgorithm} ##2\par}%
       \ifx\relax##1\relax %
         \addcontentsline{loa}{algorithm}{\protect\numberline{\thealgorithm}##2}%
       \else %
         \addcontentsline{loa}{algorithm}{\protect\numberline{\thealgorithm}##1}%
       \fi
       \kern2pt\hrule\kern2pt
     }
  }{%
     \kern2pt\hrule\relax%
   \end{center}
  }
\makeatother

\title{\Large \textbf{Enforcing Priority in Schedule-based User Equilibrium\\ Transit Assignment}
}

\author[1, 2]{\small Liyang Feng}
\author[2]{\small Hanlin Sun}
\author[3]{\small Yu (Marco) Nie}
\author[1]{\small Jun Xie\textsuperscript{*}}
\author[2]{\small Jiayang Li\thanks{Corresponding authors.
E-mail: jun.xie@swjtu.edu.cn (Jun Xie); jiayangl@hku.hk (Jiayang Li).}}

\affil[1]{\footnotesize School of Transportation and Logistics, Southwest Jiaotong University, Chengdu, China}
\affil[2]{\footnotesize Department of Data and Systems Engineering, The University of Hong Kong, Hong Kong, China}
\affil[3]{\footnotesize Department of Civil and Environmental Engineering, Northwestern University, IL, USA}

\date{}
\begin{document}

\maketitle

\begin{abstract}

Denied boarding in congested transit systems induces queuing delays and departure-time shifts that can reshape passenger flows. Correctly modeling these responses in transit assignment hinges on the enforcement of two priority rules: continuance priority for onboard passengers and first-come-first-served (FCFS) boarding among waiting passengers. Existing schedule-based models typically enforce these rules through explicit dynamic loading and group-level expected costs, yet discrete vehicle runs can induce nontrivial within-group cost differences that undermine behavioral consistency. We revisit the implicit-priority framework of \citet{nguyen_modeling_2001} [``A Modeling Framework for Passenger Assignment on a Transport Network with Timetables." \textit{Transportation Science} 35(3): 238–249], which, by encoding boarding priority through the notion of available capacity, characterizes route and departure choices based on realized personal (rather than group-averaged) travel experiences.
However, the framework lacks an explicit mathematical formulation and exact computational methods for finding equilibria. Here, we derive an equivalent nonlinear complementarity problem (NCP) formulation and establish equilibrium existence under mild conditions. We also show that multiple equilibria may exist, including behaviorally questionable ones.  To rule out these artifacts, we propose a refined arc-level NCP formulation that not only corresponds to a tighter, behaviorally consistent equilibrium concept but also is more computationally tractable. We reformulate the NCP as a continuously differentiable mathematical program with equilibrium constraints (MPEC) and propose two solution algorithms. Numerical studies on benchmark instances and a Hong Kong case study demonstrate that the model reproduces continuance priority and FCFS queuing and captures departure-time shifts driven by the competition for boarding priority.
\end{abstract}

\section{Introduction}

In densely populated urban areas where mass transit is a primary travel mode, waiting passengers frequently fail to board incoming transit vehicles due to their limited capacities.  In Hong Kong, for example, buses traveling from Hong Kong Island to Kowloon during the evening peak often fill up at early stops, leaving passengers at intermediate stations to watch several buses pass by without space to board, a phenomenon locally referred to as ``ding jaa" \citep{ding_zha_term}. 
At Shahe Town, one of Beijing's largest suburban residential clusters, the crowding can be so severe that commuters wait up to 30 minutes at the metro station to enter a train \citep{beijingdaily2023shahe}. Mitigating such crowdedness is therefore a central objective for transit planners and operators. To support effective decisions, however, one needs to reliably predict how passengers choose routes and adjust their departure times in response to the delays caused by such a capacity crunch. This task is usually addressed using a transit assignment model.

\subsection{Challenges}

Transit assignment models can be broadly categorized as frequency-based and schedule-based. The first class represents each line by its service frequency rather than exact departure times and assumes that all passengers at a station experience the same expected waiting time, regardless of their actual arrival time \citep{spiess_optimal_1989,wu1994transit,cominetti2001common,cepeda2006frequency,xu2020hyperpath,xu2022hyperbush}. While these models may impose vehicle capacity as a constraint at the link level, they cannot represent the capacity-induced delays \emph{physically}, especially those incurred when a passenger misses one or more passing vehicles. 
Schedule-based models are better equipped to deal with the physics of the vehicle capacity restriction, simply because they typically represent individual vehicle runs according to posted schedules.  In these models, passengers may choose to depart from home according to the expected boarding time at the first transit stop, which may not coincide with the arrival time of the first vehicle (in other words, they explicitly take the queuing time into consideration). 
The enhanced realism, however, complicates the assignment problem.  A key sticky issue is the priority rules that determine who gets a seat first when vehicle capacity is insufficient. These include \emph{continuance priority}, under which passengers already on board must remain on the vehicle, and the \textit{first-come-first-served (FCFS) rule}, according to which waiting passengers board in order of arrival \citep{hamdouch_schedule-based_2008}.

A common approach to modeling priority is through dynamic network loading (DNL) \citep{ nuzzolo_doubly_2001, poon_dynamic_2004, papola2008schedule,hamdouch_schedule-based_2011, nuzzolo_schedule-based_2012,hamdouch_new_2014, cats_dynamic_2016, gentile2016modelling, yao2017simulation, cats_learning_2020}. In these models, passengers are grouped according to their chosen route and departure time, and a DNL procedure simulates the boarding of passengers at all stops \textit{explicitly} according to the priority rules. 
Because individuals within the same group may experience different actual costs due to vehicle capacity restrictions, a flow-weighted \emph{expected} travel cost is often calculated as the representative cost of that group.  The equilibrium is defined, accordingly, as a state where no group can reduce this expected travel cost by switching routes or departure times. 
However, in schedule-based transit models, service is inherently \emph{discrete} in time.  As a result, the cost difference within a group does not always diminish even as group size approaches zero, since tight residual capacity may still split a group across successive runs. This is more than a mathematical nuance: assuming passengers included in the same group behave identically is questionable even at the limit, because the experienced cost difference would be great enough to trigger behavioral deviation\footnote{In Appendix~\ref{sec:appendix dynamic} we illustrate this issue more concretely with a simple example.}.

An alternative approach, which we refer to as implicit prioritization, was pioneered by \citet{nguyen_modeling_2001}. Instead of simulating boarding events, their framework encodes priority \textit{implicitly} through the notion of available capacity. Specifically, 
passengers already on board have the highest priority (continuance priority), followed by waiting passengers in the order they arrived at a stop (FCFS). The available capacity of a boarding arc is then defined as the remaining vehicle capacity after loading all passengers from arcs with higher priority. If the available capacity on a given arc is non-positive, passengers arriving via that arc cannot board in that run and must wait for the next. Using this representation, \citet{nguyen_modeling_2001} defined a new equilibrium principle: a passenger can only switch to a route if all boarding arcs along that route have positive available capacity at equilibrium. In other words, switching is allowed only if doing so would not displace higher-priority passengers or overload any vehicle segment. %
In this way, some passengers successfully board while others may be left behind, and each person's incentives to adjust their decisions are separately modeled based on their \textit{actual} experience. 

Despite its conceptual elegance, the framework of \citet{nguyen_modeling_2001} is not fully developed. %
While it gives a well-defined equilibrium condition, a mathematical formulation amenable to analysis or computing remains elusive. %
No exact algorithm has been proposed to find an equilibrium satisfying the original priority constraints. Instead, \citet{nguyen_modeling_2001} solved an approximate problem that relaxes the priority constraints, leaving open the question of how to obtain an exact solution.

\subsection{Our contributions}

We propose an equivalent nonlinear complementarity problem (NCP) formulation for Nguyen et al.'s model (\citeyear{nguyen_modeling_2001}). Using this formulation, we prove the existence of an equilibrium state under mild conditions, a theoretical result that, to the best of our knowledge, has not been previously established. This result lays a foundation for applying the model in general transit assignment settings. We also show that multiple equilibria exist and, more importantly, some of these are behaviorally unrealistic. %

In light of this finding, we further propose a refined NCP formulation that is proven to admit only behaviorally consistent equilibrium solutions.
The new NCP formulation is much more computationally tractable:  because it enforces priority rules at the arc level rather than at the route level, route enumeration is obviated using commonly used column generation techniques. We use the Fischer–Burmeister function to reformulate the "hard" priority conditions as the zero set of a continuously differentiable function, which serves as the upper-level objective in a mathematical program with equilibrium constraints (MPEC) that remains equivalent to the original problem. This MPEC is then solved using two methods: a descent method based on implicit differentiation and a nonlinear programming–based method.

Finally, we demonstrate the practical relevance of the refined model through a series of numerical studies, including the benchmark presented in \citet{nguyen_modeling_2001} and a hypothetical schedule-based transit network created using the Sioux Falls network.  In particular, we conduct a real-world case study of the morning commute to the University of Hong Kong, where passengers choose between bus and metro services, and metro users may experience substantial queuing at station elevators. The model correctly reproduces FCFS queuing at elevators and reveals departure-early adjustments among bus passengers driven by competition for boarding priority, with the magnitude of such adjustments increasing under higher demand levels.

\subsection{Organization}

The remainder of this paper is organized as follows. Section \ref{sec:setting} sets up the problem and introduces existing schedule-based transit assignment models with priority. Section \ref{sec:revisit} revisits \citet{nguyen_modeling_2001}'s framework by presenting an equivalent NCP formulation, establishing the existence of its solutions, and examining the behavioral unrealism exhibited by some of them. Section \ref{sec:refined} introduces our refined model, discusses its analytical properties, and develops algorithms for its computation. Section \ref{sec:experiments} reports numerical experiments that validate the proposed model and algorithms, and Section \ref{sec:conclusion} concludes the paper.

\section{Nguyen et al.'s Model}
\label{sec:setting}

\subsection{Problem setting}
\textbf{Transit network.} Consider a transit system where $\gS$ and $\gL$ are the sets of transit stops and transit lines, respectively. For each $l \in \gL$, let $s_l^{i} \in \gS$ denote the $i$-th stop it visits ($i = 1, \ldots, n_l$), where $n_l$ is the total number of stops it serves. Each line $l$ has $m_l$ runs, and each run $j$ ($j = 1, \ldots, m_l$) follows a timetable specifying its arrival time $\tau_{l,j,i}^{\text{arr}}$ and departure time $\tau_{l,j,i}^{\text{dep}}$ at stop $s_l^{i}$ ($i = 1, \ldots, n_l$). All timetable values are defined on a discretized time axis $\gT$, with the exceptions $\tau_{l,j,1}^{\text{arr}} = -\infty$ for the arrival at the first stop and $\tau_{l,j,n_l}^{\text{dep}} = \infty$ for the departure at the last stop of each run. 

Passengers access and use this system as follows. Let $\gO$ and $\gD$ denote the sets of origins and destinations. Passengers can walk from their origin to any stop within a certain walking range; for each origin $o \in \gO$, the corresponding set of reachable stops is denoted by $\gS_o^{\text{walk}} \subseteq \gS$, and the walking time between $o$ and $s \in \gS_o^{\text{walk}}$ is $t_{o,s}^{\text{walk}}$. They then travel between stops along the transit lines. Finally, they alight at a stop within the walking range of their destination and walk to the destination; for each destination $d \in \gD$, the corresponding set of reachable stops is denoted by $\gS_d^{\text{walk}} \subseteq \gS$, with walking time $t_{d,s}^{\text{walk}}$ from $s \in \gS_d^{\text{walk}}$ to $d$. During their trips, passengers may transfer between lines at any stop served by multiple lines. However, for simplicity, we do not allow transfers that require walking to a different stop, although extending the model to include such transfers is straightforward.

As an example, Figure~\ref{fig:transit-net} shows a small transit network with four stops, A, B, C, and D, and two lines.
\begin{figure}[htbp]
    \centering
    \includegraphics[width=0.75\textwidth]{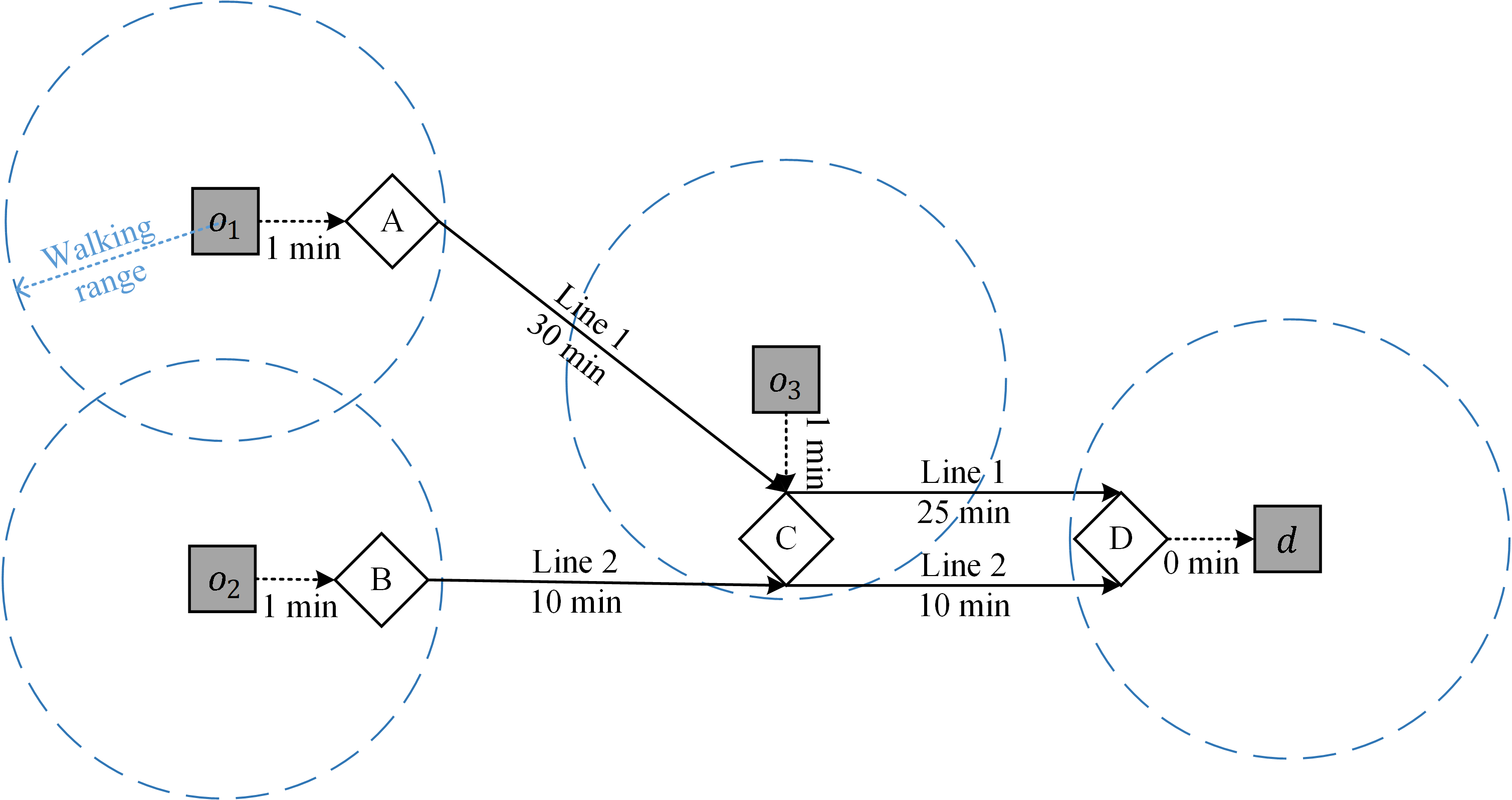}
    \caption{An example transit network.}
    \label{fig:transit-net}
\end{figure}
Line~1 is a regular service with a single run, and Line~2 is an express service with two runs. The timetable of all runs is given in Table~\ref{tab:timetable-of-transit-net}. There are three origins, $o_1$, $o_2$, and $o_3$, connected by walking to stops A, B, and C, respectively, and one destination $d$, connected by walking to stop D. The walking time between each origin and its associated stop is 1 minute, while the walking time between the destination and its associated stop is 0 minutes.
\begin{table}[htbp]
  \centering
  \caption{Timetable of the example network.}
  \resizebox{\textwidth}{!}{
  \footnotesize %
    \begin{tabular}{cccccccccc}
    \toprule
    \multirow{2}[4]{*}{Lines} & \multirow{2}[4]{*}{Runs} & \multicolumn{2}{c}{Stop A} & \multicolumn{2}{c}{Stop B} & \multicolumn{2}{c}{Stop C} & \multicolumn{2}{c}{Stop D} \\
\cmidrule{3-10}          &       & Arrival & Departure & Arrival & Departure & Arrival & Departure & Arrival & Departure \\
    \midrule
    1     & 1     & $-\infty$ & 7:25  & -     & -     & 7:55  & 7:55  & 8:20  & $\infty$ \\
    2     & 1     & -     & -     & $-\infty$ & 7:50  & 8:00  & 8:00  & 8:10  & $\infty$ \\
    2     & 2     & -     & -     & $-\infty$ & 8:10  & 8:20  & 8:20  & 8:30  & $\infty$ \\
    \bottomrule
    \end{tabular}%
    }
  \label{tab:timetable-of-transit-net}%
\end{table}%

For an OD pair $w = (o, d) \in \gW:= \gO \times \gD$, passengers may have different desired arrival times. We group them into a finite set of classes $\gB_w$, where each class $b \in \gB_w$ has a desired arrival-time window $[\tau_{w,b}^-, \tau_{w,b}^+]$ and a fixed demand $d_{w,b} \in \mathbb{R}_+$. Each passenger chooses both a departure time and a travel path. When all transit vehicles have sufficient capacity so that passengers can always board the first vehicle that arrives after they reach a stop, specifying a passenger's departure time together with their access, transfer, and egress stops and the sequence of transit lines fully determines their arrival time at the destination: any two passengers with the same such choices will experience identical arrival times.
However, under oversaturated conditions, some passengers may be denied boarding and must wait for later vehicles. In this case, the above coarse description is no longer sufficient to capture their actual experience. One must instead describe their detailed spatio–temporal events, including the times at which they board, transfer between vehicles, and alight.

\medskip
\noindent
\textbf{Event–activity graph.} To represent such events, we construct an event–activity graph $\gH(\gE, \gA)$, in which each node $E \in \gE$ corresponds to an event occurring at a specific time \citep{yin2025real}. These include: (1) $E_{o, t}^{\text{str}}$ (timestamp $t$), which represents the set-out event of passengers with origin $o \in \gO$ who begin their trips at time $t \in \gT$; (2) $E_{l,j,i}^{\text{dep}}$ (timestamp $\tau_{l,j,i}^{\text{dep}}$), the departure event of run $j = 1, \ldots, m_l$ of line $l \in \gL$ from stop $s_l^i$ for $i = 1, \ldots, n_l$; and (3) $E_{l,j,i}^{\text{arr}}$ (timestamp $\tau_{l,j,i}^{\text{arr}}$), the arrival event of run $j = 1, \ldots, m_l$ of line $l$ at stop $s_l^i$ for $i = 1, \ldots, n_l$. Additionally, we create two types of virtual nodes, (4) $E_{o}^{\text{ogn}}$ for each $o \in \gO$, and (5) $E_{d}^{\text{dst}}$ for each $d \in \gD$, to represent, respectively, the starting point of the passenger's trip before entering the transit system and the endpoint after completing the trip. To connect these event nodes, we introduce \emph{activity arcs}, each $A \in \gA \subseteq \gE \times \gE$ representing a feasible spatio–temporal movement undertaken by passengers, as follows.

\begin{itemize}[leftmargin=1.75em, topsep=6pt, itemsep=2pt, parsep=2pt]

    \item[(1)] \textbf{Access arcs}: $\gA^{\text{access}}$.  
    For each origin $o \in \gO$ and starting time $t \in \gT$, add an arc from $E_{o}^{\text{ogn}}$ to $E_{o,t}^{\text{str}}$, representing passengers at origin $o$ choosing to start their trip at time $t$.

    \item[(2)] \textbf{Boarding arcs}: $\gA^{\text{boarding}}$.  
    For each $o \in \gO$ and $t \in \gT$, and for each line $l \in \gL$ and stop index $i = 1,\ldots,n_l$,  
    if $s_l^i \in \gS_o^{\text{walk}}$, then add an arc from $E_{o,t}^{\text{str}}$ to $E_{l,j,i}^{\text{dep}}$ for each run $j = 1,\ldots,m_l$ satisfying
    $
        \tau_{l,j,i}^{\text{dep}} \ge t + t_{o,s_l^i}^{\text{walk}}.
    $
    These arcs represent passengers walking to stop $s_l^i$ and attempting to board any run that departs after they arrive at the stop.

    \item[(3)] \textbf{Dwelling arcs}: $\gA^{\text{dwelling}}$.  
    For each line $l \in \gL$ and run $j = 1,\ldots,m_l$, add an arc from $E_{l,j,i}^{\text{arr}}$ to $E_{l,j,i}^{\text{dep}}$ for each $i = 1,\ldots,n_l$, representing the period during which passengers remain on board while the vehicle dwells at stop $s_l^i$.

    \item[(4)] \textbf{Riding arcs}: $\gA^{\text{riding}}$.  
    For each line $l \in \gL$ and run $j = 1,\ldots,m_l$, add an arc from $E_{l,j,i}^{\text{dep}}$ to $E_{l,j,i+1}^{\text{arr}}$ for each $i = 1,\ldots,n_l-1$, representing passengers remaining on board as the vehicle travels from stop $s_l^i$ to stop $s_l^{i+1}$.

    \item[(5)] \textbf{Transfer arcs}: $\gA^{\text{transfer}}$. 
    For any two lines $l, l' \in \gL$ and stop indices $i = 1,\ldots,n_l$ and $i' = 1,\ldots,n_{l'}$,  
    if $s_l^i = s_{l'}^{i'}$, then for each run $j = 1,\ldots,m_l$ and $j' = 1,\ldots,m_{l'}$,
    if
    $
        \tau_{l,j,i}^{\text{arr}} \le \tau_{l',j',i'}^{\text{dep}},
    $
    add an arc from $E_{l,j,i}^{\text{arr}}$ to $E_{l',j',i'}^{\text{dep}}$.  
    This arc represents passengers alighting from line $l$ at stop $s_l^i$ and transferring to a run of line $l'$ that departs later from the same stop.

    \item[(6)] \textbf{Egress arcs}: $\gA^{\text{egress}}$.  
    For each destination $d \in \gD$, and each line $l \in \gL$ and stop index $i = 1,\ldots,n_l$,  
    if $s_l^i \in \gS_d^{\text{walk}}$, then add an arc from $E_{l,j,i}^{\text{arr}}$ to $E_{d}^{\text{dst}}$ for each run $j = 1,\ldots,m_l$.  
    This represents passengers alighting at stop $s_l^i$ and walking to destination $d$.
    
\end{itemize}
Figure~\ref{fig:EA-graph} illustrates the event-activity graph corresponding to the transit network in Figure~\ref{fig:transit-net} and the timetable in Table~\ref{tab:timetable-of-transit-net}. Circular nodes represent vehicle arrival and departure events, square nodes represent origins and destinations, and hexagonal nodes represent passenger starting times. For simplicity, we include one departure time for each of the origins $o_1$ and $o_2$ (7:24 and 7:53, respectively) and two departure times for origin $o_3$ (7:49 and 8:09).
\begin{figure}[ht]
    \centering
    \includegraphics[width=0.7\textwidth]{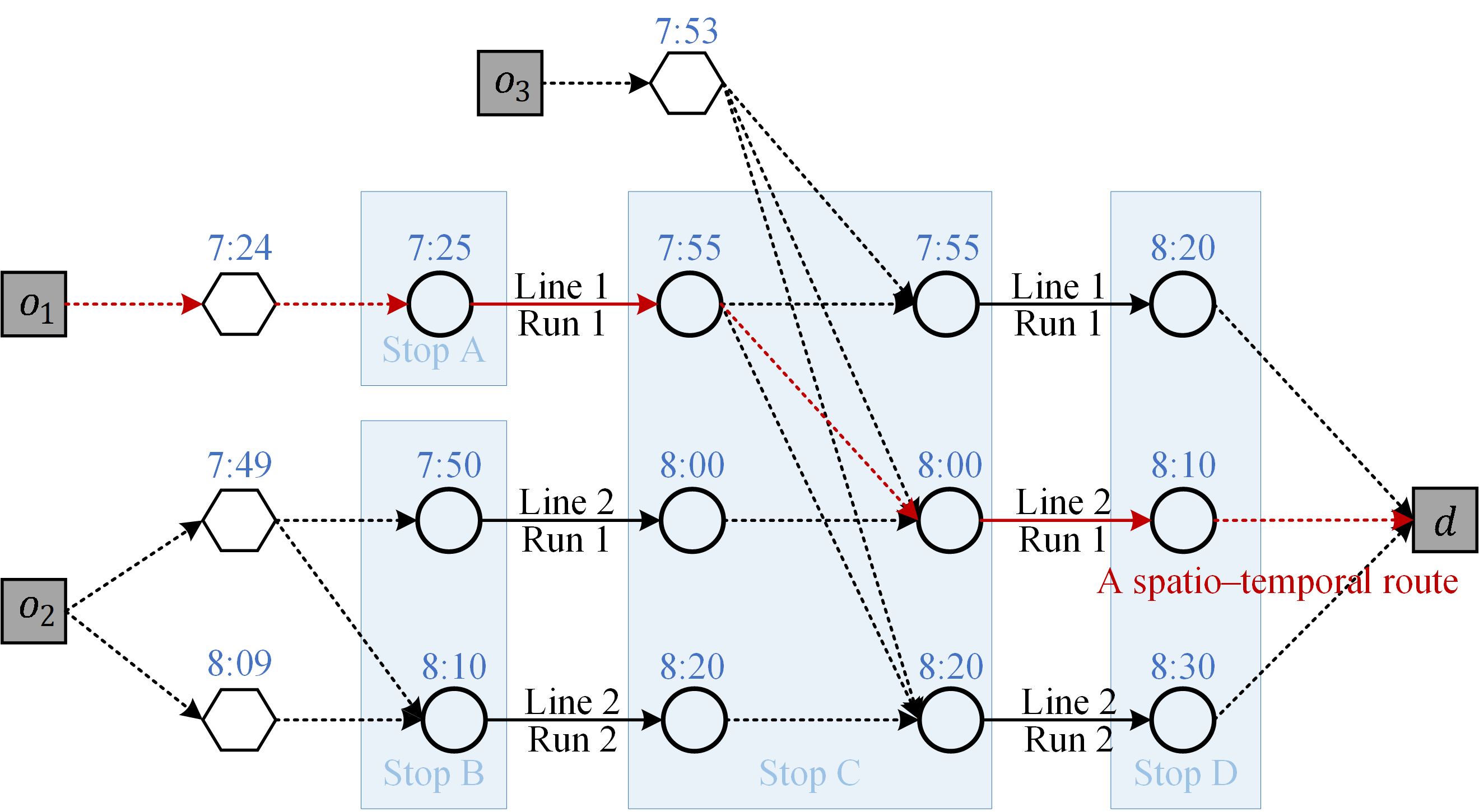}
    \caption{Event–activity graph of the example transit network.}
    \label{fig:EA-graph}
\end{figure}

\medskip
\noindent
\textbf{Spatio-temporal route.} Under the above construction, each path from $E_{o}^{\text{ogn}}$ to $E_{d}^{\text{dst}}$ represents a feasible \emph{spatio–temporal route} for passengers of the OD pair $w = (o,d) \in \gW$, for example, the red path highlighted in Figure~\ref{fig:EA-graph}. For each such OD pair $w$, we denote by $\gR_w$ the set of all routes between $E_{o}^{\text{ogn}}$ and $E_{d}^{\text{dst}}$. For each $w \in \gW$, class $b \in \gB_w$, and route $r \in \gR_w$, let $f_{w, b}^r$ be the number of passengers of class $b$ choosing route $r$, and collect all such variables into the vector $\vf$. The feasible set of $\vf$, denoted by $\gF$, consists of all $\vf$ satisfying
\begin{align}
    &\sum_{r \in \gR_w} f_{w, b}^r = d_{w, b}, \quad \forall w \in \gW, \ \forall b \in \gB_w, \label{cons:demand}\\
    &f_{w, b}^r\geq 0, \quad \forall w \in \gW, \ \forall b \in \gB_w, \ \forall r \in \gR_w, \label{cons:non-negative}\\
    &x_A\leq u_A, \quad \forall A \in \gA^{\text{riding}}, \label{cons:capacity}
\end{align}
where $x_A=\sum_{w\in \gW}\sum_{b\in \gB_w}\sum_{r\in \gR_w}f_{w, b}^r\delta_{w, b}^{r,A}$ denotes the total flow assigned to arc $A \in \gA$; $\delta_{w, b}^{r, A}=1$ if arc $A$ belongs to route $r$ and $\delta_{w, b}^{r, A}=0$ otherwise; $u_A\in \sR_+$ is the vehicle capacity of the line associated with $A$.

We do not impose a specific functional form for costs here; the subsequent analysis only requires mild regularity conditions, which we state when needed.
Let $c_{w,b}^r(\vf)$ denote the total travel cost experienced by passengers of OD pair $w\in \gW$, class $b\in \gB_w$, who follow route $r \in \gR_w$. This total cost typically includes components such as travel-time cost, monetary fare, on-board crowding disutility, and penalties for early or late arrival relative to the desired arrival-time window.

\subsection{User equilibrium with implicit priority}
\label{sec:implicit}
Under the implicit-priority framework, priority rules (e.g., continuance priority for on-board passengers and FCFS for boarding passengers) are represented through the available capacity defined for arcs that terminate at a vehicle departure event; these capacities determine which passengers can board the vehicle and which passengers must wait. We denote by $\gA^{\text{priority}}=\gA^{\text{boarding}}\cup \gA^{\text{dwelling}}\cup \gA^{\text{transfer}}$ the set of all such arcs that end at a vehicle departure event, namely the boarding, dwelling, and transfer arcs. For any departure event $E_{l,j,i}^{\text{dep}}$, all arcs ending at this event are equipped with a total order $\prec$ that specifies their loading priorities: the dwelling arc has the highest priority (continuance priority), and the remaining boarding and transfer arcs are ordered in ascending order of their passengers' arrival times at stop $s_{l}^{i}$ (reflecting the FCFS rule). Given this priority order, the available capacity on an arc $A\in \gA^{\text{priority}}$ is defined as
\begin{align}
    q_A(\vx)=u_{\text{riding}(A)}-\sum_{A'\in \text{Prior}(A)}x_{A'}, \notag
\end{align}
Here, $u_{\text{riding}(A)}$ denotes the vehicle capacity of the riding arc $\text{riding}(A)$ whose tail node coincides with the head node of $A$.
The set $\text{Prior}(A)=\{A'\in \gA^{\text{priority}}: \text{head}(A')=\text{head}(A), A'\preceq A\}$ collects all arcs with priority equal to or higher than that of $A$, where $\text{head}(A)$ is the head node of $A$. Thus, $q_A$ represents the residual vehicle capacity after loading passengers assigned to arc $A$ and all higher-priority arcs. In other words, it determines the maximum number of passengers that can be loaded on arcs of lower priority than $A$. For example, as illustrated in Figure~\ref{fig:implict}, the dwelling arc $(E_{2,1,2}^{\text{arr}}, E_{2,1,2}^{\text{dep}})$ has an available capacity of 3.
\begin{figure}[htbp]
    \centering
    \includegraphics[width=0.65\textwidth]{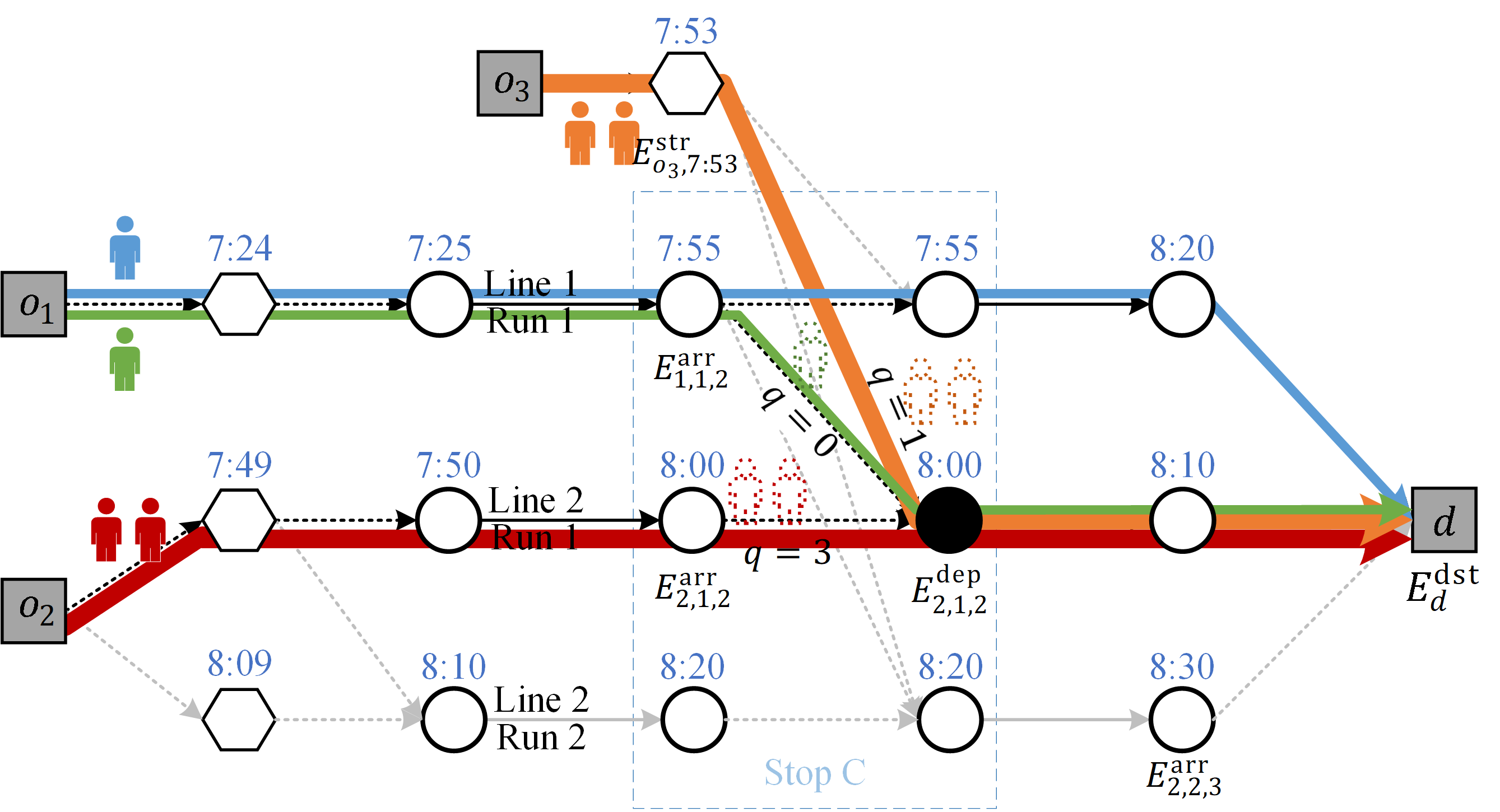}
    \caption{Illustration of the implicit priority and UEIP.}
    \label{fig:implict}
\end{figure}
Hence, at most 3 passengers in total can be loaded via the lower-priority boarding arc $(E_{o_3,7:53}^{\text{str}}, E_{2,1,2}^{\text{dep}})$ and transfer arc $(E_{1,1,2}^{\text{arr}}, E_{2,1,2}^{\text{dep}})$. Since the two passengers originating from $o_3$ arrive earlier at stop C (at 7:54, given a departure time of 7:53 and a 1-minute walking time), they are loaded first via the boarding arc. After they board, the available capacity associated with the boarding arc is reduced to 1, so at most one passenger can subsequently be loaded via the lower-priority transfer arc.

For each OD pair $w \in \gW$, class $b \in \gB_w$, and route $r \in \gR_w$, let $Q_{w,b}^r(\vf)=\min\{q_A(\vx): A\in \gA_{w,b,r}^{\text{priority}}\}$ be the route available capacity, where $\gA_{w,b,r}^{\text{priority}}$ is the set of boarding, dwelling, and transfer arcs belonging to this route. Then we can define the availability of a route and a user equilibrium with implicit priority as follows \citep{nguyen_modeling_2001}:
\begin{definition}[Route availability]
    Given any feasible flow vector $\vf \in \gF$, for each OD pair $w \in \gW$ and class $b \in \gB_w$, a route $r\in \gR_w$ is said to be available if every arc in $\gA_{w,b,r}^{\text{priority}}$ is available. Equivalently, this condition holds if $Q_{w,b}^r(\vf) > 0$.
\end{definition}

\begin{definition}[User equilibrium with implicit priority]
    A feasible flow vector $\vf\in \gF$ is a user equilibrium with implicit priority (UEIP) if no individual can reduce their travel cost $c_{w,b}^r(\vf)$ by unilaterally switching to another available route. Formally, let $\gR_{w,b,r}^{\text{dominate}}(\vf)=\{r'\in \gR_w: c_{w,b}^{r'}(\vf) < c_{w,b}^r(\vf)\}$ be the set of routes that dominate $r$ in terms of travel cost, and let the total available capacity of these dominant routes be $\Tilde{Q}_{w,b}^r(\vf)=\sum_{r'\in \gR_{w,b,r}^{\text{dominate}}(\vf)}Q_{w,b}^{r'}(\vf)$. Then, a $\vf\in \gF$ is a UEIP solution if 
    \begin{align}
        f_{w,b}^{r}=0 \quad \text{whenever } \Tilde{Q}_{w,b}^r(\vf)>0, \quad \forall w \in \gW, \ \forall b \in \gB_w, \ \forall r\in \gR_w. \label{equ:uep Nguyen}
    \end{align}
\end{definition}

\subsection{Illustrative example}
Here we illustrate the notion of UEIP using a simple example. For simplicity, we ignore crowding disutility and fares and assume that route travel cost consists only of constant in-vehicle travel time and late-arrival penalties. Passengers share a desired arrival window of [8:10, 8:20], so a lateness penalty of 10 minutes applies only on the egress arc $(E_{2,2,3}^{\text{arr}}, E_{d}^{\text{dst}})$, with no penalties on other arcs. The demand for each OD pair is 2, and the vehicle capacity is 5.  

Table~\ref{tab:UEIP} reports a UEIP solution $\vf^*$ and a non-UEIP solution $\vf$, and the corresponding flows are depicted in Figures~\ref{fig:implict} and~\ref{fig:implict-non-UEIP}, respectively. In this example, priority matters at stop~C, where it determines which passengers are allowed to board the express service \emph{line~2 run~1}. As discussed earlier, the specification of available capacities on the dwelling, boarding, and transfer arcs ensures that the UEIP solution $\vf^*$ is fully consistent with the priority rule. 
\begin{table}[htbp]
  \centering
  \caption{UEIP and non-UEIP results for the example transit network.}
  \footnotesize %
    \begin{tabular}{cclccccc}
    \toprule
    OD    & Routes & Description & Cost & $\vf^*$ & $Q(\vf^*)$ & $\vf$ & $Q(\vf)$ \\
    \midrule
    $(o_1,d)$ & $r_1$  & 7:24 - Line 1 Run 1 & 56    & 1     & 3     & 0     & 3 \\
    $(o_1,d)$ & $r_2$  & 7:24 - Line 1 Run 1 - Line 2 Run 1 & 46    & 1     & 0     & 2     & 0 \\
    $(o_1,d)$ & $r_3$  & 7:24 - Line 1 Run 1 - Line 2 Run 2 & 76    & 0     & 3     & 0     & 3 \\
    $(o_2,d)$ & $r_4$  & 7:49 - Line 2 Run 1 & 21    & 2     & 3     & 2     & 3 \\
    $(o_2,d)$ & $r_5$  & 7:49 - Line 2 Run 2 & 51    & 0     & 5     & 0     & 5 \\
    $(o_2,d)$ & $r_6$  & 8:09 - Line 2 Run 2 & 31    & 0     & 5     & 0     & 5 \\
    $(o_3,d)$ & $r_7$  & \textbf{7:53 - Line 1 Run 1} & \textbf{27}    & \textbf{0}     & \textbf{5}     & \textbf{1}     & \textbf{4} \\
    $(o_3,d)$ & $r_8$  & \textbf{7:53 - Line 2 Run 1} & \textbf{17}    & \textbf{2}     & \textbf{1}     & \textbf{1}     & \textbf{2} \\
    $(o_3,d)$ & $r_9$  & 7:53 - Line 2 Run 2 & 47    & 0     & 5     & 0     & 5 \\
    \bottomrule
    \end{tabular}%
  \label{tab:UEIP}%
\end{table}%
By contrast, the non-UEIP solution shown in Figure~\ref{fig:implict-non-UEIP} violates this rule. Specifically, one passenger from $o_3$ who arrives earlier at stop~C yields a seat on the faster \emph{line~2 run~1} (route $r_8$ in Table~\ref{tab:UEIP}, with cost 17) to a transferring passenger who arrives later, and instead boards \emph{line~1 run~1} (route $r_7$, with cost 27). Within the implicit-priority framework, this assignment cannot be an equilibrium: under $\vf$, route $r_8$ still has one unit of available capacity, so the passenger currently using $r_7$ has a strict incentive to switch to the available, lower-cost route $r_8$.
\begin{figure}[htbp]
    \centering
    \includegraphics[width=0.65\textwidth]{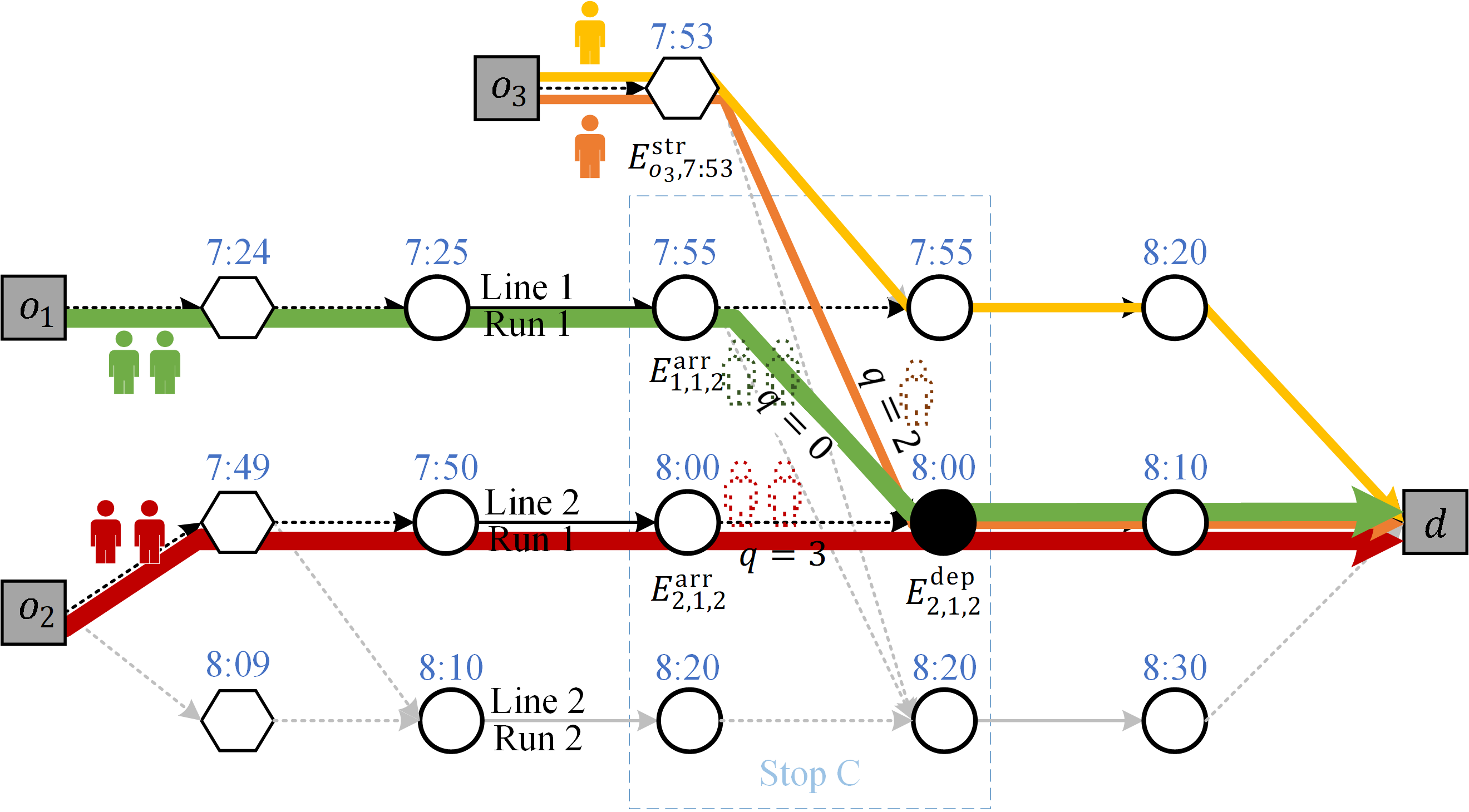}
    \caption{Visualization of a non-UEIP flow for the example transit network.}
    \label{fig:implict-non-UEIP}
\end{figure}

This example shows that the implicit-priority framework naturally enforces the boarding-priority rule: only flow patterns that respect the specified priority structure can constitute a user equilibrium. However, the implicit priority framework of \citet{nguyen_modeling_2001} lacks a formal mathematical formulation and an accompanying solution algorithm. The subsequent sections will fill these gaps.

\section{Analysis of Nguyen et al.'s Model}
\label{sec:revisit}
In this section, we will analyze \citet{nguyen_modeling_2001}'s model based on a newly proposed mathematical formulation. Throughout the analysis, the following assumptions are imposed.
\begin{assumption}
\label{ass:cost}
    For each OD pair $w \in \gW$, each class $b \in \gB_w$, and each route $r \in \gR_w$, the travel cost $c_{w,b}^r(\vf)$ is positive for all feasible $\vf \in \gF$.
\end{assumption}

\begin{assumption}
\label{ass:route}
    Given any feasible flow vector $\vf \in \gF$, for each OD pair $w \in \gW$ and class $b \in \gB_w$, there exists at least one route $r \in \gR_w$ such that $Q_{w,b}^r(\vf) > 0$, i.e., at least one route is available.
\end{assumption}

Assumption~\ref{ass:route} rules out degenerate situations in which a feasible assignment $\vf \in \gF$ exhausts the capacities of \emph{all} possible transit routes for an OD pair, including those that involve delays, which is a reasonable condition under normal operating circumstances.
The remainder of this section is organized as follows. 
Section~\ref{sec:equivalent} presents an equivalent reformulation of \citet{nguyen_modeling_2001}'s model. Section~\ref{sec:uniqueness} examines the existence and uniqueness of the solution. Finally, a numerical example illustrates that the model may admit behaviorally unrealistic solutions (Section \ref{sec:unexpected}).

\subsection{A new UEIP formulation}
\label{sec:equivalent}

To facilitate model formulation, we first present the alternative UEIP conditions that are equivalent to Condition~\eqref{equ:uep Nguyen} and more explicitly characterize the relationship between route costs and route flows.

\begin{proposition}
    A feasible solution $\bm{f}^* \in \gF$ satisfies the Condition~\eqref{equ:uep Nguyen} if and only if there exists $\mu_{w, b}\in \sR$ for each $w \in \gW$ and $b \in \gB_w$ such that 
    \begin{subequations}
    \label{equ:uep my}
        \begin{align}
            &c_{w, b}^{r}(\vf^*)\geq \mu_{w, b} \quad \text{if } Q_{w, b}^r(\vf^*)>0, \quad \forall r\in \gR_w, \label{equ:uep my1}\\
            &\mu_{w, b}\geq c_{w, b}^{r}(\vf^*) \quad \text{if } f_{w, b}^{r*}>0, \quad \forall r\in \gR_w. \label{equ:uep-my2}
        \end{align}
    \end{subequations} 
\end{proposition}
\begin{proof}
    First, suppose that $\vf^*\in \gF$ satisfies Equation \eqref{equ:uep my}. For each OD $w\in \gW$, class $b\in \gB_w$, and route $r\in \gR_w$, if $\Tilde{Q}_{w, b}^r(\vf^*)>0$, then there must be at least one route $r'\in \gR_{w, b, r}^{\text{dominate}}(\vf^*)$ satisfying $c_{w, b}^r(\vf^*)>c_{w, b}^{r'}(\vf^*)$ and $Q_{w, b}^{r'}(\vf^*)>0$. By Equation \eqref{equ:uep my1}, we have $\mu_{w, b}\leq c_{w, b}^{r'}(\vf^*) <c_{w, b}^r(\vf^*)$. Then, according to \eqref{equ:uep-my2}, we must have $f_{w, b}^{r*}=0$. Hence, Condition \eqref{equ:uep Nguyen} holds.

    Conversely, suppose that $\bm{f}^* \in \gF$ satisfies Equation \eqref{equ:uep Nguyen}. For each OD $w\in \gW$ and class $b\in \gB_w$, given that $\bar{r}$ is the most costly of all used routes, namely $c_{w, b}^{\bar{r}}(\vf^*)=\max\{c_{w, b}^r(\vf^*): f_{w, b}^{\bar{r}*}>0, r\in \gR_{w, b}\}$, and $\mu_{w, b}=c_{w, b}^{\bar{r}}(\vf^*)$, Equation \eqref{equ:uep-my2} holds. Then, we prove Equation \eqref{equ:uep my1} by contradiction. For each $w\in \gW$ and $b\in \gB_w$, suppose that there exists an available route $r'$ (i.e., $Q_{w, b}^{r'}(\vf^*)>0$) with $c_{w, b}^{r'}(\vf^*)<\mu_{w, b}=c_{w, b}^{\bar{r}}(\vf^*)$. Then, we must have $r'\in \gR_{w, b, \bar{r}}^{\text{dominate}}(\vf^*)$, and thus $\Tilde{Q}_{w, b}^{\bar{r}}(\vf^*)=\sum_{r''\in \gR_{w, b, \bar{r}}^{\text{dominate}}(\vf^*)}Q_{w, b}^{r''}(\bm{f}^*)\geq Q_{w, b}^{r'}(\vf^*)>0$. According to Condition \eqref{equ:uep Nguyen}, this conflicts with that $\bar{r}$ is a used route. By contradiction, Equation \eqref{equ:uep my1} holds.
\end{proof}

This proposition implies that, under UEIP, the cost of any used route does not exceed the cost of other available routes, with $\vmu=(\mu_{w, b})_{w\in \gW, b\in \gB_w}$ serving as the threshold separating these two categories of routes. To close the potential gap between the cost of a used route and this threshold, we introduce an additional cost variable $V_{w, b}^r$ for each OD $w\in \gW$, class $b\in \gB_w$, and route $r\in \gR_w$, and collect all such variables into the vector $\mV$. Specifically, for each used route we set $V_{w, b}^r = \mu_{w, b} - c_{w, b}^r(\vf)$, while for available routes we let $V_{w, b}^r = 0$. Based on this construction, we obtain the following NCP model: 
\begin{subequations}
\label{prob:cp-route}
    \begin{align}
        &0\leq \vf \bot \vc(\vf)+\mV-\Lambda\vmu \geq 0, \label{equ:cp route1} \\
        &0\leq \vmu \bot \Lambda^T\vf-\vd \geq 0, \label{equ:cp route2} \\
        &0\leq \mV \bot \mQ(\vf) \geq 0, \label{equ:cp-route3}
    \end{align}
\end{subequations}
where $\vc(\vf)=(c_{w,b}^r(\vf))_{w\in \gW, b\in \gB_w, r\in \gR_w}$ and $\mQ(\vf)=(Q_{w,b}^r(\vf))_{w\in \gW, b\in \gB_w, r\in \gR_w}$ denote the vectors of travel costs and available capacities, respectively; $\Lambda=[\lambda_{r, b}]_{|n_f|\times|\gB|}$ denotes the route-class incident matrix, with $\lambda_{r, b}=1$ if route $r$ belongs to class $b$, and $\lambda_{r, b}=0$ otherwise; $n_f=\sum_{w\in \gW}|\gB_w|\times|\gR_w|$ is the dimension of route flow vector $\vf$. The subsequent theorem shows that this NCP formulation is equivalent to the UEIP conditions.

\begin{theorem}
    \label{thm:equivalent}
     Under Assumption~\ref{ass:cost}, the NCP \eqref{prob:cp-route} is equivalent to the UEIP conditions \eqref{cons:demand}, \eqref{cons:non-negative}, \eqref{cons:capacity}, and \eqref{equ:uep my}.
\end{theorem}
\begin{proof}
    \textbf{Necessity}: Suppose that the solution $(\bm{f}^*, \bm{\mu}^*)$ satisfies the UEIP conditions \eqref{cons:demand}, \eqref{cons:non-negative}, \eqref{cons:capacity}, and \eqref{equ:uep my}. Let $\vx^*=\Delta^T\vf^*$, where $\Delta=[\delta_{w, b}^{r,A}]_{|n_f|\times|\gA|}$ is the route-arc incident matrix. We will then show that the complementarity condition \eqref{prob:cp-route} holds. 
    
    \textbf{Condition \eqref{equ:cp route2}}: Condition \eqref{cons:demand} shows that $\Lambda^T\vf^*-\vd=0$ and $\vmu^* \bot \Lambda^T\vf^*-\vd$. Equation \eqref{equ:uep-my2} implies $\vmu$ is a positive vector. Therefore, Condition \eqref{equ:cp route2} holds.
    
    \textbf{Condition \eqref{equ:cp-route3}}: For any arc $A\in \gA^{\text{priority}}$, we have $q_A(\vx^*)=u_{\text{riding}(A)}-\sum_{A'\in \text{Prior}(A)}x_{A'}^* \geq u_{\text{riding}(A)} - x_{\text{riding}(A)}^*$.
    By the capacity constraint \eqref{cons:capacity}, we have $q_A(\vx^*)\geq 0$. Therefore, $Q_{w,b}^r(\vf^*)=\min\{q_A(\vx^*): A\in \gA_{w,b,r}^{\text{priority}}\}\geq 0$ for all $w \in \gW$, $b \in \gB_w$, and $r\in \gR_w$.

    For each OD $w\in \gW$, class $b\in \gB_w$, and route $r\in \gR_w$, letting
    \begin{align}
        V_{w, b}^r=\left\{
        \begin{array}{rl}
        0, & \quad \text{if } Q_{w,b}^r(\vf^*)>0 \\
        \mu_{w, b}^* - c_{w, b}^r(\vf^*), & \quad \text{if } Q_{w,b}^r(\vf^*)=0 \text{ and } f_{w, b}^{r*}>0 \\
        \max\{\mu_{w, b}^* - c_{w, b}^r(\vf^*), 0\}, &  \quad \text{if } Q_{w,b}^r(\vf^*)=0 \text{ and } f_{w, b}^{r*}=0 
        \end{array},\right. 
        \label{equ:equivalence proof1}
    \end{align}
    we then have $\mV \bot \mQ(\vf^*)$. Meanwhile, considering that $\mu_{w, b}^* - c_{w, b}^r(\vf^*)\geq 0$ when $f_{w, b}^{r*}>0$ by Equation \eqref{equ:uep-my2}, we must have $\mV\geq 0$. Therefore, Condition \eqref{equ:cp-route3} holds.

    \textbf{Condition \eqref{equ:cp route1}}: Following Equation\eqref{equ:equivalence proof1}, we also have $\vc(\vf)+\mV-\Lambda\vmu \geq 0$ because (I) $V_{w, b}^r\geq \mu_{w, b}^* - c_{w, b}^r(\vf^*)$, when $Q_{w,b}^r(\vf^*)=0$; (II) $V_{w, b}^r=0$ and $c_{w, b}^r(\vf^*)\geq \mu_{w, b}^*$ when $Q_{w,b}^r(\vf^*)>0$ by Equation \eqref{equ:uep my1}.
    
    The orthogonality between $\vf$ and $\vc(\vf)+\mV-\Lambda\vmu$ must stands when $f_{w, b}^{r*}=0$. If $f_{w, b}^{r*}>0$, there are two cases. Case (I): If $Q_{w,b}^r(\vf^*)>0$,  Equations \eqref{equ:uep my1} and \eqref{equ:uep-my2} show $c_{w, b}^r(\vf^*)\geq \mu_{w, b}^*\geq c_{w, b}^r(\vf^*)$, so $c_{w, b}^r(\vf^*)=\mu_{w, b}^*$. Then, we have $c_{w, b}^r(\vf^*) + V_{w, b}^r - \mu_{w, b}^*=0$ as $V_{w, b}^r = 0$. Case (II): If $Q_{w,b}^r(\vf^*) = 0$, we also have $c_{w, b}^r(\vf^*) + V_{w, b}^r - \mu_{w, b}^*=0$ because $V_{w, b}^r$ is set to $\mu_{w, b}^* - c_{w, b}^r(\vf^*)$ in Equation \eqref{equ:equivalence proof1}. Therefore, orthogonality is satisfied universally, and Condition \eqref{equ:cp route1} holds.

    \textbf{Sufficiency}: Suppose that the solution $(\vf^*, \vmu^*, \mV^*)$ satisfies the complementarity condition \eqref{prob:cp-route}. Letting $\vx^*=\Delta^T\vf^*$, we then show that the UEIP conditions \eqref{cons:demand}, \eqref{cons:non-negative}, \eqref{cons:capacity}, and \eqref{equ:uep my} hold.

    \textbf{Condition \eqref{cons:non-negative}:} The complementarity condition \eqref{equ:cp route1} directly suggests that  $f_{w, b}^{r*}\geq 0$ for all $w \in \gW$, $b \in \gB_w$, and $r\in \gR_w$.

    \textbf{Condition \eqref{cons:capacity}:} The non-negativity of $\mQ(\vf^*)$ suggests that $q_A(\vx^*)\geq 0$ for all $A\in \gA^{\text{priority}}$. For any riding arc $\bar{A}\in \gA^{\text{riding}}$, let $A$ be the incoming arc with the lowest loading priority. We then have $0\leq q_A(\vx^*)=u_{\bar{A}}-\sum_{A'\in \text{Prior}(A)}x_{A'}^* = u_{\bar{A}} - x_{\bar{A}}^*$. Thus, the capacity constraint \eqref{cons:capacity} is satisfied.

    \textbf{Condition \eqref{cons:demand}:} Suppose that for some OD $w\in \gW$ and class $b\in \gB_w$, 
    \begin{align}
        \sum_{r \in \gR_w} f_{w, b}^{r*} > d_{w, b}. \label{equ:equilivent proof 1}
    \end{align}
    By the complementarity condition \eqref{equ:cp route1} and \eqref{equ:cp route2}, we must have $\mu_{w, b}^*=0$ and for each route $r\in \gR_w$,
    \begin{align}
        f_{w, b}^{r*}(c_{w, b}^r(\vf^*) + V_{w, b}^{r*} - \mu_{w, b}^*)=0 \Rightarrow f_{w, b}^{r*}(c_{w, b}^r(\vf^*) + V_{w, b}^{r*})=0. \notag
    \end{align}
    In addition, since $c_{w, b}^r(\vf^*)$ is positive and $V_{w, b}^{r*}$ is non-negative, we must have $f_{w, b}^{r*}=0$ for all $w \in \gW$, $b \in \gB_w$, and $r\in \gR_w$. 
    
    On the other hand, Equation \eqref{equ:equilivent proof 1} implies that $\sum_{r \in \gR_w} f_{w, b}^{r*}>0$ and there must exist a route with positive flow ($f_{w, b}^{r*}>0$) in this OD and class, which is a contradiction. Thus, the demand constraint \eqref{cons:demand} holds.

    \textbf{Condition \eqref{equ:uep my}:} For each OD $w\in \gW$, class $b\in \gB_w$, and route $r\in \gR_w$, if $Q_{w,b}^r(\vf^*)>0$, we have $V_{w, b}^{r*}=0$ by the complementarity condition \eqref{equ:cp-route3}. Since $c_{w, b}^r(\vf^*) + V_{w, b}^{r*} - \mu_{w, b}^*$ is non-negative, we get $c_{w, b}^r(\vf^*)\geq \mu_{w, b}^*$, and hence Condition \eqref{equ:uep my1} holds.

    On the other hand, if $f_{w, b}^{r*}>0$, by the complementarity condition \eqref{equ:cp route1}, we have $c_{w, b}^r(\vf^*) + V_{w, b}^{r*} - \mu_{w, b}^*=0$. Therefore, $c_{w, b}^r(\vf^*)\leq \mu_{w, b}^*$ since $V_{w, b}^{r*}$ is non-negative, and thus Condition \eqref{equ:uep-my2} holds.
\end{proof}

\subsection{Existence and uniqueness}
\label{sec:uniqueness}
We first establish the existence of a solution to the NCP~\eqref{prob:cp-route}, and then discuss its uniqueness.

\begin{proposition}
\label{prop:route solution exists}
    Under Assumptions~\ref{ass:cost} and \ref{ass:route}, NCP~\eqref{prob:cp-route} has a solution.
\end{proposition}
\begin{proof}
    To prove the existence of a solution to NCP~\eqref{prob:cp-route}, we begin by recalling that a solution is guaranteed when the mapping is continuous and the feasible set is compact. In NCP~\eqref{prob:cp-route}, although the function
    $$
    H(\vf, \vmu, \mV)=\left(
        \begin{array}{c}
        \vc(\vf)+\mV-\Lambda\vmu \\
        \Lambda^T\vf-\vd \\
        \mQ(\vf) \\
        \end{array}
        \right)
    $$
    is continuous, the feasible set $\Omega=\{\vf\geq 0, \vmu\geq 0, \mV\geq 0\}$ is unbounded and therefore not compact. Our proof proceeds by introducing upper-bound constraints for each variable, thereby constructing a new NCP whose feasible set is compact. We then show that none of these added upper bounds are binding at the solution of the new problem. Consequently, any solution to the modified NCP also satisfies the original NCP~\eqref{prob:cp-route}.
    Choose two scalars $e_1$ and $e_2$ such that
    \begin{align}
        e_1 > \max\{d_{w, b}: w\in \gW, b\in \gB_w\} \text{ and } e_2 > \max\{c_{w, b}^r: w\in \gW, b\in \gB_w, r\in \gR_w\}. \notag
    \end{align}
    Define $\Omega'=\Omega\cap\{\vf\leq e_1\bm{1}, \vmu\leq e_2\bm{1}, \mV\leq e_2\bm{1}\}$. Since $\Omega'$ is compact, the following NCP must possess a solution: 
    \begin{subequations}
        \begin{align}
            &0\leq \vf \bot \vc(\vf)+\mV-\Lambda\vmu + \bm{\kappa}\geq 0, \label{equ:exist cc1} \\
            &0\leq \vmu \bot \Lambda^T\vf-\vd +\bm{\rho} \geq 0, \label{equ:exist cc2} \\
            &0\leq \mV \bot \mQ(\vf) +\bm{\nu} \geq 0, \label{equ:exist-cc3} \\
            &0\leq \bm{\kappa} \bot e_1\bm{1}-\vf \geq 0, \label{equ:exist cc4}\\
            &0\leq \bm{\rho} \bot e_2\bm{1}-\vmu \geq 0, \label{equ:exist cc5}\\
            &0\leq \bm{\nu} \bot e_2\bm{1}-\mV \geq 0. \label{equ:exist cc6}
        \end{align}
    \end{subequations}
    Let $(\vf^*, \vmu^*, \mV^*)$ be this solution, if we show that $\bm{\kappa}$, $\bm{\rho}$, and $\bm{\nu}$ are all equal to zero, then $(\vf^*, \vmu^*, \mV^*)$ must also be a solution of NCP~\eqref{prob:cp-route}.

    \textbf{For $\bm{\kappa}$}: Suppose $\kappa_{w, b}^r>0$ for some $w\in \gW, b\in \gB_w, r\in \gR_w$. By the complementarity condition \eqref{equ:exist cc4}, we have $f_{w, b}^{r*}=e_1>d_{w, b}$, so $\sum_{r\in \gR_w}f_{w, b}^{r*}-d_{w, b}+\rho_{w, b}>0$. Then, according to Condition \eqref{equ:exist cc2}, we have $\mu_{w, b}^*=0$. Moreover, since $f_{w, b}^{r*}=e_1>0$, we have 
    \begin{align}
        0=c_{w, b}^r(\vf^*)+V_{w, b}^{r*}-\mu_{w, b}^*+\kappa_{w, b}^r=c_{w, b}^r(\vf^*)+V_{w, b}^{r*}+\kappa_{w, b}^r >0, \notag
    \end{align}
    because $c_{w, b}^r(\vf^*)$ and $\kappa_{w, b}^r$ are positive, and $V_{w, b}^{r*}$ is non-negative. This contradiction yields $\bm{\kappa}= \vzero$.

    \textbf{For $\bm{\rho}$}: For each OD pair $w\in \gW$ and class $b\in \gB_w$, since there exists a route $r\in \gR_w$ with $Q_{w, b}^r(\vf^*)>0$ by assumption, we have $Q_{w, b}^r(\vf^*)+\nu_{w, b}^r>0$. By Condition \eqref{equ:exist-cc3}, $V_{w, b}^{r*}=0$. Therefore, 
    \begin{align}
        &0\leq c_{w, b}^r(\vf^*)+V_{w, b}^{r*}-\mu_{w, b}^*+\kappa_{w, b}^r = c_{w, b}^r(\vf^*) - \mu_{w, b}^* \notag \\
        &\Rightarrow \mu_{w, b}^*\leq c_{w, b}^r(\vf^*)< e_2. \label{equ:exist-mu}
    \end{align}
    This means that $e_2 - \mu_{w, b}^*>0$ for all $w \in \gW$ and $b \in \gB_w$, and $\bm{\rho}= \vzero$ by Condition \eqref{equ:exist cc5}.

    \textbf{For $\bm{\nu}$}: Supposing that $\nu_{w, b}^r>0$ for some $w\in \gW, b\in \gB_w, r\in \gR_w$, we then have $V_{w, b}^{r*}=e_2>0$. By Condition \eqref{equ:exist-cc3}, we have $Q_{w, b}^r(\vf^*)+\nu_{w, b}^r=0$, hence $Q_{w, b}^r(\vf^*) = -\nu_{w, b}^r < 0$. By the definition of $Q_{w, b}^r(\vf^*)$, there must exist an arc with negative available capacity in route $r$, namely one $A\in \gA_{w, b, r}^{\text{priority}}$ such that $0 > q_A(\vx^*)=u_{\text{riding}(A)}-\sum_{A'\in \text{Prior}(A)}x_{A'}^*$. 

    Let $A'$ be the arc in $\text{Prior}(A)$ that carries positive flow and has the lowest loading priority. Then, $A'$ has the same available capacity as $A$, that is $q_{A'}(\vx^*) = q_A(\vx^*) < 0$. Moreover, there is at least one used route $r'$ that passes through this arc. Without loss of generality, let $r'$ belong to OD pair $w'\in \gW$ and class $b'\in \gB_w$. By the complementarity condition \eqref{equ:exist cc1} and the fact that $f_{w', b'}^{r'*}>0$, we have 
    \begin{align}
        c_{w', b'}^{r'}(\vf^*)+V_{w', b'}^{r'*}-\mu_{w', b'}^*+\kappa_{w', b'}^{r'}=0 \Rightarrow V_{w', b'}^{r'*}=\mu_{w', b'}^* - c_{w', b'}^{r'}(\vf^*). \notag
    \end{align}
    Since $\mu_{w', b'}^*< e_2$ (Equation \eqref{equ:exist-mu}) and $c_{w', b'}^{r'}(\vf^*)$ is positive, we get $V_{w', b'}^{r'*}< e_2$. This means $\nu_{w', b'}^{r'}=0$ by Condition \eqref{equ:exist cc6}. On the another hand, $Q_{w', b'}^{r'}(\vf^*)=\min\{q_{A''}(\vx^*): A''\in \gA_{w,b,r}^{\text{priority}}\}\leq q_{A'}(\vx^*)<0$. This implies $Q_{w', b'}^{r'}(\vf^*) + \nu_{w', b'}^{r'}<0$, which contradicts Condition \eqref{equ:exist-cc3}. Therefore, $\bm{\nu}= \vzero$ holds.  
\end{proof}

Having established the existence of a solution to NCP~\eqref{prob:cp-route}, we will then turn to its uniqueness. Since the mapping $H(\vf, \vmu, \mV)$ is non-monotone, the solution set of NCP~\eqref{prob:cp-route} --- equivalently, the UEIP solution set --- is generally not unique.  We next discuss the consequence of this non-uniqueness. 
\subsection{Unexpected consequence of non-uniqueness}
\label{sec:unexpected}
Many traffic assignment problems do not admit unique solutions. When this occurs, it is common practice to treat any solution returned by a valid algorithm as equally admissible. However, as we show below, some UEIP solutions can exhibit clear violations of realistic passenger behavior. This undermines the applicability of the NCP formulation \eqref{prob:cp-route}, as it may admit solutions that are behaviorally implausible or operationally nonsensical.

Consider a toy network shown in Figure~\ref{fig:unreasonableUEP1}, which consists of two lines (Line 1 and Line 2).
\begin{figure}[htbp]
	\centering
	\subfigure[Network]{
	{\includegraphics[width=0.48\textwidth]{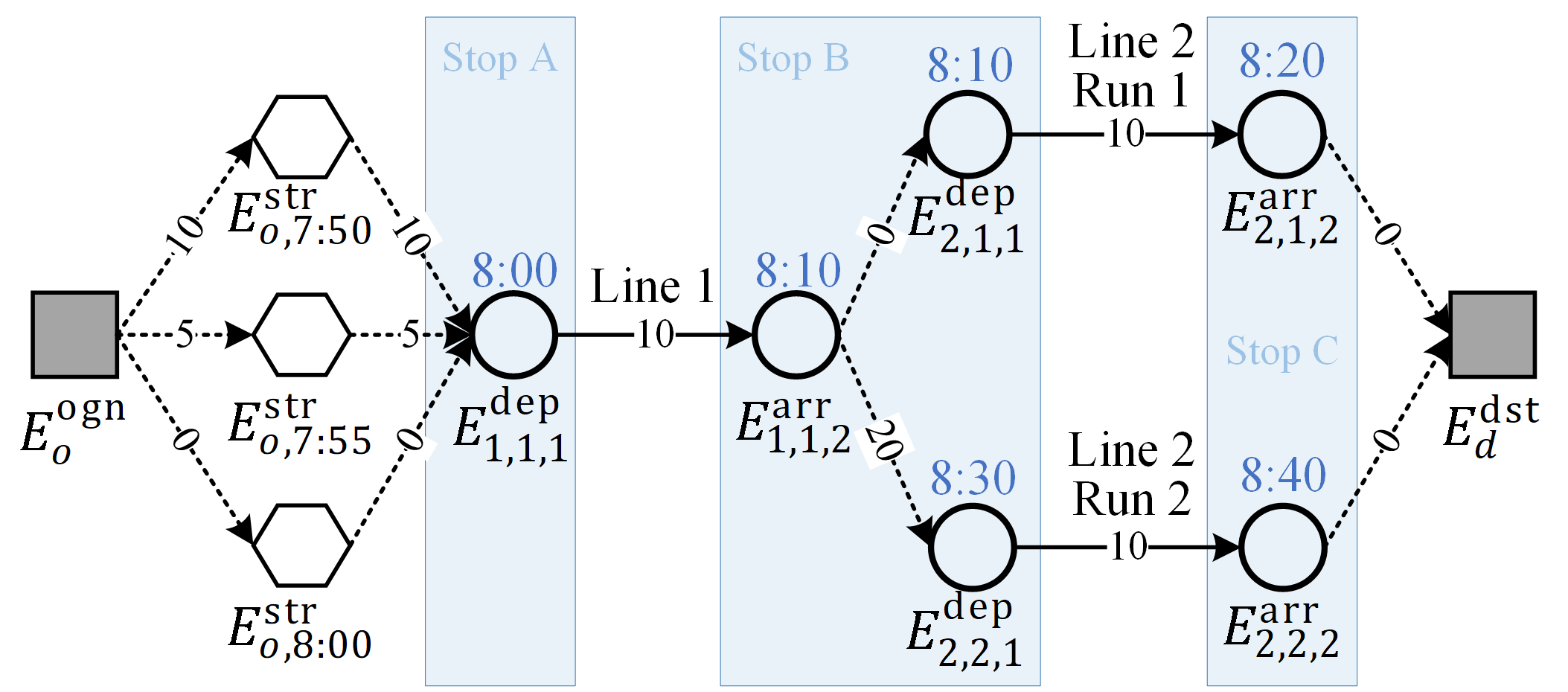}}\label{fig:unreasonableUEP1}
	}
	\subfigure[UEIP–I]{
	\includegraphics[width=0.48\textwidth]{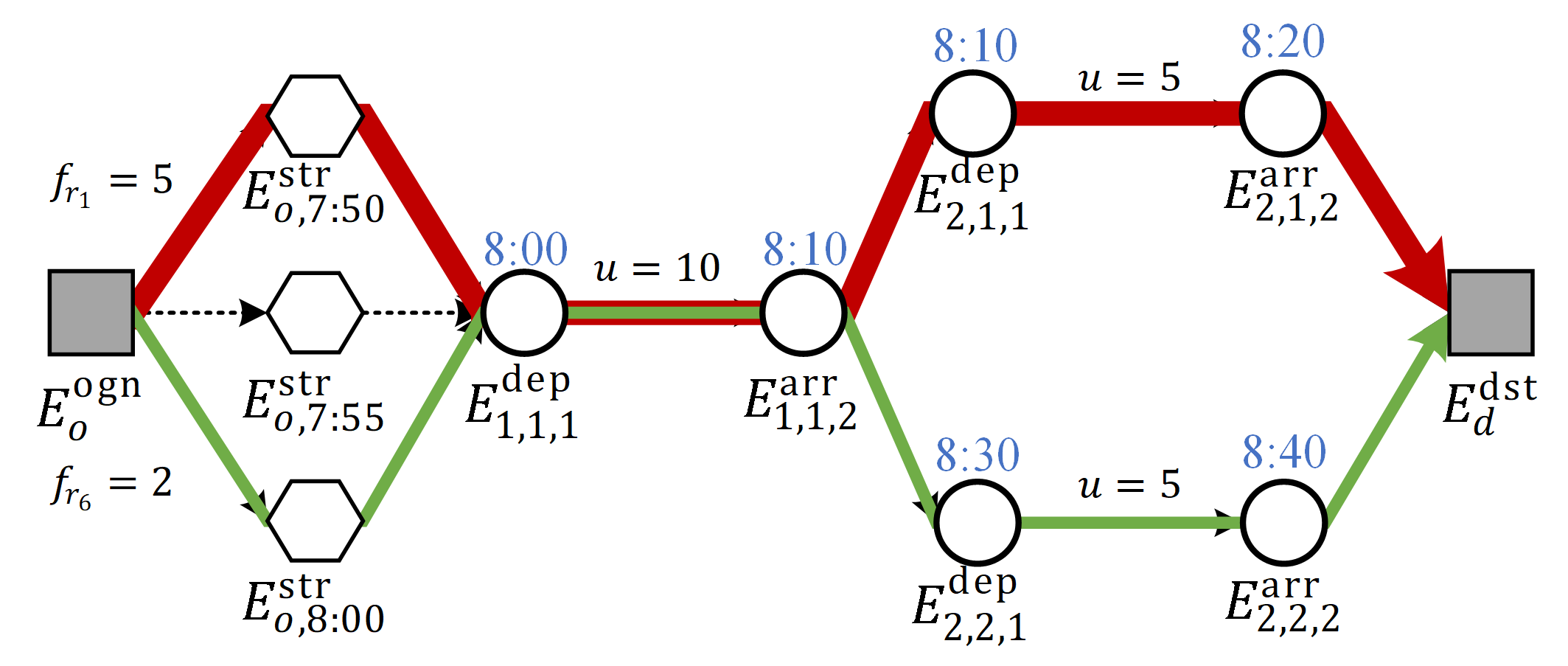}\label{fig:unreasonableUEP2}
	}
        \subfigure[UEIP–II]{
	\includegraphics[width=0.48\textwidth]{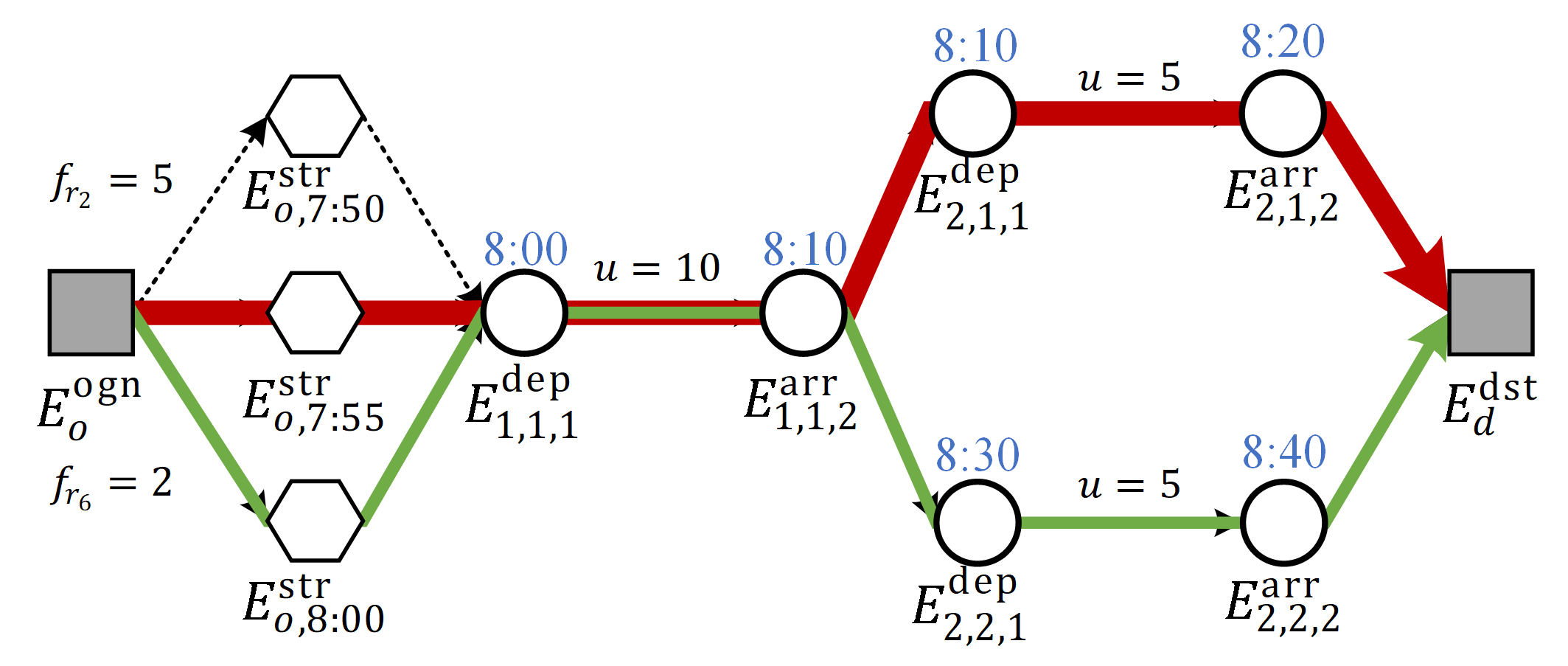}\label{fig:unreasonableUEP3}
	}
        \subfigure[UEIP–III]{
	\includegraphics[width=0.48\textwidth]{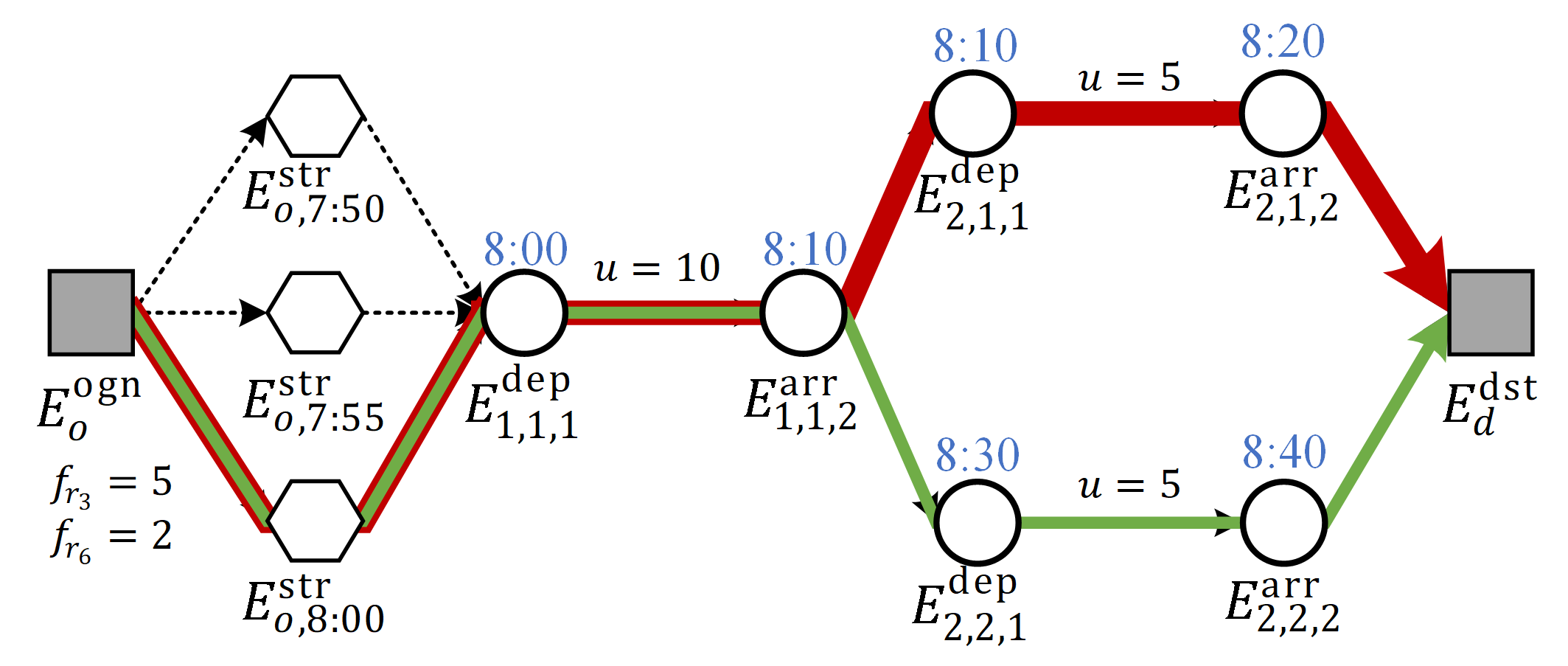}\label{fig:unreasonableUEP4}
	}
	\caption{An illustrative example of unreasonable UEIP states.}
	\label{fig:unreasonableUEP}
\end{figure}
Specifically, Line 1 has only one run, while Line 2 includes two runs (Run 1 and Run 2). Congestion effects are not considered, and each arc is labeled with its cost. Seven passengers depart from origin $o$ to destination $d$, with three starting time options: 7:50, 7:55, and 8:00. All passengers first take the high-capacity Line 1 (capacity = 10) from Stop A to Stop B. Then, a subset of them can directly transfer to the low-capacity Run 1 of Line 2 (capacity = 5) to reach the destination without waiting, while the remaining two passengers must stay at the platform and wait for Run 2 of Line 2. This setting yields three distinct UEIP solutions, namely UEIP–I, UEIP–II, and UEIP–III as shown in Figure~\ref{fig:unreasonableUEP2}--\ref{fig:unreasonableUEP4} and Table~\ref{tab:unreasonableUEP}, which differ only in the starting times chosen by the passengers taking Run 1. In UEIP–I and UEIP–II, these passengers begin their trips at 7:50 and 7:55, respectively, to catch Line 1 leaving at 8:00. However, such outcomes are unrealistic. Since Line 1 is unsaturated --- meaning that passengers can board regardless of their arrival order at Stop A --- the optimal behavior is simply to depart at 8:00, as in UEIP–III. Departing earlier yields no advantage and only increases waiting time and the early-departure penalty.

\begin{table}[htbp]
  \centering
  \caption{Summary of the three UEIP solutions.}
  \footnotesize %
    \begin{tabular}{cccccccccccc}
    \toprule
    \multirow{2}[4]{*}{Routes} & \multirow{2}[4]{*}{Trajectories} & \multirow{2}[4]{*}{Cost} & \multicolumn{3}{c}{UEIP–I} & \multicolumn{3}{c}{UEIP–II} & \multicolumn{3}{c}{UEIP–III} \\
\cmidrule{4-12}          &       &       & $\vf$     & $\mQ$     & $\mV$     & $\vf$     & $\mQ$     & $\mV$     & $\vf$     & $\mQ$     & $\mV$ \\
    \midrule
    $r_1$  & 7:50, Run 1 & 40    & 5     & 0     & 0     & 0     & 0     & $\geq 0$ & 0     & 0     & $\geq 0$ \\
    $r_2$  & 7:55, Run 1 & 30    & 0     & 0     & $\geq 10$ & 5     & 0     & 10    & 0     & 0     & $\geq 10$ \\
    $r_3$  & 8:00, Run 1 & 20    & 0     & 0     & $\geq 20$ & 0     & 0     & $\geq 20$ & 5     & 0     & 20 \\
    $r_4$  & 7:50, Run 2 & 60    & 0     & 3     & 0     & 0     & 3     & 0     & 0     & 3     & 0 \\
    $r_5$  & 7:55, Run 2 & 50    & 0     & 3     & 0     & 0     & 3     & 0     & 0     & 3     & 0 \\
    $r_6$  & 8:00, Run 2 & 40    & 2     & 3     & 0     & 2     & 3     & 0     & 2     & 3     & 0 \\
    \bottomrule
    \end{tabular}%
  \label{tab:unreasonableUEP}%
\end{table}%

To illustrate the reason behind the existence of the unrealistic solutions, we take UEIP–II as an example. Passengers assigned to Run 1 (i.e., route $r_2$) cannot shift their starting time to 8:00 (i.e., move to route $r_3$) because the arc $(E_{1,1,2}^{\text{arr}}, E_{2,1,1}^{\text{dep}})$ is treated as unavailable, thereby rendering route $r_3$ unavailable. In reality, however, we argue that route $r_3$ should be regarded as “available” for passengers on $r_2$, since transferring flow from $r_2$ to $r_3$ merely changes their starting time and does not introduce additional flow to the bottleneck arc $(E_{1,1,2}^{\text{arr}}, E_{2,1,1}^{\text{dep}})$. More generally, even if route $r$ contains some unavailable arcs, when another route $r'$ also traverses these arcs, reallocating flow from $r'$ to $r$ does not increase the load on those unavailable arcs. Such a reallocation should therefore be admissible, and route $r$ should be considered available relative to route $r'$.

In summary, this example highlights that the definition of route availability deviates from actual passenger behavior, which leads to unrealistic solutions in the UEIP solution set. In the next section, we refine the definition of available routes and introduce a new equilibrium condition together with its corresponding NCP formulation, which eliminates such unrealistic outcomes. We also develop an algorithm tailored to solve the proposed model.

\section{A Behaviorally Compliant UEIP Formulation}
\label{sec:refined}

We now proceed to a refined UEIP condition, which not only rules out the behaviorally unrealistic solutions but also enables more tractable algorithms. In what follows, Section~\ref{sec:properties} presents the refined UEIP conditions and the corresponding NCP formulation, and Section~\ref{sec:algorithms} develops the solution algorithms.

\subsection{UEIP revisited}
\label{sec:properties}

The preceding example highlights that route availability is not absolute but relative to the reference route, as reallocating flow may remain admissible if no additional load is imposed on unavailable arcs.  Accordingly, we define the available capacity of route $r'$ with respect to route $r$ as
\begin{align}
    Q_{w,b}^{r',r}(\vf)=\min\{q_A(\vx): A\in \gA_{w,b,r'}^{\text{priority}}\setminus \gA_{w,b,r}^{\text{priority}}\},\quad \forall w \in \gW, \ \forall b \in \gB_w, \ \forall r',r\in \gR_w \text{ and } r'\neq r. \notag
\end{align}
Note that since $\gA_{w,b,r'}^{\text{priority}}\setminus \gA_{w,b,r}^{\text{priority}}$ is nonempty whenever $r'\neq r$, the quantity $Q_{w,b}^{r',r}(\vf)$ is well-defined.
Based on this, we revise the definition of route availability as follows:
\begin{definition}[Relative route availability]
    Given any feasible flow vector $\vf \in \gF$, for each OD pair $w \in \gW$ and class $b \in \gB_w$, a route $r\in \gR_w$ is said to be available with respect to another route $r'\in \gR_w$ if every arc on $r'$ that is not contained in $r$ is available. Equivalently, this condition holds if $Q_{w,b}^{r',r}(\vf) > 0$.
\end{definition}

With this definition, we can tighten the equilibrium condition to rule out behavioral inconsistency. 

\begin{definition}[Refined user equilibrium with implicit priority]
\label{def:refine uep}
A feasible flow vector $\vf\in \gF$ is a refined user equilibrium with implicit priority (R-UEIP) if for each OD pair $w\in \gW$ and class $b\in \gB_w$ no individual on route $r\in \gR_w$ can reduce their travel cost $c_{w,b}^r(\vf)$ by unilaterally switching to another available route that is available with respect to $r$. Formally, let the total available capacity of the dominant routes with respect to $r$ be $\hat{Q}_{w,b}^r(\vf)=\sum_{r'\in \gR_{w,b,r}^{\text{dominate}}(\vf)}Q_{w,b}^{r',r}(\vf)$. Then, $\vf\in \gF$ is an R-UEIP solution if 
    \begin{align}
        f_{w,b}^{r}=0 \quad \text{whenever } \hat{Q}_{w,b}^r(\vf)>0, \quad \forall w \in \gW, \ \forall b \in \gB_w, \ \forall r\in \gR_w. \label{equ:ruep}
    \end{align}
\end{definition}

It is straightforward to verify that, in the toy example of the previous section, the unrealistic equilibria UEIP–I and UEIP–II do not satisfy the refined equilibrium condition, whereas UEIP–III does. The relationship between the two types of equilibrium can be summarized as follows:
\begin{proposition}
    If $\vf^*$ is an R-UEIP, then it is also a UEIP, but the converse does not hold.
\end{proposition}
\begin{proof}
    Suppose that $\vf^*$ satisfies the R-UEIP condition. Let $\vx^*=\Delta^T \vf$. For any OD pair $w\in \gW$, any class $b\in \gB_w$, and any two different routes $r, r' \in \gR_w$, it follows that $Q_{w,b}^{r',r}(\vf^*)=\min\{q_A(\vx^*): A\in \gA_{w,b,r'}^{\text{priority}}\setminus \gA_{w,b,r}^{\text{priority}}\} \geq \min\{q_A(\vx^*): A\in \gA_{w,b,r'}^{\text{priority}}\} = Q_{w,b}^{r'}(\vf^*)$. Consequently,
    \begin{align}
        \hat{Q}_{w,b}^r(\vf^*)=\sum_{r'\in \gR_{w,b,r}^{\text{dominate}}(\vf^*)}Q_{w,b}^{r',r}(\vf^*) \geq \sum_{r'\in \gR_{w,b,r}^{\text{dominate}}(\vf^*)}Q_{w,b}^{r'}(\vf^*) = \Tilde{Q}_{w,b}^r(\vf^*). \notag
    \end{align}
    Therefore, if $\Tilde{Q}_{w,b}^r(\vf^*) > 0$, then $\hat{Q}_{w,b}^r(\vf^*) > 0$. According to the R-UEIP condition~\eqref{equ:ruep}, this implies $f_{w,b}^{r*} = 0$, thus proving the forward direction. The counterexample provided in Section~\ref{sec:uniqueness} demonstrates that the converse does not hold.
\end{proof}

This result implies that the R-UEIP admits a smaller solution set, excluding certain unrealistic outcomes permitted by the UEIP. Next, we formulate an NCP model for computing the refined equilibrium. For every arc $A\in \gA^{\text{priority}}$, we introduce a variable $v_A$, collected into the vector $\vv$, which is orthogonal to the arc's available capacity; that is, $v_A\bot q_A(\vx)$. Let $\bar{\Delta}=[\delta_{w, b}^{r,A}]_{|n_f|\times|\gA^{\text{priority}}|}$ be the route-arc incident matrix. We propose the following NCP model: 
\begin{subequations}
\label{prob:cp}
    \begin{align}
        &0\leq \vf \bot \vc(\vf)+\bar{\Delta}\vv-\Lambda\vmu \geq 0, \label{equ:cp1} \\
        &0\leq \vmu \bot \Lambda^T\vf-\vd \geq 0, \label{equ:cp2} \\
        &0\leq \vv \bot \vq(\vx) \geq 0. \label{equ:cp3}
    \end{align}
\end{subequations}

\begin{proposition}
    Under Assumption~\ref{ass:cost}, if $\vf^*$ is a solution of NCP \eqref{prob:cp}, $\vf^*$ satisfies the R-UEIP conditions \eqref{cons:demand}, \eqref{cons:non-negative}, \eqref{cons:capacity}, and \eqref{equ:ruep}.
\end{proposition}
\begin{proof}
    The proof that $\bm{f}^*$ satisfies the feasible conditions~\eqref{cons:demand}, \eqref{cons:non-negative}, and~\eqref{cons:capacity} follows the same reasoning as in Theorem~\ref{thm:equivalent}, with $\mV$ replaced by $\bar{\Delta}\vv$, and is thus omitted here. Next, we prove that $\vf^*$ also satisfies condition~\eqref{equ:ruep}. 
    
    For any OD pair $w\in \gW$ and any class $b\in \gB_w$, suppose there exists a route $r\in \gR_w$ such that $\hat{Q}_{w,b}^r(\vf^*) > 0$. Then, there must exist another route $r'\in \gR_w$ satisfying $c_{w, b}^{r'}(\vf^*) < c_{w, b}^{r}(\vf^*)$ and $Q_{w,b}^{r',r}(\vf^*) > 0$. This implies $q_A > 0$ for all $A \in \gA_{w,b,r'}^{\text{priority}}\setminus \gA_{w,b,r}^{\text{priority}}$. According to the complementarity condition~\eqref{equ:cp3}, we have $v_A = 0$ for all such $A$. Substituting these $v_A$ values into~\eqref{equ:cp1} yields:
    \begin{align}
        0\leq c_{w, b}^{r'}(\vf^*) + \sum_{A\in \gA_{w,b,r'}^{\text{priority}}}v_A - \mu_{w, b} = c_{w, b}^{r'}(\vf^*) + \sum_{A\in \gA_{w,b,r'}^{\text{priority}}\cap \gA_{w,b,r}^{\text{priority}}}v_A - \mu_{w, b}. \notag
    \end{align}
    Moreover, since $c_{w, b}^{r'}(\vf^*) < c_{w, b}^{r}(\vf^*)$ and $\sum_{A\in \gA_{w,b,r'}^{\text{priority}}\cap \gA_{w,b,r}^{\text{priority}}}v_A \leq \sum_{A\in \gA_{w,b,r}^{\text{priority}}}v_A$, it follows that
    \begin{align}
        0 < c_{w, b}^{r}(\vf^*) + \sum_{A\in \gA_{w,b,r}^{\text{priority}}}v_A - \mu_{w, b}. \notag
    \end{align}
    Finally, by the complementarity condition in~\eqref{equ:cp1}, we obtain $f_{w, b}^{r*} = 0$, which proves that $\vf^*$ satisfies~\eqref{equ:ruep}.
\end{proof}

Under the same assumptions as NCP \eqref{prob:cp-route}, the existence of a solution for NCP \eqref{prob:cp} can also be guaranteed. The result regarding the existence of the solution is described as follows:

\begin{proposition}
\label{prop:model2 solution exists}
    Under Assumptions \ref{ass:cost} and \ref{ass:route}, NCP~\eqref{prob:cp} has a solution.
\end{proposition}
\begin{proof}
    See Appendix \ref{sec:appendix-proof-solution-exists}.
\end{proof}

\begin{remark}
    The introduced variable $\vv$ can be interpreted as the anxiety cost incurred by passengers when they traverse activities (arcs) for which, due to their lower loading priority, boarding cannot be guaranteed. As illustrated in Figure~\ref{fig:explain-multiplier}, 
    \begin{figure}[htbp]
    \centering
    \includegraphics[width=0.45\textwidth]{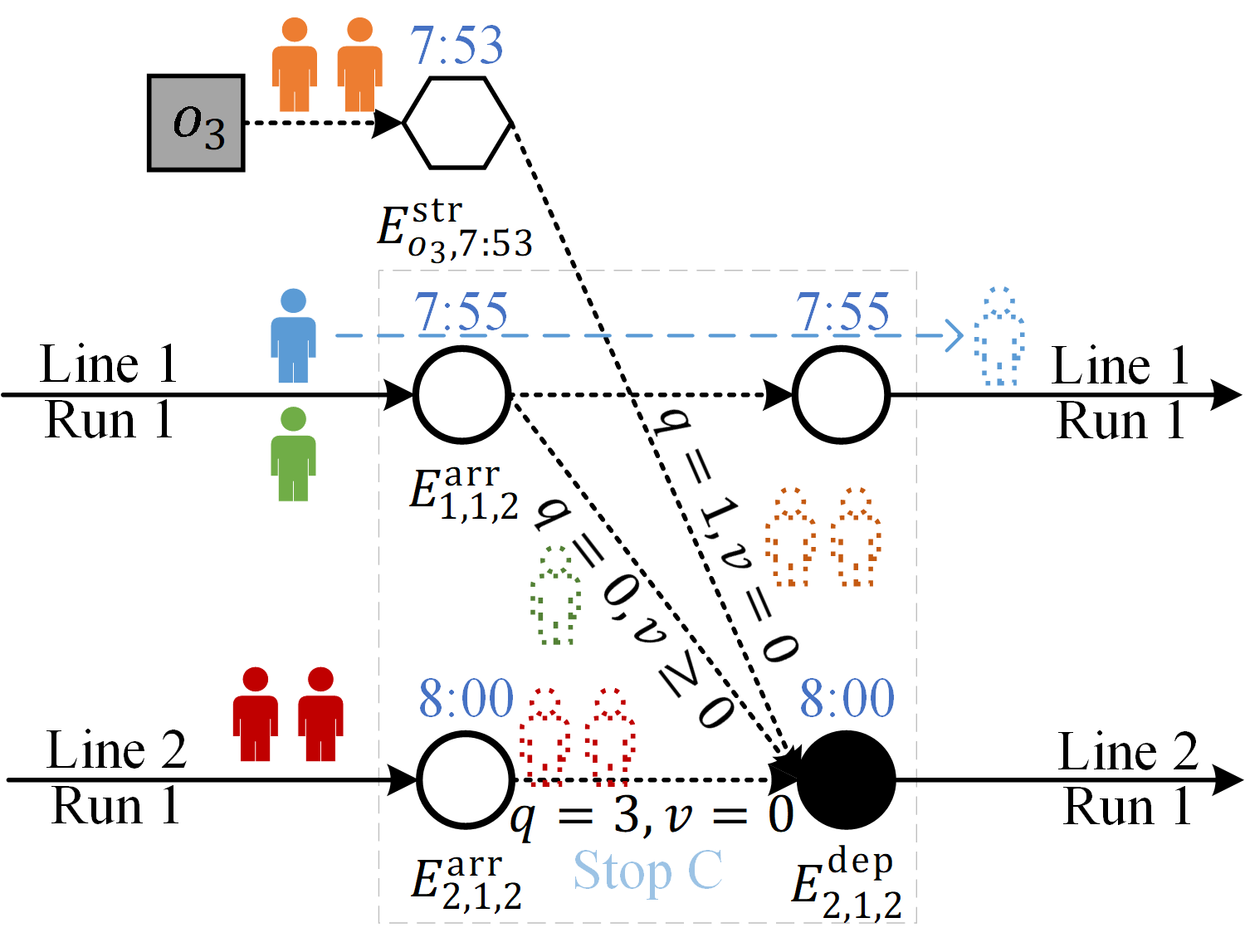}
    \caption{Illustration of the introduced variable $\vv$.}
    \label{fig:explain-multiplier}
\end{figure}
the available capacity of the dwelling arc $(E_{2,1,2}^{\text{arr}}, E_{2,1,2}^{\text{dep}})$ and the boarding arc $(E_{o_3,7:53}^{\text{str}}, E_{2,1,2}^{\text{dep}})$ are positive, meaning that passengers choosing these arcs need not worry about failing to board; hence the associated anxiety cost is zero. In contrast, the available capacities of the lowest-priority transfer arc $(E_{1,1,2}^{\text{arr}}, E_{2,1,2}^{\text{dep}})$ are zero, implying that passengers choosing this arc are not assured of boarding. Consequently, they may incur a positive anxiety cost. For every OD pair $w\in \gW$ and every $b\in B_w$, if we interpret the sum of the travel cost $c_{w, b}^r$ and the anxiety cost $\sum_{A\in \gA_{w,b,r}^{\text{priority}}}v_A$ as a generalized cost of route $r$, then the equilibrium obtained from NCP~\eqref{prob:cp} can be viewed as a state in which all passengers of the same OD pair and class experience the same generalized cost.
    
\end{remark}

Compared with the equivalent UEIP formulation \eqref{prob:cp-route} introduced in the previous section, the refined model \eqref{prob:cp} not only excludes behaviorally unrealistic solutions, but also offers computational advantages. In the equivalent formulation, the priority-enforcing constraint \eqref{equ:cp-route3} is imposed at the \emph{route} level, leading to strongly coupled constraints that are difficult to decompose. As a result, all feasible routes must be enumerated in advance as input to the solution algorithm, which is computationally expensive. In contrast, in the refined model, the priority-enforcing constraint \eqref{equ:cp3} is imposed at the \emph{arc} level, yielding a decomposable network-flow structure that is amenable to column-generation algorithms, which dynamically generate routes during the solution process.

\subsection{Solution algorithms}

\label{sec:algorithms}
To obtain the R-UEIP solution, we propose an algorithm that reformulates NCP~\eqref{prob:cp} into an MPEC.
A straightforward way to handle such complementarity systems is to apply a merit or smoothing function to all complementarity conditions and solve the resulting unconstrained optimization.
However, the form of conditions~\eqref{equ:cp1} and~\eqref{equ:cp2} resembles a classical static traffic assignment problem \citep{beckmann1956studies}.
This structural similarity allows the flows to be computed efficiently using existing assignment algorithms such as TAPAS~\citep{bar-gera_traffic_2010} and iGP~\citep{xie_greedy_2018}, and the column-generation step for dynamically constructing the route set can likewise be delegated to these algorithms.
To preserve this computational advantage, we keep \eqref{equ:cp1} and \eqref{equ:cp2} in their original form as the lower-level equilibrium constraints, and apply a merit function only to condition~\eqref{equ:cp3}, which forms the upper-level optimization.

We use the Fischer-Burmeister function \citep{fischer1992special} to encode the complementarity condition~\eqref{equ:cp3}. 
For a scalar pair $(a,b)$, the FB function is defined as
$
    \varphi(a,b) = \sqrt{a^2 + b^2} - (a + b),
$
which satisfies $\varphi(a,b)=0$ if and only if $0 \le a \perp b \ge 0$. 
Applying this to the introduced variable $\vv$ and the available capacity $\vq(\vx)$, we define 
$
    \bm{\varphi}(\vv,\vx) = \big( \varphi(v_A, q_A(\vx)) \big)_{A \in \gA^{\text{priority}}},
$
and the associated merit function
\begin{align}
    \Psi(\vv,\vx)
    = \|\bm{\varphi}(\vv,\vx)\|^2
    = \sum_{A \in \gA^{\text{priority}}} \varphi(v_A, q_A(\vx))^2. \notag
\end{align}
Then $\vv^*$ satisfies~\eqref{equ:cp3} if and only if it minimizes $\Psi(\vv,\vx)$, where the available capacities $\vq$ depend on $\vv$ indirectly through the arc flow $\vx$.

The mapping from $\vv$ to $\vx$ is defined by~\eqref{equ:cp1}–\eqref{equ:cp2}. 
For every OD pair $w\in \gW$, class $b\in \gB_w$, and route $r\in \gR_w$, let 
\begin{align}
    \hat{c}_{w, b}^r = c_{w, b}^r + \sum_{A\in \gA_{w,b,r}^{\text{priority}}}v_A \notag
\end{align}
be the generalized route cost, and collect these into $\hat{\vc}$. 
Using $\hat{\vc}$, conditions~\eqref{equ:cp1}–\eqref{equ:cp2} can be rewritten as the equivalent equilibrium system
\begin{subequations}
\label{prob:reformulate-lower}
    \begin{align}
        &0\leq \vf \ \bot\  \hat{\vc}(\vx,\vv)-\Lambda\vmu \geq 0, \label{equ:reformulate lower1} \\
        &0\leq \vmu \ \bot\  \Lambda^T\vf-\vd \geq 0. \label{equ:reformulate lower2} 
    \end{align}
\end{subequations}
This lower-level problem has the same structure as a classical static traffic assignment model, with the usual route travel costs replaced by the generalized costs $\hat{\vc}$. 

Combining the merit function $\Psi(\vv,\vx)$ with the equilibrium conditions~\eqref{prob:reformulate-lower}, NCP~\eqref{prob:cp} can be reformulated as the following MPEC:
\begin{align}
    \min_{\vv,\vx,\vmu} \ \Psi(\vv,\vx) 
    \quad \text{subject to} \quad \vv\geq \bm{0} \ \text{and} \  \text{Condition}~ \eqref{prob:reformulate-lower}. 
 \label{prob:reformulate}
\end{align}
Removing the non-negativity constraint on $\vv$ does not affect the theoretical equivalence to NCP~\eqref{prob:cp}. 
However, allowing negative components of $\vv$ would lead to negative generalized costs, violating standard assumptions in traffic assignment and complicating procedures such as shortest-path search. 
Therefore, the constraint $\vv\geq \bm{0}$ is retained in the reformulated model.

When the lower-level equilibrium problem \eqref{prob:reformulate-lower} admits a unique arc flow solution, which is typically ensured when the generalized costs are strictly increasing functions of the arc flows \citep{patriksson2015traffic}, the mapping from $\vv$ to $\vx$ is well defined. In this case, MPEC~\eqref{prob:reformulate} can be solved by an implicit approach. Specifically, we substitute the implicit function $\vx(\vv)$ into the objective of the MPEC, thereby eliminating the equilibrium constraints, and then apply a projected gradient algorithm to the resulting problem. The projection step only needs to enforce the nonnegativity constraint on $\vv$, and can therefore be carried out very efficiently. 
In contrast, when the lower-level equilibrium problem \eqref{prob:reformulate-lower} does not have a unique arc flow solution, we reformulate the equilibrium conditions in MPEC~\eqref{prob:reformulate} as a set of inequality constraints and solve the resulting model using a general nonlinear programming algorithm, such as sequential quadratic programming (SQP). The detailed algorithmic procedures for these two solution approaches --- an implicit method and a nonlinear-programming-based method --- are provided in Appendix~\ref{sec:appendix algorithm}.

\section{Numerical Examples}
\label{sec:experiments}

This section presents three numerical experiments for solving NCP \eqref{prob:cp} using the proposed method. The first experiment, based on a case study of student commuting at the University of Hong Kong, illustrates the method's ability to accurately capture priority rules and compares its results with those of an explicit priority model. The second experiment, using the benchmark network introduced by \citet{nguyen_modeling_2001}, demonstrates that our method can accurately recover the R-UEIP. Finally, the third experiment evaluates the computational efficiency of the proposed algorithm on the Sioux Falls transit network. All results presented in this section are obtained on a Windows 10 (64-bit) PC equipped with an AMD Ryzen 7 4800H 2.90~GHz CPU and 16~GB of RAM.

\medskip
\noindent
\textbf{Cost structure.} We first specify the cost functions used in the numerical experiments. All cost functions are taken from the existing literature, and the costs associated with each type of arc are assumed to be identical for all passengers, regardless of their OD pair or class.  Specifically, for each boarding arc $A = (E_{o,t}^{\text{str}}, E_{l,j,i}^{\text{dep}})$, the cost is 
$$
\eta_1 \cdot (\tau_{l,j,i}^{\text{dep}} - t),
$$
which is the time duration multiplied by $\eta_1$, the value of time. Similarly, for each transfer arc $A = (E_{l,j,i}^{\text{arr}}, E_{l',j',i'}^{\text{dep}})$, the cost is
$$
\eta_1 \cdot (\tau_{l',j',i'}^{\text{dep}} - \tau_{l,j,i}^{\text{arr}}).
$$
For riding arcs $A = (E_{l,j,i}^{\text{dep}}, E_{l,j,i+1}^{\text{arr}})$, the cost is
$$
\eta_1 \cdot (\tau_{l,j,i+1}^{\text{arr}} - \tau_{l,j,i}^{\text{dep}}) + \pi_l^{i,i+1} + \eta_2 \cdot \max\left\{0, \frac{x_{A}}{u_{A}}-\rho\right\},
$$
where $\pi_l^{i,i+1}$ is the fare or penalty associated with traveling from stop $s_l^i$ to $s_l^{i+1}$ on line $l$, $\eta_2 > 0$ is the marginal disutility of crowding and $\rho \in [0,1]$ is the crowding perception threshold \citep{hamdouch_schedule-based_2008}. Here, the congestion term becomes positive only when the load ratio $x_A / u_A$ exceeds the threshold~$\rho$.
For the dwelling arcs $A = (E_{l,j,i}^{\text{arr}}, E_{l,j,i}^{\text{dep}})$, the cost is defined similarly as
$$
\eta_1 \cdot (\tau_{l,j,i}^{\text{dep}} - \tau_{l,j,i}^{\text{arr}}) + \eta_2 \cdot \max\left\{0, \frac{x_{A}}{u_{A}}-\rho\right\}.
$$
The costs on the other two types of arcs --- access arcs $\gA^{\text{access}}$ and egress arcs $\gA^{\text{egress}}$ --- depend on the OD pair $w \in \gW$ and class $b \in \gB_w$.  For each egress arc $A = (E_{l,j,i}^{\text{arr}}, E_{d}^{\text{dst}})$, the cost consists of the walking time from stop $s_l^i$ to destination $d$, together with early-arrival and late-arrival penalties. Formally, the cost is
$$
t_{s_l^i, d}^{\text{walk}} + 
\begin{cases}
\eta_3 \cdot (\tau_{w, b}^{-} - (\tau_{l,j,i}^{\text{arr}} + t_{s_l^i, d}^{\text{walk}})), \quad &\text{if}~\tau_{l,j,i}^{\text{arr}} + t_{s_l^i, d}^{\text{walk}} < \tau_{w, b}^{-}, \\
0, \quad &\text{if}~\tau_{l,j,i}^{\text{arr}} + t_{s_l^i, d}^{\text{walk}} \in [\tau_{w, b}^{-}, \tau_{w, b}^{+}], \\
\eta_4 \cdot (\tau_{l,j,i}^{\text{arr}} + t_{s_l^i, d}^{\text{walk}} - \tau_{w, b}^{+}), \quad &\text{if}~\tau_{l,j,i}^{\text{arr}} + t_{s_l^i, d}^{\text{walk}} > \tau_{w, b}^{+},
\end{cases}
$$
where $\eta_3 > 0$ and $\eta_4 > 0$ are coefficients capturing passengers' aversion to early and late arrival, respectively. Eventually, for access arc $A = (E_{o}^{\text{ogn}}, E_{o,t}^{\text{str}})$, the cost is 
$$
\eta_5 \cdot \max \{\tau_{w, b}^{\text{free}} - t, 0 \},
$$
where $\tau_{w, b}^{\mathrm{free}}$ denotes the \emph{free-flow latest departure time} \citep{nguyen_modeling_2001}, i.e., the latest time at which a passenger of OD pair $w$ and class $b$ can depart and still arrive before $\tau_{w, b}^{+}$ when no capacity-induced delays occur. Passengers departing earlier than this time incur a penalty. The route cost $c_{w,b}^r(\vf)$ is the total arc travel cost experienced by passengers of OD pair $w$, class $b$, and follow route $r \in \gR_{w,b}$.


\subsection{Insight from a student commuting case in HKU}
This section presents a case study of morning-peak commuting by students at the University of Hong Kong to demonstrate that the proposed transit assignment model can accurately capture passengers' choice and queuing behaviors. As illustrated in Figure~\ref{fig:commuting net}, passengers depart from the Jockey Club Student Village~IV and travel to the main campus, with a desired arrival time interval of [9:00, 9:20]; being late would result in missing classes and incurring a substantial penalty. Two categories of routes are available. 
\begin{figure}[htbp]
	\centering
	\includegraphics[width=0.95\textwidth]{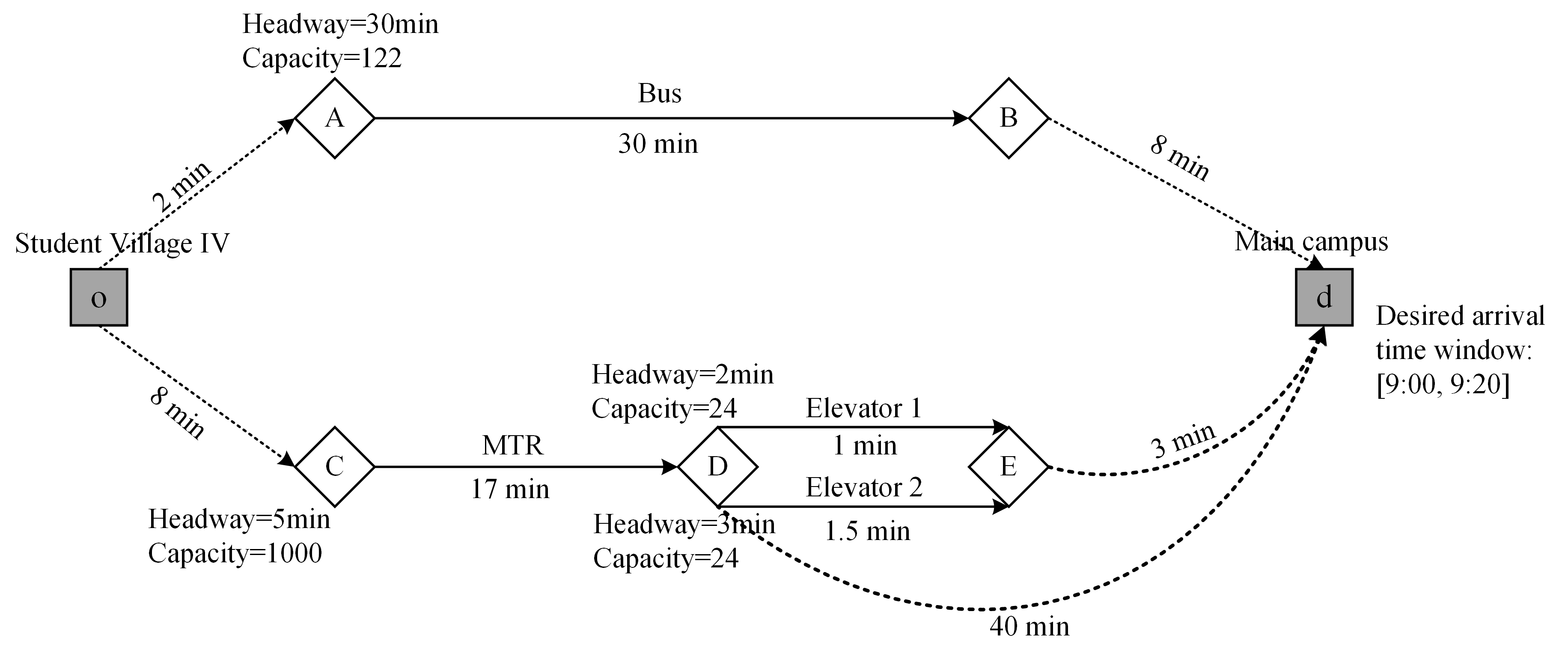}
	\caption{The example network of student commuting cases.}
	\label{fig:commuting net}
\end{figure}
One option is to take Bus~No.~71, which has a relatively long in-vehicle riding time of about 30~minutes. The other option is to take the \textit{Mass Transit Railway} (MTR), which requires only about 17~minutes of in-vehicle riding time but depends on elevators with limited capacity to access the campus from the station. In practice, passengers often face elevator queues lasting from ten to several tens of minutes. This situation gives rise to two clear trade-offs: (i) whether to take the slower but more reliable bus or the faster MTR at the risk of long elevator queues, and (ii) for those choosing the MTR, whether to depart earlier to avoid being late or depart later to reduce waiting time in the queue.
To illustrate these two trade-offs, Section~\ref{sec:commuting case tradeoff} analyzes the assignment results produced by the proposed model, and Section~\ref{sec:commuting case comparison} compares these results with those obtained from the explicit priority model (see Appendix~\ref{sec:appendix dynamic}).

During peak hours, the elevators operate at full capacity with fixed stop patterns and can thus be modeled as scheduled transit lines. Two elevators are considered in this case: Elevator~1 (e1) runs directly from the MTR station to the hilltop campus, with a travel time of 1~minute and a headway of 2~minutes; Elevator~2 (e2) provides a non-direct service with a travel time of 1.5~minutes and a headway of 3~minutes. Both elevators have a capacity of 24 passengers. Alternatively, passengers may choose to walk from the MTR station to the campus, which takes 25 minutes with an additional fatigue-equivalent penalty of 15 minutes, resulting in an effective travel time of 40 minutes.
Bus~No.~71 operates with an Enviro500MMC~11.3~m vehicle of 122-passenger capacity, while the MTR line is served by Metro-Cammell~EMU trains with a total capacity of 2504~passengers. After accounting for a background load of 1504 passengers, the feasible capacity of MTR is set to 1000. After applying the student concession, the fares for the MTR and bus are 4.4 HKD and 6.7 HKD, respectively. The parameters of the cost function are set as $\eta_1=0.55$, $\eta_2=5.0$, $\eta_3=0.2$, $\eta_4=10.0$, $\eta_5=1.0$, and $\rho=0.8$. This case is solved using the implicit method described in Appendix \ref{sec:implicit-algorithm}, with the algorithm implemented in C++.

\subsubsection{Analysis of students' travel choices under the proposed model}
\label{sec:commuting case tradeoff}
Table~\ref{tab:commuting case flow result} presents the assignment results between the bus and MTR under different total demand levels. Because the bus operates infrequently (one departure every 30~minutes), only the 8:40~a.m. trip is relevant to passengers in this scenario. The minimum travel cost of this bus trip (excluding crowding penalties and extra waiting) is~37.70. When the total demand is~200, 300, or~400, the maximum travel cost among MTR passengers is~30.35, 32.55, and~37.13, respectively, all lower than~37.70. Hence, no passengers chose the bus route. 
As the demand increases to~500, some MTR passengers must start their trips earlier to avoid being late. At this point, the maximum MTR travel cost nearly equals the bus cost, resulting in~93.39~passengers shifting to the bus.
When the demand rises to~600 and~700, the maximum MTR travel cost increases further, the bus reaches full capacity (bus flow = 122.00), and an \textit{early-departure queuing} phenomenon emerges: to secure higher boarding priority for the limited bus capacity, passengers arrive at stop~A earlier than necessary. The extent of this departure advancement depends on the perceived benefit of switching --- specifically, the difference between the maximum MTR travel cost and the bus travel cost without early departure. For total demand levels of~600 and~700, the departure advancement is~2~minutes and~5~minutes, respectively, leading to approximate cost equalization between the two modes (with minor discrepancies caused by the discrete nature of timetables and departure-time choices).

\begin{table}[htbp]
  \centering
  \caption{Assignment results between bus and MTR under different total demand levels.}
  \footnotesize %
    \begin{tabular}{ccccccc}
    \toprule
    Demand & 200   & 300   & 400   & 500   & 600   & 700 \\
    \midrule
    Max. MTR travel cost & 30.35 & 32.55 & 37.13 & 37.72 & 41.85 & 47.56 \\
    Bus travel cost & 37.70  & 37.70  & 37.70  & 37.70  & 41.85 & 46.45 \\
    Bus flow & 0.00     & 0.00     & 0.00     & 93.39 & 122.00   & 122.00 \\
    Departure advancement & 0 min   & 0 min   & 0 min    & 0 min   & 2 min     & 5 min \\
    \bottomrule
    \end{tabular}%
  \label{tab:commuting case flow result}%
\end{table}%

We further analyze the trade-off between late arrival and early trip start among passengers choosing the MTR. Whether a passenger arrives late is determined by the departure time of the elevator they take: if the elevator departs later than 9:16~a.m., the passenger will be late. Figure~\ref{fig:commuting case flow} illustrates the distribution of passengers across different MTR trains and elevator runs under various demand levels. The horizontal axis represents the elevator departure time, the left vertical axis denotes the arrival time of MTR trains at station~D, and the right vertical axis indicates the elevator run taken. Time is expressed in numerical format (e.g., 9:16~a.m. is represented as 76).
\begin{figure}[htbp]
	\centering
	\subfigure[Demand = 200]{
	{\includegraphics[width=0.48\textwidth]{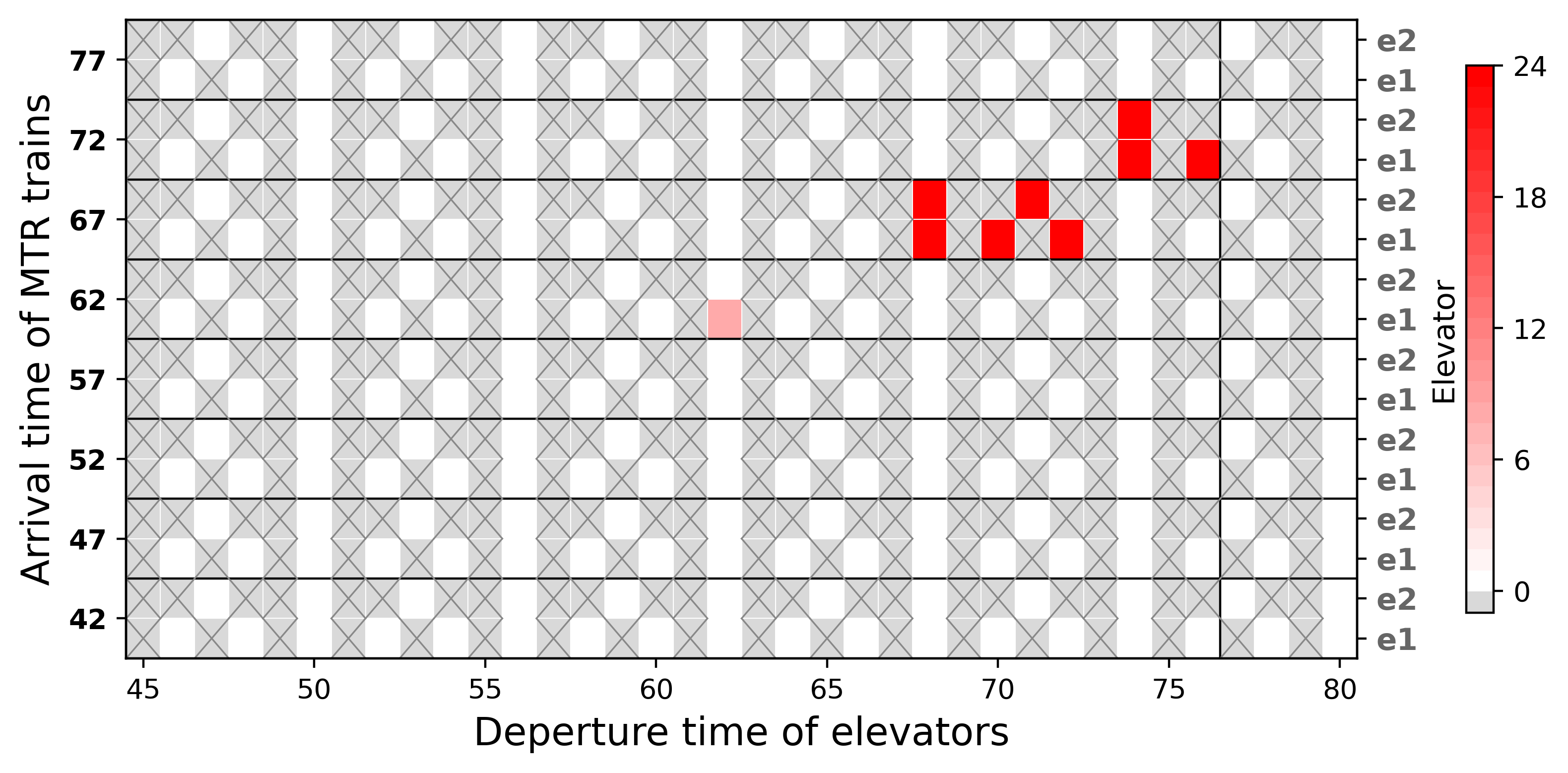}}\label{fig:commuting case flow1}
	}
	\subfigure[Demand = 300]{
	\includegraphics[width=0.48\textwidth]{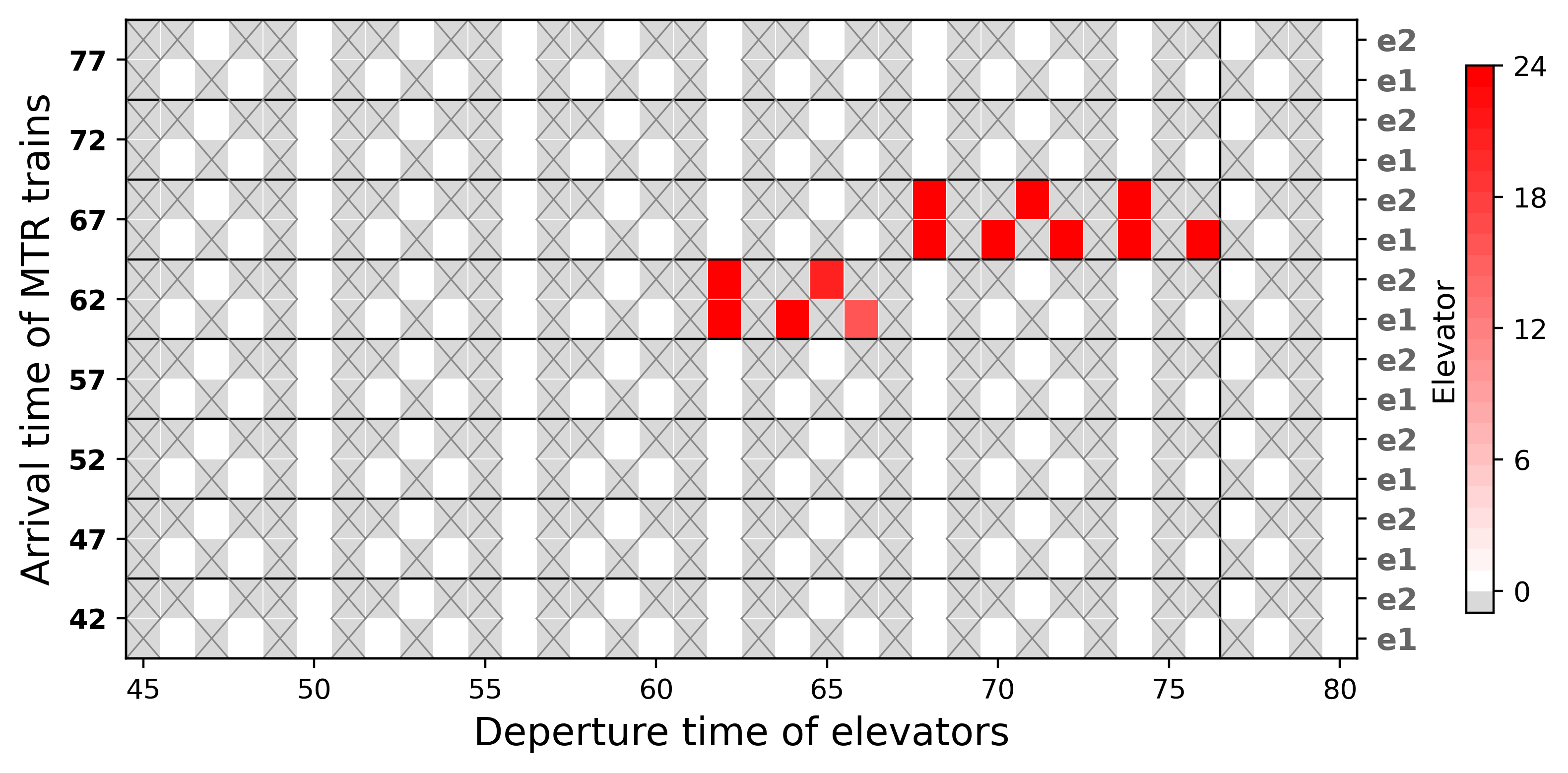}\label{fig:commuting case flow2}
	}
    \subfigure[Demand = 400]{
	\includegraphics[width=0.48\textwidth]{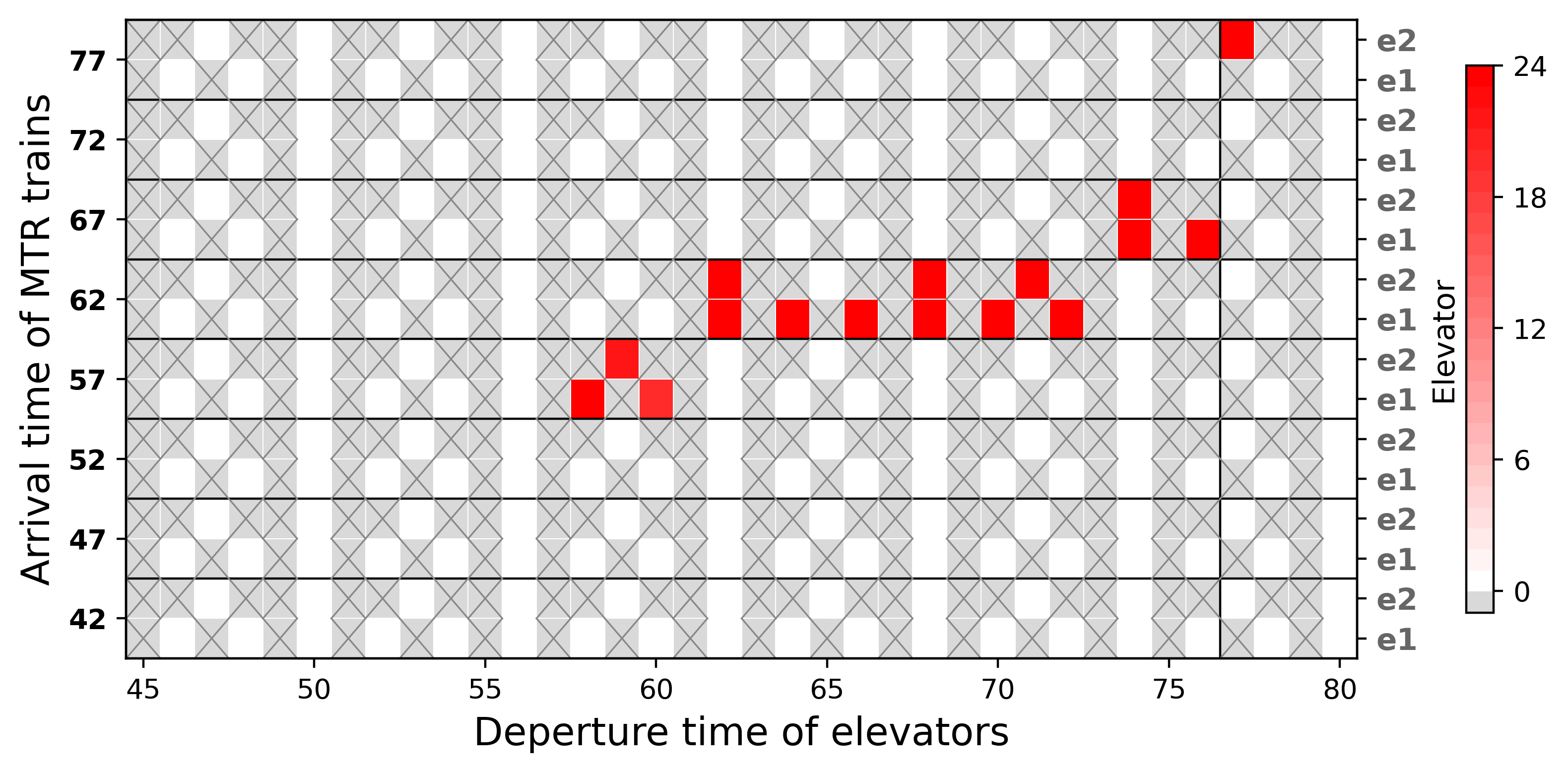}\label{fig:commuting case flow3}
	}
    \subfigure[Demand = 500]{
	\includegraphics[width=0.48\textwidth]{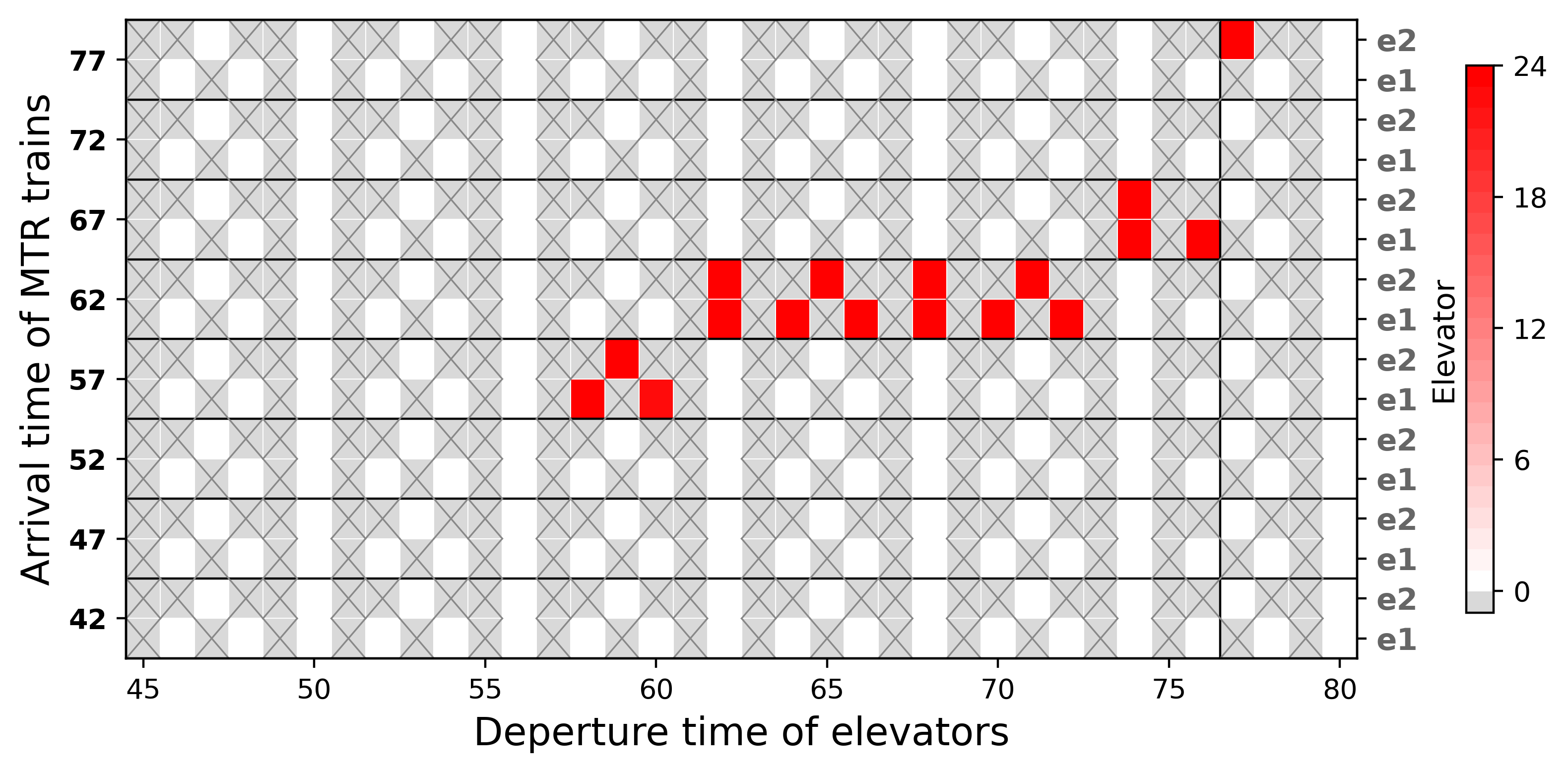}\label{fig:commuting case flow4}
	}
    \subfigure[Demand = 600]{
	\includegraphics[width=0.48\textwidth]{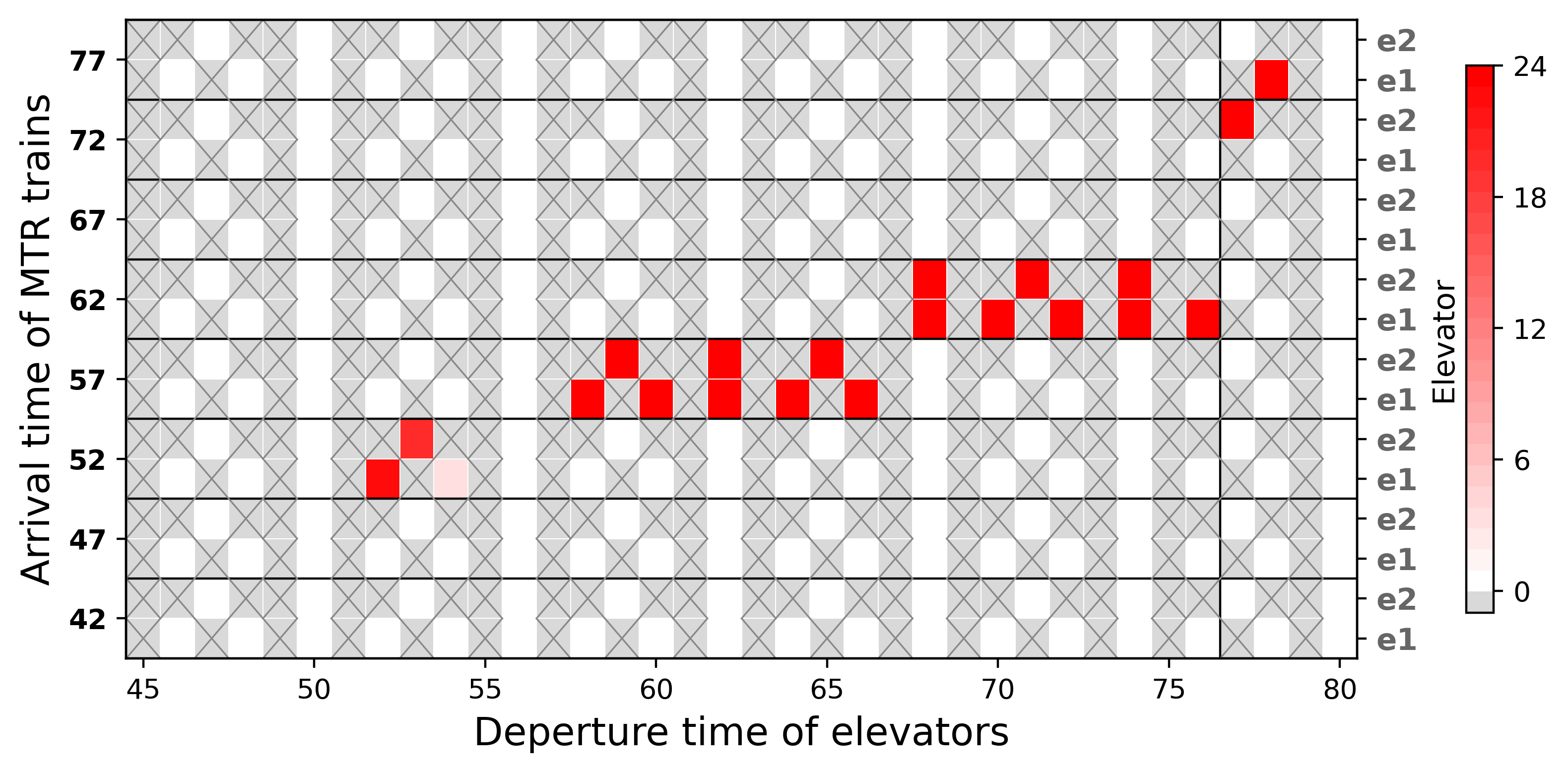}\label{fig:commuting case flow5}
	}
    \subfigure[Demand = 700]{
	\includegraphics[width=0.48\textwidth]{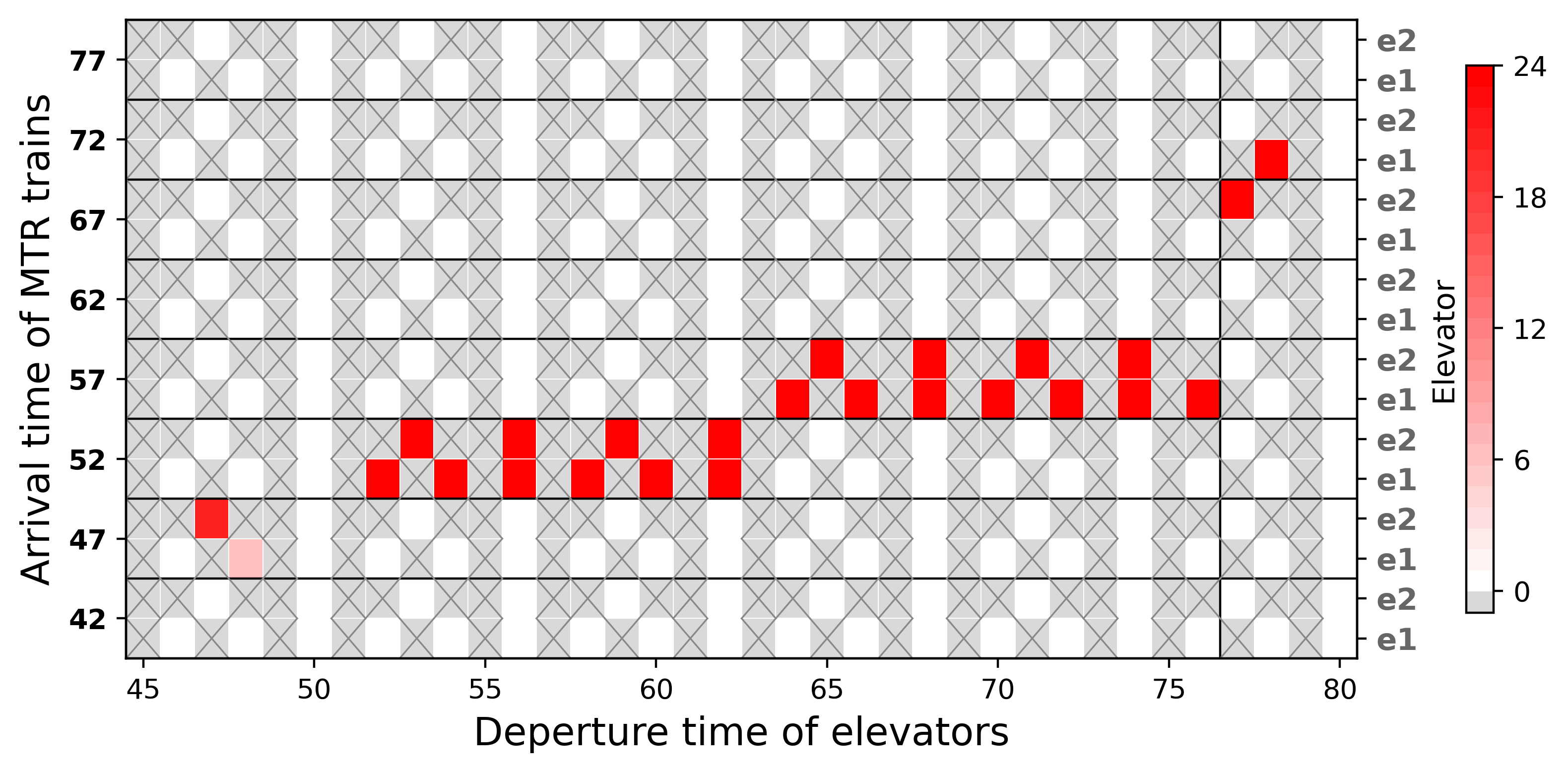}\label{fig:commuting case flow6}
	}
	\caption{Passenger distribution across MTR and elevator runs under different total demand levels. Gray crossed cells denote no elevator departure, while the shade of red indicates the passenger count on a specific MTR–elevator combination.}
	\label{fig:commuting case flow}
\end{figure}

When the total demand is~200 or~300, no passengers arrive late (Figures~\ref{fig:commuting case flow1} and~\ref{fig:commuting case flow2}). This is because the minimum possible travel cost for a late-arriving passenger is~35.63, which exceeds the maximum MTR travel cost under these two demand levels. As demand continues to increase, newly added passengers must start their trips earlier to avoid arriving late. When the demand reaches~400, the maximum travel cost among on-time passengers rises to~37.13 --- higher than~35.63 --- indicating that some passengers now prefer to start later and accept a late-arrival penalty. As shown in Figure~\ref{fig:commuting case flow3}, 24~passengers make this choice. With further increases in total demand to~500, 600, and~700, the severity of lateness grows, and the number of late passengers increases to~24, 48, and~48, respectively (Figures~\ref{fig:commuting case flow4} –- \ref{fig:commuting case flow6}).

In addition to increasing the number of late passengers, higher demand also leads to longer elevator waiting times. For instance, among passengers taking the MTR train arriving at 9:02 a.m. (time~62), no queue forms when the demand is~200. When the demand rises to~300 and~400, the maximum elevator waiting time extends to~4~and~10~minutes, respectively. At a demand of~600, passengers experience at least~6~minutes and up to~14~minutes of queuing before boarding the elevator. Such severe congestion patterns are highly consistent with real-world observations, demonstrating the capability of the proposed model to realistically capture queuing phenomena.

\subsubsection{Comparison with the explicit priority model}
\label{sec:commuting case comparison}
For comparison, we solved an explicit priority model for this case study using an MSA-based algorithm implemented in C++ (see Appendix~\ref{sec:appendix dynamic} for the model and algorithmic details). When the total demand reaches 600 or 700, the relative gap of the explicit priority model fails to converge below $10^{-2}$, rendering the results incomparable. Therefore, in this section, we set the total demand to 500, under which the relative gap converges below $10^{-5}$, indicating that a reasonably accurate equilibrium has been achieved. 

Overall, both the proposed implicit priority model and the explicit priority model yield equilibrium states that comply with the passenger priority rule, yet their flow distribution patterns differ. In the explicit priority model results, 72.21 passengers choose to travel by bus, whereas in the proposed model, as shown in Table~\ref{tab:commuting case flow result}, the bus flow reaches 93.39.
The passenger distributions across MTR trains and the resulting elevator queues are illustrated in Figure~\ref{fig:commuting case compare}. In both results, the elevator queues follow the FCFS principle --- that is, passengers arriving on earlier MTR trains always board elevators before those from later trains (as visually reflected in Figure~\ref{fig:commuting case compare}, where red blocks representing earlier arrivals appear below or to the left of those for later arrivals on the same elevator).

\begin{figure}[htbp]
	\centering
    \subfigure[The proposed implicit priority model]{
	\includegraphics[width=0.48\textwidth]{Figures/elevatorqueue500.png}\label{fig:commuting case compare1}
	}
    \subfigure[The explicit priority model]{
	\includegraphics[width=0.48\textwidth]{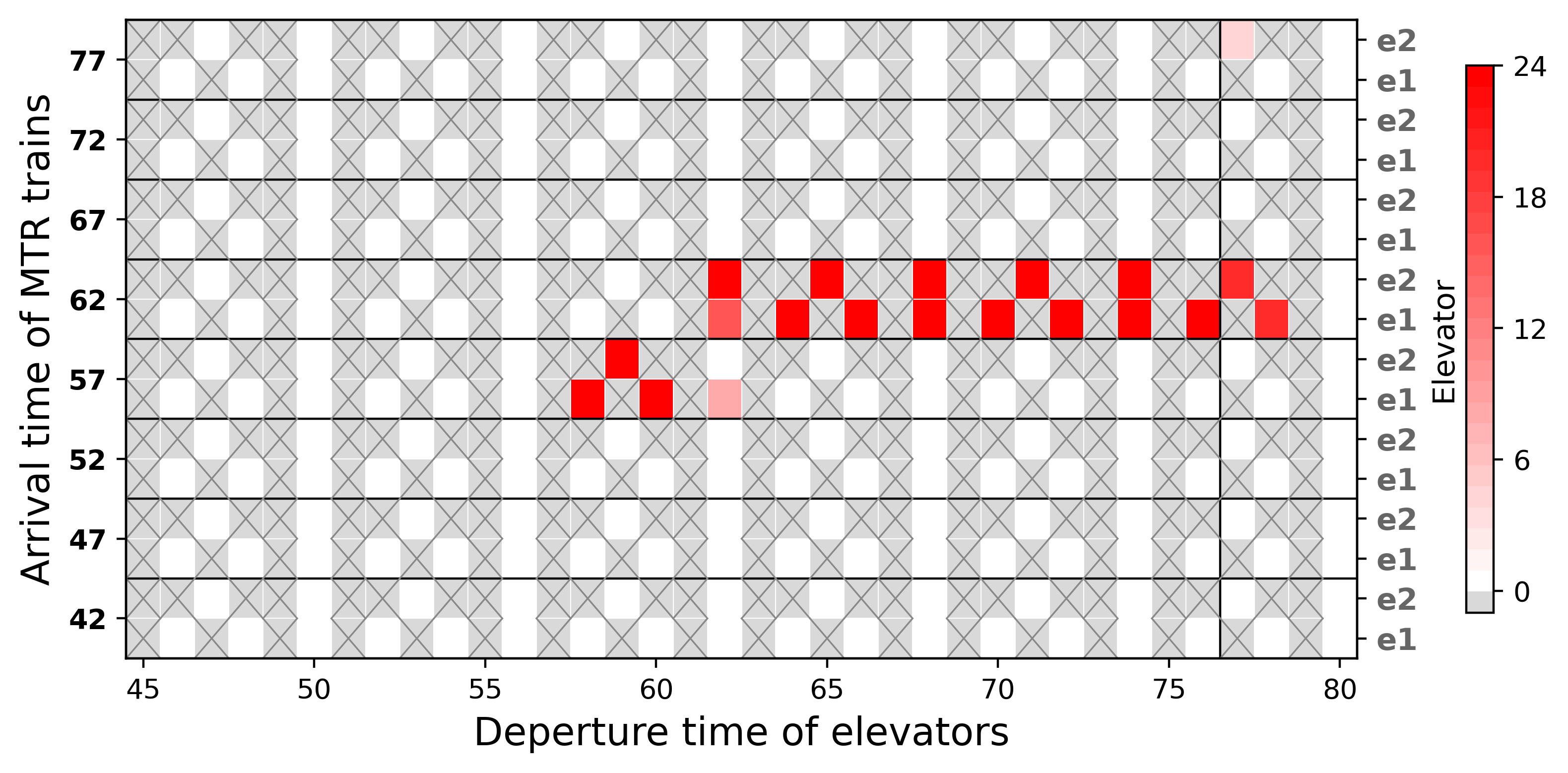}\label{fig:commuting case compare2}
	}
	\caption{Passenger distributions across MTR trains and elevator queues under the proposed implicit priority model and the explicit priority model.}
	\label{fig:commuting case compare}
\end{figure}

However, the explicit priority model results exhibit one unrealistic phenomenon. Consider passengers arriving at 9:02~a.m. (time~62) and 9:17~a.m. (time~77) and the two elevators departing at 9:17~a.m. (time~77) and 9:18~a.m. (time~78). Some passengers who arrive earlier leave a capacity of~4.61~on elevator~“e2” (departing at~77) for later arrivals, but instead board elevator~“e1” (departing at~78), thereby incurring higher travel costs --- an irrational choice. 
The cause lies in the decision mechanism of the explicit priority model, which is based on expected travel costs. Specifically, passengers arriving at 9:02 a.m. are divided into two transit plans that differ only in whether they take elevator “e1” or “e2”. Although it would reduce the actual travel cost for some passengers to switch from e1 (departing at 9:18 a.m.) to “e2” (departing at 9:17 a.m.), such a transfer does not occur because, from the model's perspective, the two transit plans already have identical expected travel costs, and thus are in equilibrium.
This example highlights a key limitation of the explicit priority model, mentioned earlier: it artificially treats passengers who are not intrinsically connected as a single decision-making entity and enforces equilibrium by assuming they make choices collectively based on expected travel costs. In contrast, our proposed model determines route choice based on actual travel costs, leading to greater behavioral realism and flexibility.

\subsection{Verification accuracy in benchmark network}
The benchmark network, illustrated in Figure~\ref{fig:benchmark net}, comprises six stops and four lines, each with a capacity of 20 passengers. Passengers are divided into four OD pairs $(o_1, d_1)$, $(o_1, d_2)$, $(o_2, d_1)$, and $(o_2, d_2)$, all having a desired arrival time window of [8:45, 9:00]. The demand for each OD is 10. The cost parameters are set as $\eta_1=\eta_5=1.00$ and $\eta_2=\eta_3=\eta_4=0.00$. \citet{nguyen_modeling_2001} obtained a route flow vector $\bm{f}^0$ by solving an approximate problem of the UEIP, as reported in the fourth column of Table \ref{tab:route flow benchmark net}. This solution does not satisfy the UEIP and the R-UEIP conditions, nor does it respect the vehicle capacity constraints. We adopt it as the initial route flow vector. The initial $v$ for all incoming arcs to the two saturated arcs, $(E_{2,1,2}^{\text{dep}}, E_{2,1,3}^{\text{arr}})$ and $(E_{1,1,2}^{\text{dep}}, E_{1,1,3}^{\text{arr}})$, is set to 5.00, while the initial $v$ for all other arcs is set to 0.00. The initial $\mu$ is $\mu_{(o_1, d_1)}^0=65.00, \mu_{(o_1, d_2)}^0=75.00, \mu_{(o_2, d_1)}^0=60.00, \mu_{(o_2, d_2)}^0=75.00$. Here, we solve this example using a nonlinear-programming-based approach described in Appendix \ref{sec:nonlinear-algorithm}, specifically the SQP solver of the \textit{fmincon} function in MATLAB.
\begin{figure}[htbp]
	\centering
	\includegraphics[width=\textwidth]{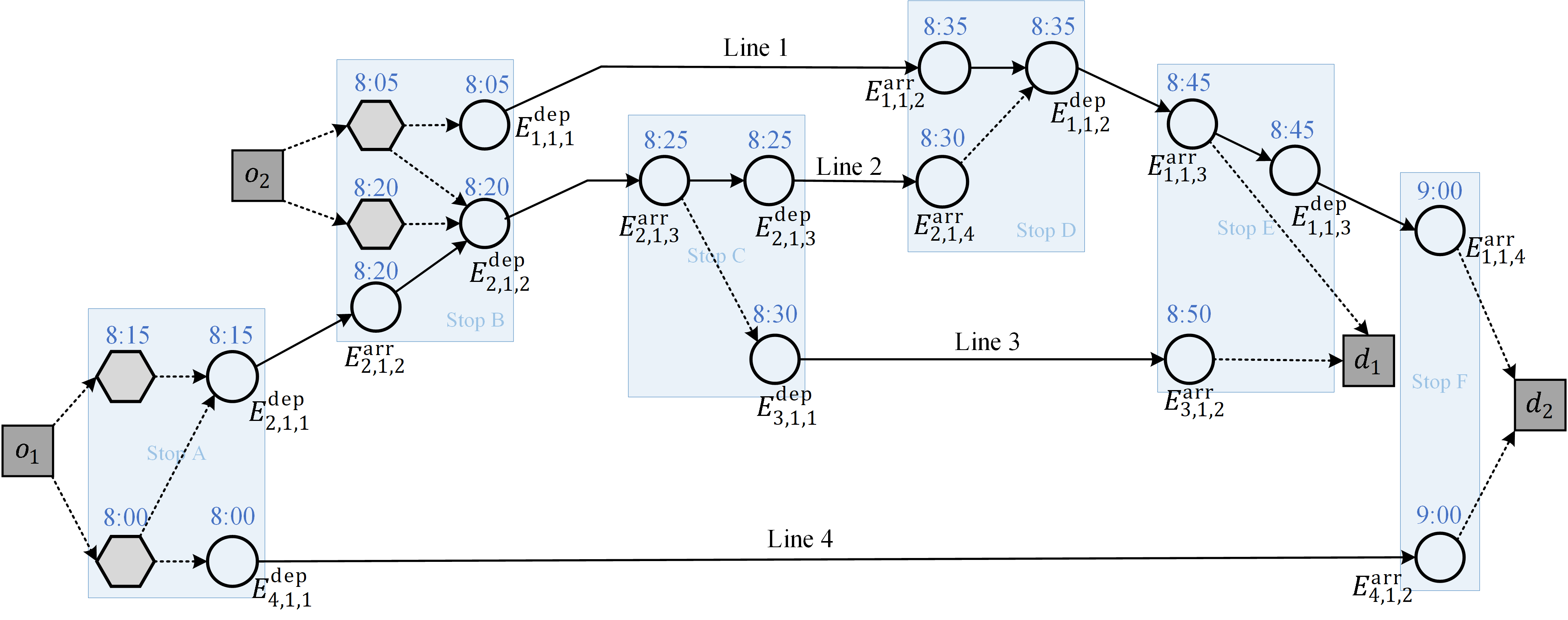}
	\caption{The benchmark network.}
	\label{fig:benchmark net}
\end{figure}

After 55 iterations, the algorithm converges in approximately 0.11 seconds; the final objective value is $\Psi(\vv,\vx)=6.55\times 10^{-6}$. This indicates that a highly accurate global optimum of the MPEC \eqref{prob:reformulate} has been obtained, corresponding to an R-UEIP solution. Columns 5, 6, and 7 of Table \ref{tab:route flow benchmark net} report the optimal route flows $f_{w,b}^{r*}$, the total available capacities of the dominant routes $\hat{Q}_{w,b}^r$, and the travel costs $c_{w,b}^r$, respectively. It can be readily verified that all routes with $\hat{Q}_{w,b}^r>0$ are unused, thereby confirming that the solution satisfies the R-UEIP definition (Definition \ref{def:refine uep}).

\begin{table}[htbp]
  \centering
  \caption{Results of route flow in the benchmark network.}
  \footnotesize %
    \begin{tabular}{cclrrrrr}
    \toprule
    OD pairs & Routes & Route description & \multicolumn{1}{c}{$\vf^0$} & \multicolumn{1}{c}{$\vf^*$} & \multicolumn{1}{c}{$\bm{\hat{Q}}$} & \multicolumn{1}{c}{$\vc$} \\
    \midrule
    \multirow{4}[2]{*}{$(o_1, d_1)$} & 1     & 8:15 - Line 2 - Stop D - Line 1 & 0.000 & 0.000  & 0.000  & 30.000  \\
          & 2     & 8:15 - Line 2 - Stop C - Line 3 & 10.000 & 10.000 & 0.000  & 35.000 \\
          & 3     & 8:00 - Line 2 - Stop D - Line 1 & 0.000 & 0.000  & 10.000  & 60.000 \\
          & 4     & 8:00 - Line 2 - Stop C - Line 3 & 0.000 & 0.000  & 0.00  & 65.000 \\
    \midrule
    \multirow{3}[2]{*}{$(o_1, d_2)$} & 5     & 8:15 - Line 2 - Stop D - Line 1 & 6.236 & 5.000  & 0.000  & 45.000 \\
          & 6     & 8:00 - Line 2 - Stop D - Line 1 & 0.000 & 0.000  & 5.000  & 75.000 \\
          & 7     & 8:00 - Line 4 & 3.764 & 5.000  & 0.000 & 75.000 \\
    \midrule
    \multirow{5}[2]{*}{$(o_2, d_1)$} & 8     & 8:20 - Line 2 - Stop D - Line 1 & 0.000 & 0.000  & 0.000  & 25.000 \\
          & 9     & 8:20 - Line 2 - Stop C - Line 3 & 4.764 & 5.000  & 0.000  & 30.000 \\
          & 10    & 8:05 - Line 2 - Stop D - Line 1 & 0.000 & 0.000  & 0.000  & 55.000 \\
          & 11    & 8:05 - Line 2 - Stop C - Line 3 & 0.000 & 0.000  & 5.000  & 60.000 \\
          & 12    & 8:05 - Line 1 & 5.236 & 5.000  & 0.000  & 55.000 \\
    \midrule
    \multirow{3}[2]{*}{$(o_2, d_2)$} & 13    & 8:20 - Line 2 - Stop D - Line 1 & 0.000 & 0.000  & 0.000  & 40.000 \\
          & 14    & 8:05 - Line 2 - Stop D - Line 1 & 0.000 & 0.000  & 0.000  & 70.000 \\
          & 15    & 8:05 - Line 1 & 10.000 & 10.000 & 0.000  & 70.000 \\
    \bottomrule
    \end{tabular}%
  \label{tab:route flow benchmark net}%
\end{table}%

\subsection{Verification computational efficiency in the Sioux Falls network}

The Sioux Falls transit network, adapted from \citet{szeto_transit_2014}, consists of 24 stops, 10 lines, and 16 OD pairs. The demand for each OD is listed in Table~\ref{tab:demand}. Passengers share a desired arrival time window of [60, 90]. As the original network specification provides only the headways of each line, the schedule is constructed by assigning the first departure of each line at time 0 according to the given headways. In total, the 10 lines generate 113 runs, and the resulting expanded event-activity graph comprises 1,525 nodes and 12,951 arcs. The cost function parameters are set as $\eta_1=1.0$, $\eta_2=2.0$, $\eta_3=\eta_4=\eta_5=1.2$, and $\rho=0.0$. In this example, the solution is obtained via the implicit method.

\begin{table}[htbp]
    \centering
    \caption{Demand data of the Sioux Falls transit network.}
    \footnotesize %
    \begin{tabular}{cccccccc}
        \toprule
        OD & Demand & OD & Demand & OD & Demand & OD & Demand \\
        \midrule
        (1,13) & 200 & (2,13) & 100 & (3,13) & 200 & (4,13) & 200 \\
        (1,20) & 100 & (2,20) & 100 & (3,20) & 100 & (4,20) & 100 \\
        (1,21) & 100 & (2,21) & 100 & (3,21) & 100 & (4,21) & 100 \\
        (1,24) & 200 & (2,24) & 100 & (3,24) & 200 & (4,24) & 200 \\
        \bottomrule
    \end{tabular}
    \label{tab:demand}
\end{table}

Table~\ref{tab:num of ite} reports the number of iterations and the CPU time required by Algorithm~\ref{alg:IM} for the merit function $\Psi(\vv,\vx)$ to fall below the predefined thresholds. The initial objective function value was $1.02\times 10^7$. After 3560 iterations (approximately 22.45 minutes), the value dropped below 10, and after 8276 iterations (55.31 minutes), it further decreased to below 1.

\begin{table}[htbp]
    \centering
    \caption{The number of iterations and CPU time required to achieve a predetermined convergence threshold.}
    \footnotesize %
    \begin{tabular*}{\textwidth}{@{}@{\extracolsep{\fill}}cccccccc@{}}
        \toprule
        Threshold & $10^4$ & $10^3$ & $10^2$ & $10$ & $1$\\
        \midrule
        Iterations & 34 & 396 & 1116 & 3560 & 8276\\
        CPU time & 15.22 sec & 2.51 min & 6.80 min & 22.45 min & 55.31 min\\
        \bottomrule
    \end{tabular*}
    \label{tab:num of ite}
\end{table}

To more intuitively illustrate the convergence accuracy, Figure~\ref{fig:q and cost diff} presents the scatter plots of the available capacity $Q_{w,b}^{r',r}$ and the cost difference $c_{w,b}^r - c_{w,b}^{r'}$ for all used routes $r$ and their dominated counterparts $r'$, when $\Psi(\vv,\vx)$ falls below 10 and 1, respectively.
According to the definition of the R-UEIP (Definition~\ref{def:refine uep}), for any route pair $(r, r')$, at least one of these two quantities --- available capacity or cost difference --- must be zero. As shown in Figure~\ref{fig:q and cost diff}, when $\Psi(\vv,\vx) < 10$, nearly all points lie close to either the $x$-axis or the $y$-axis. The most deviated point has an available capacity of 111.77 and a cost difference of 2.69. When $\Psi(\vv,\vx) < 1$, the convergence precision further improves, as all points lie even closer to the coordinate axes. For the most deviated pair, the available capacity is 104.68, and the cost difference is 0.97. These results indicate that the solution has achieved a satisfactory level of convergence accuracy.

\begin{figure}[htbp]
	\centering
    \subfigure[$\Psi<10$]{
	\includegraphics[width=0.48\textwidth]{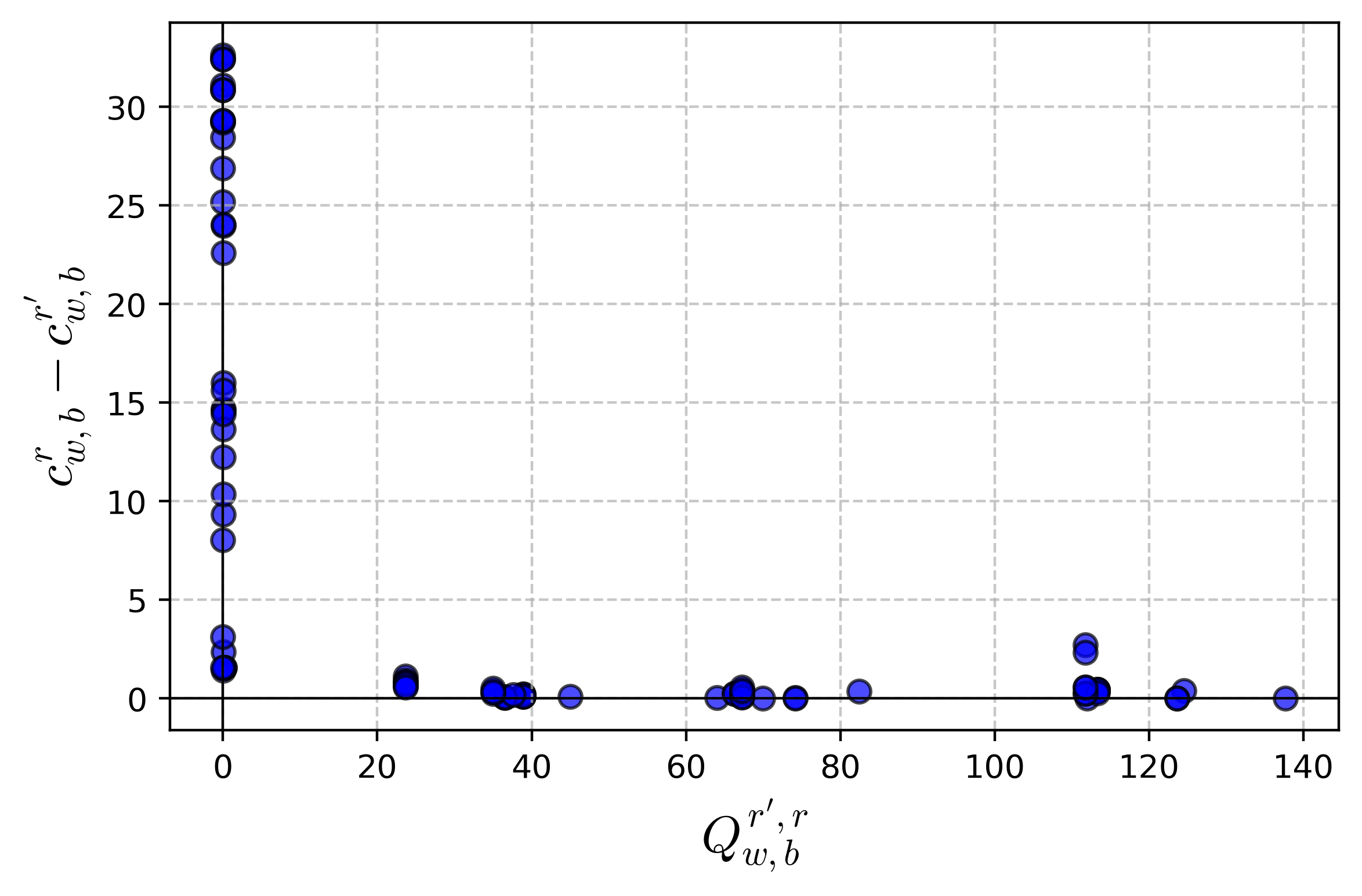}\label{fig:q and cost diff1}
	}
    \subfigure[$\Psi<1$]{
	\includegraphics[width=0.48\textwidth]{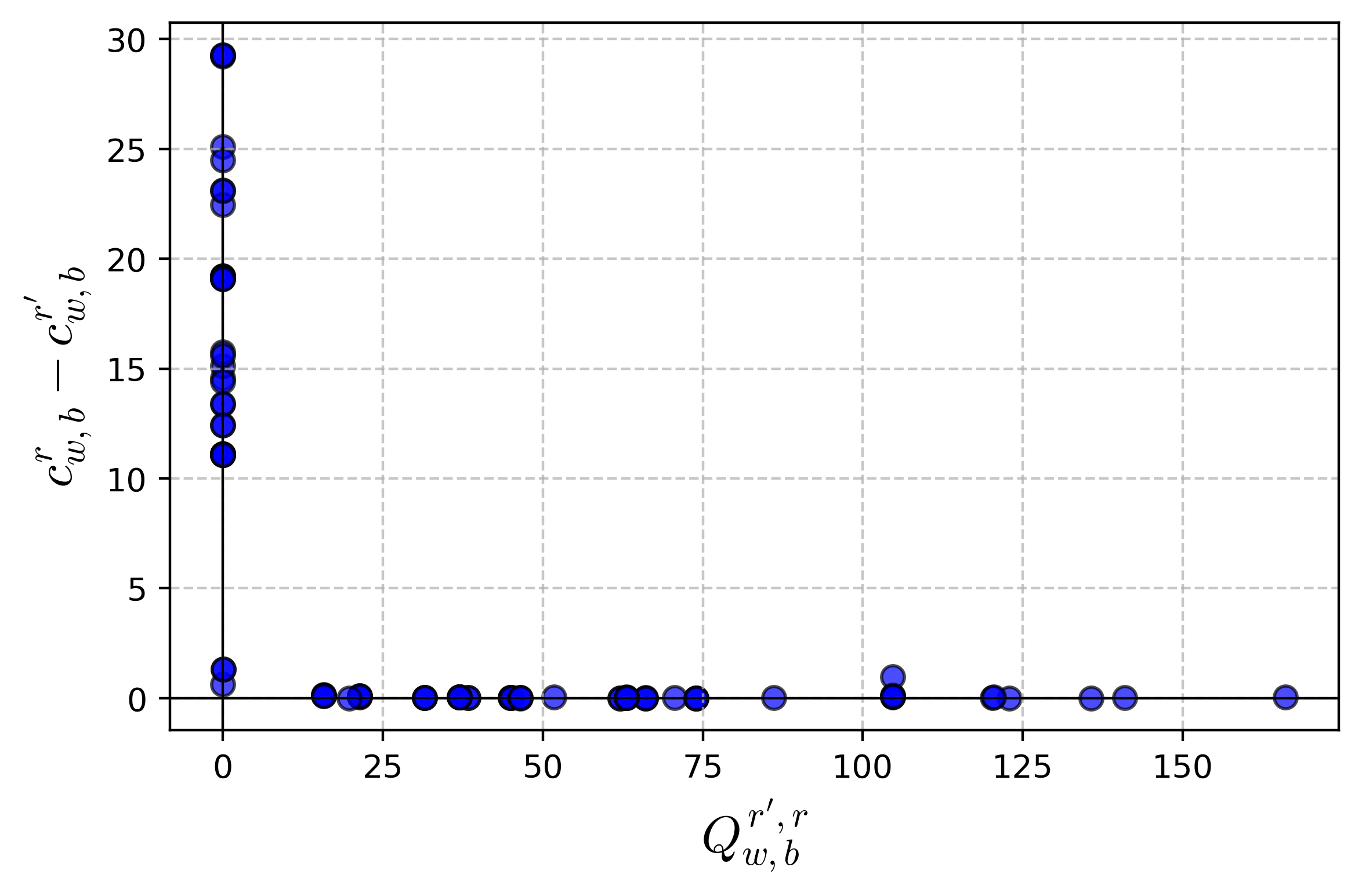}\label{fig:q and cost diff2}
	}
	\caption{Available capacity and cost difference at different convergence precisions.}
	\label{fig:q and cost diff}
\end{figure}

We also applied an MSA-based algorithm to solve the explicit priority model on this network. However, the algorithm failed to achieve an acceptable level of convergence: after one hour of computation, the minimum relative gap remained at~0.014. At this level of precision, the resulting route choices were still far from equilibrium. 
Table~\ref{tab:costdiff dynamic} reports, for each OD pair, the maximum cost difference among transit plan-departure time pairs with non-negligible flow (greater than~1.0). For example, in OD~(1,20), the maximum cost difference reached 56.71. Within this OD, two transit plan-departure time pairs carried flows of~87.50 and~11.94, with expected travel costs of~73.21 and~129.92, respectively. 
Similarly, OD~(2,13) exhibited a maximum cost difference of~11.57 between two transit plan-departure time pairs with flows of~30.23 and~4.17 and expected costs of~82.93 and~94.50. 

\begin{table}[htbp]
    \centering
    \caption{Maximum cost differences within OD pairs under the solution of the explicit priority model.}
    \footnotesize %
    \begin{tabular}{cccccccc}
        \toprule
        OD & Max.~Diff. & OD & Max.~Diff. & OD & Max.~Diff. & OD & Max.~Diff. \\
        \midrule
        (1,13) & 0.00 & (2,13) & 11.57 & (3,13) & 1.51 & (4,13) & 0.00 \\
        (1,20) & 56.71 & (2,20) & 4.40 & (3,20) & 0.00 & (4,20) & 0.00 \\
        (1,21) & 32.82 & (2,21) & 1.27 & (3,21) & 0.03 & (4,21) & 0.00 \\
        (1,24) & 4.07 & (2,24) & 0.00 & (3,24) & 10.73 & (4,24) & 0.00 \\
        \bottomrule
    \end{tabular}
    \label{tab:costdiff dynamic}
\end{table}

\section{Conclusion}
\label{sec:conclusion}

This paper revisits and extends a schedule-based transit assignment framework that captures passenger priority rules using an implicit method. We first reformulate the UEIP of \citet{nguyen_modeling_2001} as an NCP and, under mild conditions, prove the existence of an equilibrium state --- providing a formal foundation that had not been previously established. The NCP formulation also reveals that equilibrium may be non-unique and can include behaviorally unrealistic states. To address this, we refine the original framework by modifying the definition of route availability and the NCP formulation so that all such behaviorally inconsistent equilibria are excluded. The refined model remains structurally tractable: by imposing the priority-enforcing condition at the arc rather than the route level, it avoids explicit enumeration of all feasible routes and admits a decomposable network-flow structure. By smoothing one complementarity condition with the Fischer-Burmeister function, we obtain a continuously differentiable MPEC, for which we develop two solution approaches  ---  an implicit method and a nonlinear-programming–based method.

Using a real-world case study of student commuting at the University of Hong Kong, we first demonstrate that the model can realistically capture FCFS queuing at elevators as well as departure-early behavior among bus passengers driven by competition for boarding priority. A comparison with an explicit priority model highlights a key advantage of the proposed approach: by basing decisions on realized travel costs rather than expected ones, the refined framework avoids behaviorally implausible outcomes --- for example, passengers accepting substantial lateness penalties merely to preserve expected-cost consistency --- and instead yields equilibria that are behaviorally coherent. In addition, numerical experiments on a benchmark and the Sioux Falls transit networks show that the proposed algorithms can compute high-accuracy equilibria with reasonable computational effort.

This study also suggests several directions for future research. First, the proposed algorithm does not guarantee global optimality, and its scalability to very large networks is limited. Developing algorithms with provable convergence properties that can efficiently handle large-scale networks would be of substantial practical value for real-world applications of the model. Second, although the UEIP equilibrium concept is formally defined, further behavioral validation is needed to explain how it can be reached. In the road traffic assignment literature, extensive effort has been devoted to understanding the behavioral foundations and stability properties of user equilibrium (UE) \citep[e.g.,][]{smith1984stability,yang2005evolutionary,li2024wardrop} and stochastic user equilibrium (SUE) \citep[e.g.,][]{horowitz1984stability,cascetta1989stochastic,watling2003dynamics,cantarella2016modelling}. In a similar vein, a promising direction here is to investigate whether UEIP can be interpreted as the equilibrium outcome of a behaviorally plausible day-to-day adjustment process.
Moreover, recent studies have leveraged day-to-day dynamical models to examine which equilibrium, among multiple solutions, is more likely to emerge in reality \citep{li2024day}. Our findings on the non-uniqueness of UEIP, including the identification of solutions that appear behaviorally implausible, naturally raise the question of whether such outcomes can be ruled out from a dynamic perspective.
Finally, most existing transit network design studies rely on transit assignment models without priority rules when evaluating design alternatives \citep{yin2021timetable,xie_schedule-based_2021,feng5295135bilevel}, and thus overlook the effects of queuing and denied boarding induced by passenger prioritization. This is largely because incorporating DNL-based priority assignment into network design problems leads to prohibitive computational complexity. Our DNL-free framework may offer a more tractable way to integrate priority-driven passenger behavior into transit network design, and this integration is a promising direction for future research.

\bibliographystyle{apalike}
{\small
\bibliography{ly_STAP}
}

\appendix
\newpage
\section{An Explicit-Priority Transit Assignment Model and Its Issues}
\label{sec:appendix dynamic}
Existing explicit priority models assume that passengers choose line segments \citep{yao2017simulation} or strategies \citep{hamdouch_schedule-based_2008}, while a DNL procedure determines their actual boarding, transfer, and alighting times, as well as their realized travel experiences. In this appendix, we consider a setting in which passengers choose their starting time and a sequence of line segments. 

For an OD pair $w = (o, d) \in \gW := \gO \times \gD$, a \emph{transit plan} is any path in this transit network that starts at origin $o$, ends at destination $d$, and specifies the passenger's access stop, in-vehicle line segments (including any in-station transfers), and final alighting stop. We denote the set of such plans for OD pair $w$ by $\gP_w$. Each passenger chooses a transit plan $p \in \gP_w$ and a starting time $t \in \gT$. Let $g_{w,b}^{p,t}$ denote the number of class $b$ passengers of OD pair $w$ choosing $(p,t)$, and collect all such variables in the vector $\vg$. The feasible set $\gG$ consists of all $\vg$ satisfying
\begin{align}
    &\sum_{p \in \gP_w} \sum_{t \in \gT} g_{w, b}^{p, t} = d_{w, b}, \quad \forall w \in \gW, \ \forall b \in \gB_w, \notag\\
    &g_{w, b}^{p, t}\geq 0, \quad \forall w \in \gW, \ \forall b \in \gB_w, \ \forall (p, t) \in \gP_w \times \gT. \notag
\end{align}

We first present the DNL procedure, followed by a description of the explicit priority model and its corresponding MSA-based solution algorithm. Finally, we use a simple example to illustrate the behavioral issues inherent in this class of explicit priority models.

\subsection{Dynamic network loading}
\label{sec:dnl}
Formally, DNL can be viewed as a mapping $D: \gG \rightarrow \gF$ that transforms any flow vector of transit plans and starting times $\vg \in \gG$ into the resulting flow vector of spatio-temporal routes $\vf = D(\vg)$. Specifically, for every node $E \in \gE$ and every arc $A \in \gA$, let $x_E^{w,b,p,t}$ and $x_A^{w,b,p,t}$ denote the flow of OD pair $w \in \gW$ and class $b \in \gB_w$ associated with plan–departure pair $(p,t) \in \gP_w \times \gT$ passing through this node and this arc, respectively. The DNL procedure takes as input the given vector of transit plans and departure times $\vg \in \mG$, and simulates the corresponding node and arc flows,
\[
\bigl(x_E^{w,b,p,t}\bigr)_{\forall E \in \gE,\, \forall w \in \gW,\, \forall b \in \gB_w,\, \forall (p,t) \in \gP_w \times \gT}
\quad\text{and}\quad
\bigl(x_A^{w,b,p,t}\bigr)_{\forall A \in \gA,\, \forall w \in \gW,\, \forall b \in \gB_w,\, \forall (p,t) \in \gP_w \times \gT},
\]
subject to the boarding priority. Finally, for each OD pair $w\in\gW$, class $b\in\gB_w$, and $(p,t)\in\gP\times\gT$, let $\gR_{w,b}^{p,t}$ denote the set of all spatio-temporal routes whose starting time, boarding stop, line choices, transfer sequence, and final alighting stop coincide with those specified by $(p,t)$. The spatio-temporal route flows are obtained according to the flow-splitting proportions on each arc:
\begin{align}
    f_{w, b}^r = g_{w, b}^{p, t} \sum_{A\in \gA_w^r}\frac{x_A^{w,b,p,t}}{\sum_{A'\in \gA_{\text{tail}(A)}^+}x_{A'}^{w,b,p,t}} , \forall r\in \gR_{w, b}^{p, t}, \label{equ:dnl route flow}
\end{align}
where $\gA_w^r$ is the set of arcs belonging to route $r$; $\text{tail}(A)$ is the tail node of $A$; and $\gA_{E}^+$ and $\gA_{E}^-$ are the set of outgoing and incoming arcs of $E$, respectively.

Algorithm~\ref{alg:dnl} presents the DNL computation in detail. First, the flows of all plan–departure pairs $(p,t)$ for each OD pair $w \in \gW$ and class $b \in \gB_w$ are loaded at their origin node $E_{o}^{\text{ogn}}$:
\begin{align}
    x_{E_{o}^{\text{ogn}}}^{w,b,p,t} = g_{w, b}^{p, t}. \label{equ:dnl origin}
\end{align}

Subsequently, to preserve the temporal consistency of the transit system, the loading to all other nodes in the event–activity graph $\gH(\gE,\gA)$ is carried out in a combined topological and chronological order. There are three cases: (I) If a node $E$ represents a set-out event, then for each incoming arc $A \in \gA_{E}^-$, passengers at the origin node can be loaded onto this arc only if their chosen departure time $t$ coincides with the timestamp of $E$. Let $\text{timestamp}(E)$ denote the timestamp of $E$. The arc flow is then updated as
\begin{align}
    x_A^{w,b,p,t}
    =
    \begin{cases}
        x_{\text{tail}(A)}^{w,b,p,t},
        & \text{if } \text{timestamp}(E) = t, \\[4pt]
        0, 
        & \text{if } \text{timestamp}(E) \neq t,
    \end{cases}.
    \label{equ:dnl access arc flow}
\end{align}
(II) If a node $E$ represents a vehicle departure event, passengers arriving via the incoming arcs and boarding the subsequent in-vehicle riding arc $\text{riding}(E)$ are loaded according to the boarding priority until the vehicle capacity $u_{\text{riding}(E)}$ is reached. Accordingly, we traverse all incoming arcs of node $E$ in order of decreasing priority and compute the corresponding arc flows. 
For each incoming arc $A \in \gA_{E}^-$, let $\bar{\gP}_w^A$ denote the set of transit plans of OD pair $w \in \gW$ that pass through the stop corresponding to $A$ and choose the line associated with $E = \text{head}(A)$. Let $\bar{x}_E^{w,b,p,t}$ and $\bar{u}_{\text{riding}(E)}$ denote, respectively, the remaining passenger flow and residual capacity after loading all higher-priority incoming arcs. The total number of passengers at $\text{tail}(A)$ who wish to board at event $E$ under capacity-free conditions is
\[
x_A^{\text{wishing}}
=
\sum_{w\in \gW}\sum_{b\in \gB_w}\sum_{p\in \bar{\gP}_w^A}\sum_{t\in \gT}
\bar{x}_{\text{tail}(A)}^{w,b,p,t}.
\]
When capacity is binding and not all of these passengers can board, we assume they are selected with equal probability. Consequently, the arc flows are updated as
\begin{align}
    x_A^{w,b,p,t}
    =
    \begin{cases}
        \bar{x}_{\text{tail}(A)}^{w,b,p,t},
        & \text{if } \bar{u}_{\text{riding}(E)} \geq x_A^{\text{wishing}}, \\[4pt]
        \displaystyle
        \bar{x}_{\text{tail}(A)}^{w,b,p,t}\,\frac{\bar{u}_{\text{riding}(E)}}{x_A^{\text{wishing}}},
        & \text{if } \bar{u}_{\text{riding}(E)} < x_A^{\text{wishing}},
    \end{cases}.
    \label{equ:dnl boarding arc flow}
\end{align}
(III) If node $E$ corresponds to any other type of event, no priority rules or capacity constraints apply. For any incoming arc $A \in \gA_{E}^-$, all flow at the tail node $\text{tail}(A)$ is directly loaded onto the arc. Thus, the arc flows are updated as
\begin{align}
    x_A^{w,b,p,t} = x_{\text{tail}(A)}^{w,b,p,t}.
    \label{equ:dnl other arc flow}
\end{align}
In all cases, the flow at node $E$ is obtained by aggregating the flows on all incoming arcs:
\begin{align}
    x_E^{w,b,p,t} = \sum_{A \in \gA_{E}^-} x_A^{w,b,p,t}.
    \label{equ:dnl node flow}
\end{align}

\begin{breakablealgorithm}
\caption{Dynamic network loading}
\label{alg:dnl}
\small
\begin{algorithmic}[1]
    \State \textbf{Input:} The node set $\gE^*$ sorted in topological and chronological order. The incoming arc set $\gA_{E}^{-*}$ is sorted in priority order for all departure node $E$. The flow vector $\vg$.
    \State Load the passenger flows on the virtual origin nodes according to Equation \eqref{equ:dnl origin}.
    \For{each node $E\in \gE^*$ and $E\notin O$}
    \If{$E$ is a set-out node}
        \For{each arc $A \in \gA_{E}^-$}
        \State Update $x_A^{w,b,p,t}$ according to Equations \eqref{equ:dnl access arc flow}.
        \EndFor
    \Else
        \If{$E$ is a vehicle departure node}
        \State Set $\bar{u}_{\text{riding}(E)} = u_{\text{riding}(E)}$.
            \For{each arc $A \in \gA_{E}^{-*}$}
            \State Update $x_A^{w,b,p,t}$ according to Equations \eqref{equ:dnl boarding arc flow}.
            \State Set $\bar{u}_{\text{riding}(E)} = \bar{u}_{\text{riding}(E)} - \sum_{w\in \gW}\sum_{b\in \gB_w}\sum_{p\in \bar{\gP}_w^A}\sum_{t\in \gT} x_A^{w,b,p,t}$, and $\bar{x}_E^{w,b,p,t} = \bar{x}_E^{w,b,p,t} - x_A^{w,b,p,t}$.
            \EndFor
        \Else
            \For{each arc $A \in \gA_{E}^-$}
            \State Update $x_A^{w,b,p,t}$ according to Equations \eqref{equ:dnl other arc flow}.
            \EndFor
        \EndIf
    \EndIf
    \State Update $x_E^{w,b,p,t}$ according to Equations \eqref{equ:dnl node flow}. Set $\bar{x}_E^{w,b,p,t} = x_E^{w,b,p,t}$.
    \EndFor
    \State Obtain the spatio-temporal route flow according to Equations \eqref{equ:dnl route flow}.
\end{algorithmic}
\end{breakablealgorithm}

\subsection{Variational inequality model and the MSA algorithm}
\label{sec:explicit vi model}
A key implication of DNL under capacity constraints is that passengers choosing the same transit plan and starting time $(p,t)$ may experience heterogeneous realized experiences.  Consequently, $(p,t)$ does not correspond to a single realized spatio-temporal route, and its travel cost is not uniquely defined for all passengers adopting that strategy. Explicit-priority models therefore evaluate $(p,t)$ using an \emph{expected cost}, defined as the flow‑weighted average of the realized costs over all spatio–temporal routes induced by DNL:
\begin{align}
e_{w,b}^{p,t}(\vg)=
\left\{
\begin{array}{rl}
\displaystyle \sum_{r\in \gR_{w,b}^{p,t}} c_{w,b}^r(\vf)\,\frac{f_{w,b}^r}{g_{w,b}^{p,t}},
& \quad \text{if } g_{w,b}^{p,t}>0,\\[8pt]
\displaystyle c_{w,b}^{\bar r}(\vf),
& \quad \text{if } g_{w,b}^{p,t}=0,
\end{array}
\right.
\label{eq:explicit_expected_cost}
\end{align}
where $\vf=D(\vg)$. Here, $\bar r\in \gR_{w,b}^{p,t}$ is the spatio-temporal route obtained by injecting an infinitesimal flow of class $b$ passengers adopting $(p,t)$ on top of the background loading $\vg$ and, at each boarding, dwelling, or transfer opportunity, assigning this infinitesimal flow to the first-arriving vehicle of the chosen line that has residual capacity. This convention ensures that $e_{w,b}^{p,t}(\vg)$ is well-defined even when $(p,t)$ is unused.

\begin{definition}[User equilibrium with explicit priority]
A feasible flow vector $\vg\in\gG$ is a \emph{user equilibrium with explicit priority} (UEEP) if, for every OD pair $w\in\gW$, class $b\in\gB_w$, and strategy $(p,t)\in\gP\times\gT$, no passenger can reduce their expected cost $e_{w,b}^{p,t}(\vg)$ by unilaterally switching to another transit plan or starting time.
\end{definition}

The definition of UEEP is equivalent to the following variational inequality problem:
\begin{align}
\label{prob:dynamic model}
    \ve(\vg^*)^\top (\vg - \vg^*) \;\geq\; 0, 
    \quad \forall\, \vg \in \mG,
\end{align}
where $\vg^*$ denotes the equilibrium flow vector and $\ve(\cdot)$ is the corresponding expected cost mapping.

Since there is no closed-form relationship between the flow vector $\vg$ and the expected cost vector $\ve$, the mapping $\ve(\vg)$ is, in general, not differentiable. Therefore, first-order algorithms are typically employed to solve problem~\eqref{prob:dynamic model}, among which the method of successive averages (MSA) is the most commonly used. The detailed steps of the algorithm are given in Algorithm~\ref{alg:msa}.

\begin{breakablealgorithm}
\caption{Method of Successive Averages}
\label{alg:msa}
\small
\begin{algorithmic}[1]
    \State \textbf{Step 0.} Initialize with a feasible solution $\vg^1 \in \mG$ obtained from the all-or-nothing assignment. Set $k=1$.
    \State \textbf{Step 1.} Execute Algorithm~\ref{alg:dnl} to obtain the arc flow $\vx$ and route flow $\vf$. Update the expected cost $\ve(\vg^k)$.
    \State \textbf{Step 2.} For each OD pair $w\in \gW$ and each class $b\in B_w$, identify the minimum-cost pair $(\bar{p},\bar{t})$ and set $y_{w, b}^{\bar{p},\bar{t}} = d_{w, b}$ and $y_{w, b}^{p, t} = 0$ for all $p\neq \bar{p}, t\neq \bar{t}$. Denote $\vy = (y_{w, b}^{p, t})_{\forall w \in \gW, \forall b \in \gB_w, \forall (p,t) \in \gP_w \times \gT}$.
    \State \textbf{Step 3.} Compute the relative gap:
    \begin{align}
        RG = 1 - \frac{\sum_{w\in \gW}\sum_{b\in \gB_w} e_{w, b}^{\bar{p},\bar{t}}  d_{w, b}}{\sum_{w\in \gW}\sum_{b\in \gB_w}\sum_{p\in \gP_w}\sum_{t\in \gT} e_{w, b}^{p, t}  g_{w, b}^{p, t}}. \notag
    \end{align}
    If $RG < \epsilon$, \textbf{stop}; otherwise, update $\vg^{k+1} = \frac{k \vg^k + \vy}{k+1}$, set $k = k+1$, and return to Step~1.
\end{algorithmic}
\end{breakablealgorithm}

\subsection{Illustrative example and behavioral issues}
\label{sec:explicit model issue}
This class of explicit-priority models cannot fully accommodate the inherently discrete nature of transit service in time, which leads to behaviorally questionable outcomes. To illustrate this point, we revisit the example used in Section~\ref{sec:implicit} for the implicit priority model. Table~\ref{tab:UEEP} and Figure~\ref{fig:explict} present a user equilibrium outcome under the explicit priority model~\eqref{prob:dynamic model}.

\begin{table}[htbp]
  \centering
  \caption{User equilibrium outcome of the explicit priority model on the example transit network.}
  \footnotesize %
    \begin{tabular}{ccllr}
    \toprule
    OD & Starting time & Path description & Flow & Expected cost \\
    \midrule
    $(o_1,d)$ & \textbf{7:24}  & \textbf{Line 1} & \textbf{0.5}   & \textbf{56} \\
    $(o_1,d)$ & \textbf{7:24}  & \textbf{Line 1 - Line 2} & \textbf{1.5}   & $\bm{46\times 1/1.5 + 76\times0.5/1.5=56}$ \\
    $(o_2,d)$ & 7:49  & Line 2 & 2     & 21 \\
    $(o_3,d)$ & 7:53  & Line 1 & 0     & 27 \\
    $(o_3,d)$ & 7:53  & Line 2 & 2     & 17 \\
    \bottomrule
    \end{tabular}%
  \label{tab:UEEP}%
\end{table}%

\begin{figure}[htbp]
    \centering
    \includegraphics[width=0.7\textwidth]{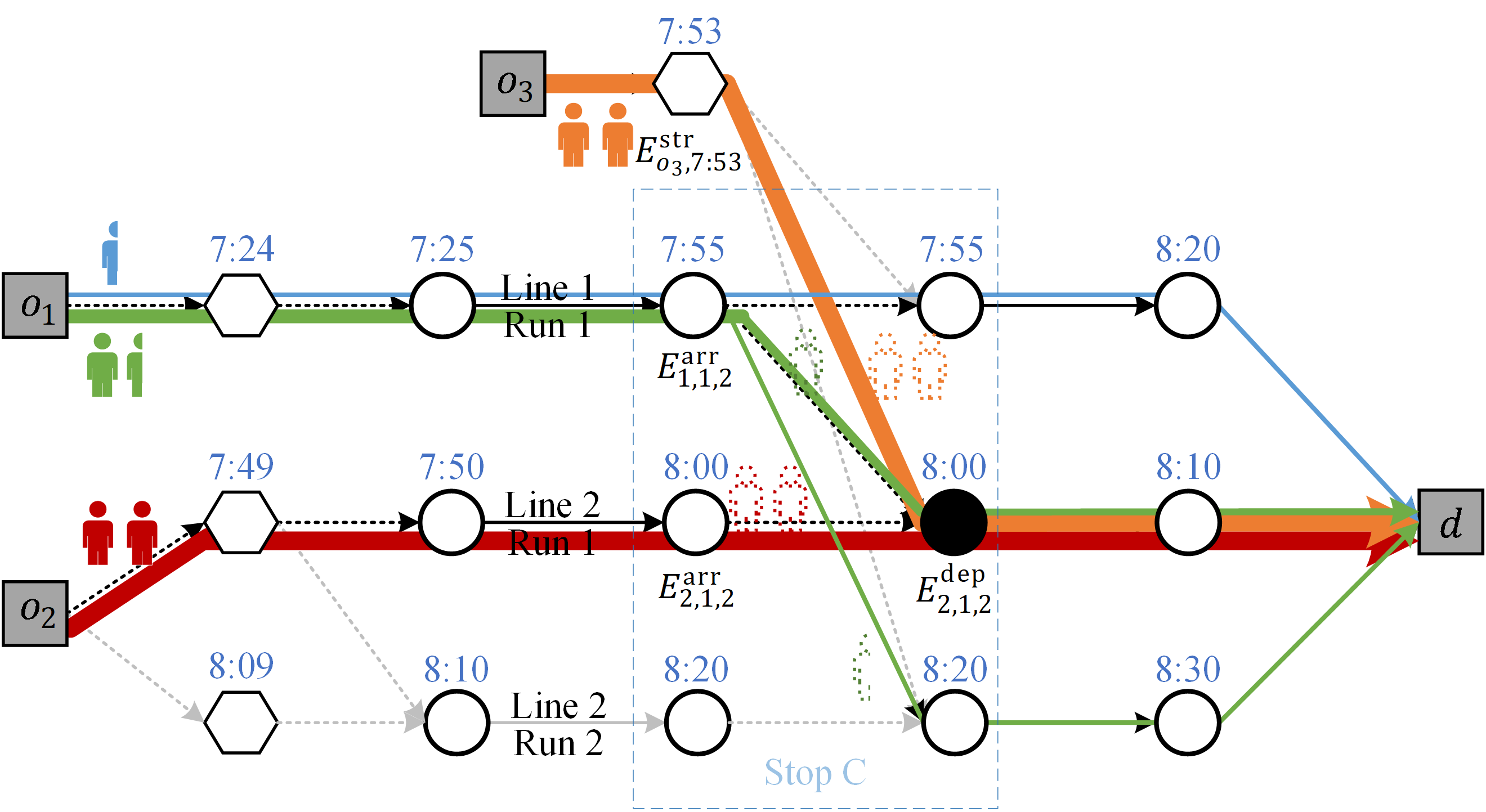}
    \caption{Visualization of user equilibrium flows for the explicit priority model on the example transit network.}
    \label{fig:explict}
\end{figure}

Among the passengers of OD pair $(o_1,d)$ who choose to take line 1 and transfer to line 2, vehicle capacity and priority rules imply that only one passenger can successfully board line 2 run 1 at stop~C (with a cost of 46), while the remaining 0.5 units of passenger flow must wait for line~2 run~2 (with a cost of 76). The expected cost for passengers choosing this option equals the cost of choosing "line~1 without transfer", both being 56, and the solution is therefore regarded as an equilibrium. In reality, however, the passengers who are left behind and forced to wait for line~2 run~2, thus incurring the higher realized cost, would be inclined to change their future decisions, for instance by choosing not to transfer.

In dynamic traffic assignment on road networks \citep{friesz1993variational}, where the service provided by each link is typically modeled as continuous over time, explicit priority frameworks can partially alleviate the discrepancy between expected and realized costs by reducing the time-step size and thereby shrinking the size of each passenger group. In public transit, by contrast, service is inherently discontinuous in time, and existing explicit priority frameworks have not been modified to account for this discontinuity. As a result, they cannot bridge the gap between expected costs in the equilibrium definition and the actual travel experience of individual passengers.

\section{Proof of Proposition \ref{prop:model2 solution exists}}
\label{sec:appendix-proof-solution-exists}
Given that the mapping of NCP~\eqref{prob:cp} 
$$
\bar{H}(\vf, \vmu, \vv)=\left(
    \begin{array}{c}
    \vc(\vf)+\bar{\Delta}\vv-\Lambda\vmu \\
    \Lambda^T\vf-\vd \\
    \vq(\vx) \\
    \end{array}
    \right)
$$
is continuous, while the feasible set $\bar{\Omega}=\{\vf\geq 0, \vmu\geq 0, \vv\geq 0\}$ is unbounded and therefore not compact, the proof of Proposition \ref{prop:model2 solution exists} follows the same structure as that of Proposition \ref{prop:route solution exists}. We introduce upper-bound constraints for each variable, yielding a modified NCP whose feasible set is compact. We then show that none of these upper bounds are binding at the solution, implying that any solution of the modified NCP also satisfies the original NCP~\eqref{prob:cp}.
    
Choose two scalars $e_1$ and $e_2$ such that
\begin{align}
    e_1 > \max\{d_{w, b}: w\in \gW, b\in \gB_w\} \text{ and } e_2 > \max\{c_{w, b}^r: w\in \gW, b\in \gB_w, r\in \gR_w\}. \notag
\end{align}
Let $\bar{\Omega}'=\bar{\Omega}\cap\{\vf\leq e_1\bm{1}, \vmu\leq e_2\bm{1}, \vv\leq e_2\bm{1}\}$. Since $\bar{\Omega}'$ is compact, the following NCP must admit a solution: 
\begin{subequations}
    \begin{align}
        &0\leq \vf \bot \vc(\vf)+\bar{\Delta}\vv-\Lambda\vmu + \bm{\kappa}\geq 0, \label{equ:exist2 cc1} \\
        &0\leq \vmu \bot \Lambda^T\vf-\vd +\bm{\rho} \geq 0, \label{equ:exist2 cc2} \\
        &0\leq \vv \bot \vq(\vx) +\bm{\theta} \geq 0, \label{equ:exist2 cc3} \\
        &0\leq \bm{\kappa} \bot e_1\bm{1}-\vf \geq 0, \label{equ:exist2 cc4}\\
        &0\leq \bm{\rho} \bot e_2\bm{1}-\vmu \geq 0, \label{equ:exist2 cc5}\\
        &0\leq \bm{\theta} \bot e_2\bm{1}-\vv \geq 0. \label{equ:exist2 cc6}
    \end{align}
\end{subequations}
Let $(\vf^*, \vmu^*, \vv^*)$ denote this solution. If we demonstrate that $\bm{\kappa}=\bm{\rho}=\bm{\theta}=\bm{0}$, then $(\vf^*, \vmu^*, \vv^*)$ must also be a solution of NCP~\eqref{prob:cp-route}.

\textbf{For $\bm{\kappa}$}: Suppose $\kappa_{w, b}^r>0$ for some $w\in \gW, b\in \gB_w, r\in \gR_w$. By the complementarity condition \eqref{equ:exist2 cc4}, we have $f_{w, b}^{r*}=e_1>d_{w, b}$, so $\sum_{r\in \gR_w}f_{w, b}^{r*}-d_{w, b}+\rho_{w, b}>0$. Then, according to Coondition \eqref{equ:exist2 cc2}, we have $\mu_{w, b}^*=0$. Moreover, since $f_{w, b}^{r*}=e_1>0$, we have 
\begin{align}
    0=c_{w, b}^r(\vf^*)+\sum_{A\in \gA_{w,b,r}^{\text{priority}}}v_A^*-\mu_{w, b}^*+\kappa_{w, b}^r=c_{w, b}^r(\vf^*)+\sum_{A\in \gA_{w,b,r}^{\text{priority}}}v_A^*+\kappa_{w, b}^r >0, \notag
\end{align}
because $c_{w, b}^r(\vf^*)$ and $\kappa_{w, b}^r$ are positive, and $\sum_{A\in \gA_{w,b,r}^{\text{priority}}}v_A^*$ is non-negative. This contradiction yields $\bm{\kappa}= \vzero$.

\textbf{For $\bm{\rho}$}: For each OD pair $w\in \gW$ and class $b\in \gB_w$, since there exists a route $r\in \gR_w$ with $Q_{w, b}^r(\vf^*)>0$ by assumption, each arc belonging to $\gA_{w,b,r}^{\text{priority}}$ satisfies $q_A(\vx^*)>0$ and $q_A(\vx^*)+\theta_A>0$. By Condition \eqref{equ:exist2 cc3}, we have $v_A^*=0$ for all $A\in \gA_{w,b,r}^{\text{priority}}$. Therefore, 
\begin{align}
    &0\leq c_{w, b}^r(\vf^*)+\sum_{A\in \gA_{w,b,r}^{\text{priority}}}v_A^*-\mu_{w, b}^*+\kappa_{w, b}^r = c_{w, b}^r(\vf^*) - \mu_{w, b}^* \notag \\
    &\Rightarrow \mu_{w, b}^*\leq c_{w, b}^r(\vf^*)< e_2. \label{equ:exist2 mu}
\end{align}
This means that $e_2 - \mu_{w, b}^*>0$ for all $w\in \gW$ and  $b\in \gB_w$, and $\bm{\rho}= \vzero$ by Equation \eqref{equ:exist2 cc5}.

\textbf{For $\bm{\theta}$}: For any departure event, let $A_1$ be the arc that ends at this event with the highest loading priority. If $\theta_{A_1}>0$, the complementarity condition  \eqref{equ:exist2 cc6} suggests $v_{A_1}^*=e_2$. Meanwhile, Equation \eqref{equ:exist2 mu} implies that $\mu_{w, b}^*<e_2=v_{A_1}^*$ for all $w\in \gW$ and  $b\in \gB_w$. Hence, for each route $r$ such that $A_1\in \gA_{w,b,r}^{\text{priority}}$, we have 
\begin{align}
    c_{w, b}^r(\vf^*)+\sum_{A\in \gA_{w,b,r}^{\text{priority}}}v_A^*-\mu_{w, b}^*+\kappa_{w, b}^r = c_{w, b}^r(\vf^*)+\sum_{A\in \gA_{w,b,r}^{\text{priority}}}v_A^*-\mu_{w, b}^* > v_{A_1}^* - \mu_{w, b}^* >0. \notag
\end{align}
By Condition \eqref{equ:exist2 cc1}, all of these routes satisfy $f_{w, b}^{r*}=0$. This means that the flow of arc $A_1$ is zero, namely $x_{A_1}^*=\sum_{w\in \gW}\sum_{b\in \gB_w}\sum_{r\in \gR_w}f_{w, b}^{r*}\delta_{w, b}^{r,A_1} =0$. Then, we have the available capacity $q_{A_1}(\vx^*)=u_{\text{riding}(A_1)}-x_{A_1}^*=u_{\text{riding}(A_1)}>0$, and $q_{A_1}(\vx^*)+\theta_{A_1}>0$, and consequently, $v_{A_1}^*=0$ by Condition \eqref{equ:exist2 cc3}. This contradicts $v_{A_1}^*=e_2>0$, and hence $\theta_{A_1}=0$.

Next, let $A_2$ denote the arc ending at this event with the second-highest priority. If $\theta_{A_2}>0$, we can get $v_{A_2}^*=e_2$ and $x_{A_2}^*=0$ using a similar proof. This lead to the available capacity of $A_2$ is equal to that of $A_1$, i.e., $q_{A_2}(\vx^*)=u_{\text{riding}(A_2)}-x_{A_1}^*-x_{A_2}^*=q_{A_1}$. In addition, the result $\theta_{A_1}=0$ above, together with Condition \eqref{equ:exist2 cc3}, indicates that $q_{A_1}\geq 0$. Hence, $q_{A_2}\geq 0$ and $q_{A_2}+\theta_{A_2}>0$. According to Condition \eqref{equ:exist2 cc3}, we encounter the contradiction that $0=v_{A_2}^*=e_2>0$. Therefore, $\theta_{A_2}=0$. 

Proceeding in this manner for all arcs ending at the departure event, and repeating the process for every departure event in the network, we obtain $\bm{\theta}= \vzero$.

\section{Algorithms for MPEC~\eqref{prob:reformulate}}
\label{sec:appendix algorithm}
This section introduces the implicit method, the nonlinear-programming-based method, and the method for obtaining an initial solution for solving MPEC~\eqref{prob:reformulate}.

\subsection{Implicit method}
\label{sec:implicit-algorithm}
Considering the implicit function $\vx(\vv)$, MPEC \eqref{prob:reformulate} is simplified to a problem involving only the upper-level variable $\vv$, which is $\min_{\vv\geq 0} \Psi(\vv,\vx(\vv))$. We develop a projected gradient method combined with an Armijo line search to solve it. The entire algorithm can be summarized as follows.
\begin{algorithm}
\caption{Implicit method}
\label{alg:IM}
\small
\begin{algorithmic}[1]
    \State \textbf{Input:} A initial solution $(\vv^0,\vx^0,\vmu^0)$. Set $k=0$.
    \State \textbf{Step 1.} Calculate the gradient $\nabla \Psi(\vv^k)$.
    \State \textbf{Step 2.} For each $i=0, 1, 2, \cdots$: Let $\vv'=[\vv^k - 2^{-i}\nabla \Psi(\vv^k)]_+$.
    \State \quad \textbf{Step 2.1.} Obtain a solution $\vx'$ of Problem~\eqref{prob:reformulate-lower}, and then compute the available capacity $\vq(\vx')$ and the merit function $\Psi(\vv')$.
    \State \quad \textbf{Step 2.2.} If $\Psi(\vv')\leq \Psi(\vv^k) - \gamma 2^{-i} \nabla \Psi(\vv^k)^T \nabla \Psi(\vv^k)$, go to Step 3.
    \State \textbf{Step 3.} Set $\vv^{k+1}=\vv'$. If $\Psi(\vv^{k+1})<\epsilon_1$, then stop. Otherwise, $k=k+1$, and go to Step 1.
\end{algorithmic}
\end{algorithm}

In Step 1, the gradient of the merit function $\Psi(\vv,\vx(\vv))$ regarding $v_A$ is
\begin{align}
    &\frac{\partial \Psi(\vv,\vx(\vv))}{\partial v_A}=\sum_{A'\in \gA^{\text{priority}}}2\varphi(v_{A'}, q_{A'})\frac{\partial \varphi(v_{A'}, q_{A'})}{\partial v_A}, \notag\\
    &\frac{\partial \varphi(v_{A'}, q_{A'})}{\partial v_A}=\left\{
    \begin{array}{rl}
    (v_{A'}+q_{A'}\frac{\partial q_{A'}}{\partial v_A})\frac{1}{\sqrt{v_{A'}^2+q_{A'}^2}}-1-\frac{\partial q_{A'}}{\partial v_A}, & \text{if } A'= A \notag\\
    q_{A'}\frac{\partial q_{A'}}{\partial v_A}\frac{1}{\sqrt{v_{A'}^2+q_{A'}^2}}-\frac{\partial q_{A'}}{\partial v_A}, & \text{otherwise}
    \end{array}, \right. \notag\\
    &\frac{\partial q_{A'}}{\partial v_A}=-\sum_{A''\in \text{Prior}(A')}\frac{\partial x_{A''}}{\partial v_A}. \notag
\end{align}
Since the mapping from $\vv$ to $\vx$ is the NCP \eqref{prob:reformulate-lower}, the derivative $\frac{\partial x_{A'}}{\partial v_A}$ does not have a closed-form formulation. To obtain this derivative, we employ the sensitivity analysis method \citep{tobin1988sensitivity,patriksson2004sensitivity}.

In Step 2, $[\cdot]_+$ denotes the projection onto the non-negative space. The lower-level problem \eqref{prob:reformulate-lower} is solved using the iGP algorithm, an efficient static traffic assignment problem solver proposed by \citet{xie_greedy_2018}.

\subsection{Nonlinear-programming-based method}
\label{sec:nonlinear-algorithm}
We rewrite the equilibrium constraints \eqref{prob:reformulate-lower} as a set of nonlinear inequalities, which is 
\begin{align}
    J(\vv,\vx,\vmu)=\left\{
        \begin{array}{c}
        \vf \\
        \hat{\vc}(\vx,\vv)-\Lambda\vmu \\
        -\vf(\hat{\vc}(\vx,\vv)-\Lambda\vmu) \\
        \vmu \\
        \Lambda^T\vf-\vd \\
        -\vmu(\Lambda^T\vf-\vd) \\
        \vv
        \end{array}
        \right\} \geq 0. \notag
\end{align}
This leads to the transformation of MPEC \eqref{prob:reformulate} into the following nonlinear programming problem
\begin{subequations}
    \begin{align}
        \min_{\vv,\vx,\vmu} \Psi(\vv,\vx,\vmu) 
        \text{ subject to } J(\vv,\vx,\vmu)\geq 0. \notag
    \end{align}
\end{subequations}
This nonlinear programming problem can be solved using various existing nonlinear optimization algorithms, such as Sequential Quadratic Programming (SQP). 

\subsection{Initial solution}
\label{sec:initial-solution}
To initialize the MPEC~\eqref{prob:reformulate}, we first obtain the initial $\vx^0$ and $\vmu^0$ by solving a relaxed transit assignment in which the available capacity constraints $q_A(\vx)\geq 0, \forall A\in \gA^{\text{priority}}$ are replaced by standard vehicle capacity constraints $u_A\geq x_A, \forall A\in \gA^{\text{riding}}$. Let $\vu=\{u_A: A\in \gA^{\text{riding}}\}$ denote the capacity vector of riding arcs. The relaxed complementarity formulation is
\begin{subequations}
\begin{align}
    &0\le \vf \bot \vc(\vf)+\hat{\Delta}\bm{\nu}-\Lambda\vmu \ge 0, \notag\\
    &0\le \vmu \bot \Lambda^T\vf-\vd \ge 0, \notag\\
    &0\le \bm{\nu} \bot \vu(\vx) \ge 0,  \notag
\end{align}
\end{subequations}
where $\bm{\nu}=\{\nu_A: A\in \gA^{\text{riding}}\}$ are the multipliers for vehicle capacity constraints and $\hat{\Delta}=[\delta_{w, b}^{r,A}]_{|n_f|\times|\gA^{\text{riding}}|}$ is the route–arc incidence matrix. This model can be regarded as an extension of the classical capacitated traffic assignment problem~\citep{larsson_augmented_1995} to the transit context (with a different definition of the travel cost function $\vc$), and can be efficiently solved by existing algorithms such as ALM-GP or ALM-Greedy~\citep{nie_models_2004,feng_efficient_2020}.

Once $\vx^0$ is obtained, we initialize the upper-level vector $\vv^0$ from the multipliers $\bm{\nu}$ of the relaxed problem. For each arc $A\in \gA^{\text{priority}}$, we set
\begin{align}
    v_A^0=\left\{
    \begin{array}{rl}
    \nu_{\text{riding}(A)}, & q_A\leq 0 \\
    0, & q_A> 0
    \end{array}. \right. \notag
\end{align}
This initialization assigns positive values to $\vv^0$ only on unavailable arcs, thereby guiding the subsequent iterations toward satisfying the available capacity constraints~$\vq(\vx)\geq 0$. This initialization method is used as the default unless otherwise stated.

\end{document}